%% file: main.tex
\documentclass[cleveref, english]{lmcs} %%% last changed 2014-08-20
\pdfoutput=1

\usepackage{lastpage}
\lmcsdoi{21}{1}{28}
\lmcsheading{}{\pageref{LastPage}}{}{}%
{Mar.~06,~2024}{Mar.~24,~2025}{}


\usepackage{hyperref}
\usepackage[T1]{fontenc}
\usepackage[utf8]{inputenc}
\usepackage{amsmath,amsfonts,amssymb,amsthm}
\usepackage{graphicx}
\usepackage{todonotes}
\hypersetup{hidelinks,breaklinks}

\usepackage{diagbox}
\usepackage{array} %to center content in tabular environments
\usepackage{arydshln} %dashed lines
\usepackage{makecell}

\usepackage{mathcommand}
\usepackage[notion, quotation, electronic]{knowledge}

\begin{document}

\title{Characterising memory in infinite games}
%\thanks{thanks, optional.}	%optional

\titlecomment{This paper is an extended version of~\cite{CO23Memory}, including full proofs and additional examples.}
\thanks{Antonio Casares is supported by the Polish National Science Centre (NCN) grant “Polynomial finite state computation” (2022/46/A/ST6/00072).}	%optional

% affiliations are numbered automatically with a, b, c (see below)
% use the optional argument to indicate the affiliation(s) of each author
% omit the argument if there is only one author, or only one affiliation
\author[A.~Casares]{Antonio Casares\lmcsorcid{0000-0002-6539-2020}}[a]
\author[P.~Ohlmann]{Pierre Ohlmann\lmcsorcid{0000-0002-4685-5253}}[b]

% affiliation 1 (automatically numbered a)
\address{University of Warsaw, Poland}
%optional
% write emails for all authors having that affiliation
\email{antoniocasares@mimuw.edu.pl}  %optional

% affiliation 2 (automatically numbered b)
\address{CNRS, Laboratoire d'Informatique et des Systèmes, Marseille, France}	%optional
%\address{University of Warsaw, Poland}
\email{pierre.ohlmann@lis-lab.fr}  %optional

%% etc.

%% required for running head on odd and even pages, use suitable
%% abbreviations in case of long titles and many authors:

%%%%%%%%%%%%%%%%%%%%%%%%%%%%%%%%%%%%%%%%%%%%%%%%%%%%%%%%%%%%%%%%%%%%%%%%%%%

%% the abstract has to PRECEDE the command \maketitle:
%% be sure not to issue the \maketitle command twice!

\input{macros}
\input{knowledge}

\begin{abstract}
  \noindent 
  This paper is concerned with games of infinite duration played over potentially infinite graphs.
  Recently, Ohlmann (TheoretiCS 2023) presented a characterisation of objectives admitting optimal positional strategies, by means of universal graphs: an objective is positional if and only if it admits well-ordered monotone universal graphs.
  We extend Ohlmann's characterisation to encompass (finite or infinite) memory upper bounds.
  
  We prove that objectives admitting optimal strategies with $\varepsilon$-memory less than $m$ (a memory that cannot be updated when reading an $\varepsilon$-edge) are exactly those which admit well-founded monotone universal graphs whose antichains have size bounded by $m$.
  We also give a characterisation of chromatic memory by means of appropriate universal structures.
  Our results apply to finite as well as infinite memory bounds (for instance, to objectives with finite but unbounded memory, or with countable memory strategies).
  
  We illustrate the applicability of our framework by carrying out a few case studies, we provide examples witnessing limitations of our approach, and we discuss general closure properties which follow from our results.
\end{abstract}

\maketitle

\noindent
This document contains hyperlinks.
Each occurrence of a notion is linked to its definition.
On an electronic device, the reader can click on words or symbols (or just hover over them on some PDF readers) to see their definition.

%% start the paper here:

\section{Introduction}
\label{sec:introduction}

\subsection{Context}

%\paragraph{Games and strategy complexity.}
We study zero-sum turn-based "games" on graphs, in which two players, that we call "Eve" and "Adam", take turns in moving a token along the edges of a given (potentially infinite) edge-coloured directed graph.
Vertices of the graph are partitioned into those belonging to Eve and those belonging to Adam.
When the token lands in a vertex owned by player X, it is this player who chooses where to move next.
This interaction, which is sometimes called a play, goes on in a non-terminating mode, producing an infinite sequence of colours.
We fix in advance an "objective" $W$, which is a language of infinite sequences of colours; plays producing a sequence of colours in $W$ are considered to be winning for Eve, and plays that do not satisfy the objective $W$ are winning for the opponent Adam.

In order to achieve their goal, players use "strategies", which are representations of the course of all possible plays together with instructions on how to act in each scenario.
In this work, we are interested in optimal "strategies" for "Eve", that is, strategies that guarantee a "victory" whenever this is possible.
More precisely, we are interested in the complexity of such "strategies", or in other words, in the succinctness of the representation of the space of plays.
The simplest "strategies" are those that assign in advance an outgoing edge to each vertex owned by Eve, and always play along this edge, disregarding all the other features of the play.
All the information required to implement such a strategy appears in the game graph itself.
These "strategies" are called "positional" (or memoryless).
However, in some scenarios, playing optimally requires distinguishing different plays that end in the same vertex; one should remember other features of plays.
An example of such a "game" is given in Figure~\ref{fig:simple_example_intro}.

\begin{figure}[h]
	\begin{center}
		\includegraphics[width=0.75\linewidth]{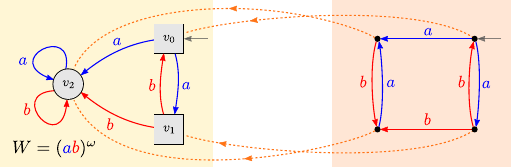}
	\end{center}
	\caption{On the left, a "game" with "objective" $W = (ab)^\omega$; in words, "Eve"  should ensure that the play alternates between $a$-edges and $b$-edges. We represent Eve's vertices as circles and Adam's as squares.
	On the right, a winning "strategy" for "Eve" which uses one state of memory for $v_0$, one state of memory for $v_1$, and two states of memory for $v_2$.
	Note that two "states of memory" for $v_2$ are required here: a "positional" "strategy" would always follow the same self-loop and therefore cannot "win".
	One can prove that any "game" with "objective" $W$ which is won by "Eve" can be won even when restricting to "strategies" with two "states of memory", such as the one above.
	To conclude, the memory requirements for $W$ is exactly two.
	}\label{fig:simple_example_intro}
\end{figure}

Given an "objective" $W$, the question we are interested in is:

\begin{center}
``What is the minimal strategy complexity required for "Eve"\\ to play optimally in all games with "objective" $W$?''
\end{center}

\paragraph{Positional objectives and universal graphs.}
As mentioned above, an important special case is that of "positional" "objectives", those for which "Eve" does not require any "memory" to play optimally.
A considerable body of research, with both theoretical and practical reach, has been devoted to the study of "positionality".
By now it is quite well-understood which "objectives" are "positional" for both players (bi-positional), thanks to the works of Gimbert and Zielonka~\cite{GZ05} for finite game graphs, and of Colcombet and Niwi\'nski~\cite{CN06} for arbitrary game graphs.
However, a precise understanding of which "objectives" are "positional" for "Eve" -- regardless of the opponent -- remains somewhat elusive, even though this is  a more relevant question in most application scenarios.

A recent progress in this direction was achieved by Ohlmann~\cite{Ohlmann23UnivJournal}, using "totally ordered" "monotone" "universal" "graphs".
Informally, an edge-coloured graph is "universal" with respect to a given "objective" $W$ if it "satisfies" $W$ (all paths "satisfy" $W$), and homomorphically "embeds" all "graphs" "satisfying" $W$.
An ordered "graph" is "monotone" if its edge relations are monotone:
\[
	v \geq u \re c u' \geq v' \implies v \re c v', \text{ for every colour } c.
\]
%Stated differently, one can push the head of edges to lower vertices, or pull their tails to higher vertices, while remaining in the graph.
Ohlmann's main result is a characterisation of "positionality" (assuming existence of a neutral letter): an "objective" is "positional" if and only if it admits "well-ordered" "monotone" "universal" "graphs".

\paragraph{From positionality to finite memory.}
"Positional objectives" have good theoretical properties and do often arise in applications (in particular, parity, Rabin or energy objectives).
It is also true, however, that this class lacks in expressivity and robustness: only a handful of "objectives" are "positional", and very few closure properties are known to hold for "positional" "objectives"\footnote{Kopczy\'nski conjectured in his thesis~\cite{Kop08Thesis} that "positional" "prefix-independent" "objectives" are closed under union. This conjecture was recently disproved by Kozachinskiy~\cite{Kozachinskiy22EnergyGroups} over finite game graphs, but it remains open for infinite graphs.}.

In contrast, "objectives" admitting optimal finite "memory@@strategy" "strategies" are much more general; for instance they encompass all $\omega$-regular "objectives"~\cite{GH82} (in fact, it was recently established~\cite{BRV23TheoretiCS} that optimal finite "chromatic memory" for both players characterises $\omega$-regularity).
Moreover, in practice, finite memory strategies can be implemented by means of a program, and memory bounds for "Eve" directly translates in space and time required to implement controllers, which gives additional motivation for their systematic study.

Formally, when moving from "positionality" to finite "memory@@strategy", a few modelling difficulties arise, giving rise to a few different notions.
Most prominently, one may or may not include uncoloured edges ($\eps$-edges) in the "game", over which the "memory state" cannot be "updated"; additionally one may or may not restrict to "chromatic memories", meaning those that record only the colours that have appeared so far.
We now discuss some implications of these two choices.

It is known that allowing $\eps$-edges impacts the difficulty of the "games", in the sense that it may increase the "memory@@strategy" required for winning "strategies"~\cite{Casares22, Kop08Thesis, Zielonka98}, thus leading to two different notions of "memory" (that we call "$\eps$-memory" and "$\eps$-free" memory).
It is natural to wonder whether one of the two notions should be preferred over the other.
We argue that allowing $\eps$-edges turns out to be more natural in many applications.
First, we notice that currently existing characterisations of the "memory@@objective" (for "Muller objectives"~\cite{DJW97} and for "topologically closed objectives"~\cite{CFH14}) do only apply to the case of "$\eps$-memory".
More importantly, "games" induced by logical formulas in which players are interpreted as the existential player (controlling existential quantifiers and disjunctions) and the universal player (controlling universal quantifiers and conjunctions) naturally contain $\eps$-edges (along which the memory indeed should not be allowed to be "updated").

It was originally conjectured by Kopczyński~\cite{Kop08Thesis} that "chromatic@@strategy" "strategies" have the same power than non-chromatic ones. It was not until recently that this conjecture was refuted~\cite{Casares22}, and since then several works have provided new examples separating both notions~\cite{CCL22SizeGFG, Kozachinskiy22InfSeparation, Kozachinskiy22ChromaticMem}.
It now appears from recent dedicated works~\cite{BORV23Journal, BRV23TheoretiCS, BLORV22Journal, Casares22} that "chromatic memory" is an interesting notion in itself.

The main challenge in the study of %positionality, and more generally of 
strategy complexity is to prove upper bounds on "memory@@objective" requirements of a given "objective".
A great feature of Ohlmann's result~\cite{Ohlmann23UnivJournal} is that it turns a question about "games" to a question about "graphs", which are easier to handle.
Despite its recent introduction, Ohlmann's framework has already proved instrumental for deriving strong "positionality" results in the context of "objectives" recognised by finite B\"uchi automata~\cite{BCRV24HalfJournal}, and more recently for arbitrary $\omega$-regular objectives~\cite{CO24}.

\subsection{Contribution}

The present paper builds on the aforementioned work of Ohlmann by extending it to encompass the more general setting of finite (or infinite) "memory@@objective" bounds.
This yields the first known characterisation results for objectives with given memory bounds, and provides a (provably) general tool for establishing memory upper bounds.

Doing so requires relaxing from totally to partially ordered graphs, while keeping the same "monotonicity" requirement, along with some necessary technical adjustments.
We essentially prove that the memory of an "objective" corresponds to the size of antichains in its "well-founded" "monotone" "universal" graph; however it turns out that the precise situation is more intricate.
It is summed up in Figure~\ref{fig:summary_of_results} and explained in more details below.

\begin{figure}[h]
	\begin{center}
		\includegraphics[width= \linewidth]{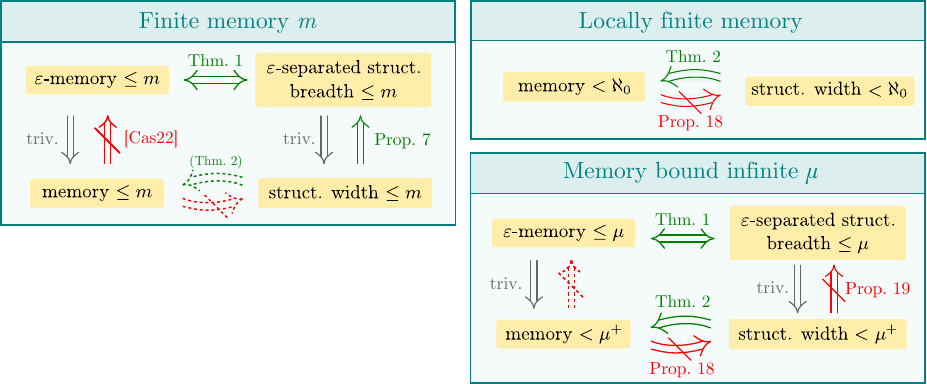}
	\end{center}
	\caption{A summary of our main contributions. The three larger boxes correspond to the three regimes encompassed by our results: finite memory, locally finite memory and larger "cardinal" bounds. Each of the smaller boxes correspond to classes of "objectives", where ``struct.'' stands for ``existence of "well-founded" "monotone" "universal" graphs''; for example, the box labelled ``$\eps$-separated struct. breadth $\leq m$'' stands for ``existence of "$\eps$-separated" "well-founded" "monotone" "universal" graphs of "breadth" $\leq m$''. The dotted implications follow from combining other implications in the figure. For $m=1$, all notions collapse to a single equivalence, which corresponds to Ohlmann's characterisation.}\label{fig:summary_of_results}
\end{figure}

It is convenient for us to define "strategies" directly as "graphs" (see Figure~\ref{fig:simple_example_intro} for an example, and Section~\ref{sec:preliminaries} for formal details), which allows us in particular to introduce new classes of "objectives" such as those admitting "locally finite memory", discussed in more details below.
For the well-studied case of finite memory bounds, our definition of "memory@@strategy" coincides with the usual one.

\paragraph{Universal structures for memory.}
Our main contribution lies in introducing generalisations of Ohlmann's structures, and proving general connections between existence of such "universal" structures for a given "objective" $W$, and "memory@@objective" bounds for $W$ 
(Section~\ref{sec:statement-main_results}).

The first variant we propose is obtained by relaxing the "monotonicity" requirement to "partially ordered" "graphs"; Theorem~\ref{thm:implication_non_eps} states that (potentially infinite) bounds on "antichains" of a "well-founded" "monotone" "universal" graph translate to memory bounds.

The second variant we propose, called "$\eps$-separated" structures, is tailored to capture "$\eps$-memory".
These are "monotone" graphs where the partial order coincides with $\re \eps$ and is constrained to be a disjoint union of "well-orders"; the "breadth" of such a "graph" refers to the number of such "well-orders".
Theorem~\ref{thm:characterisation_eps} states that the existence of such "universal" structures of "breadth" $\mu$ actually characterises having "$\eps$-memory" $\leq \mu$.
Additionally, we define "chromatic@@graph" "$\eps$-separated" structures (over which each colour acts uniformly), and establish that they capture "$\eps$-chromatic memory".

Applying (infinite) Dilworth's theorem we obtain that for finite $m$, one may turn any "monotone" graph of "width" $m$ to an "$\eps$-separated" one with "breadth" $m$ (Proposition~\ref{prop:non_eps-implies-eps}), and therefore in the setting of finite memory, the two notions collapse. 
We are able to establish most (but not all) of our results in the more general framework of quantitative "valuations"; similarly as Ohlmann~\cite{Ohlmann23UnivJournal}, we show how the notions instantiate in the qualitative case, how they can be simplified assuming prefix-invariance properties, and propose a general useful tool for deriving "universality" proofs (Lemma~\ref{lem:rongeur_de_croute}).

\paragraph{Counterexamples for a complete picture.}
We provide additional negative results (Section~\ref{sec:counter-examples}) which set the limits of our approach, completing the picture in Figure~\ref{fig:summary_of_results}.
Namely, we build two families of counterexamples that are robust to larger "cardinals"; these give general separations of "$\eps$-free memory" and "$\eps$-memory"\footnote{This result was already known for finite memory~\cite{Casares22}.} (Proposition~\ref{prop:infinite-antichains-no-eps-memory}), and negate the possibility of a converse for Theorem~\ref{thm:implication_non_eps} (Proposition~\ref{prop:eps-memory-greater-than-memory}).
This supports our informal claim that "$\eps$-memory" is better behaved than "$\eps$-free" "memory@@strategy".

\paragraph{Examples and applications.}
We argue (Section~\ref{sec:examples}) that our framework provides a very useful and flexible tool for studying memory requirements given concrete "objectives"; we provide a few illustrative examples for which we derive upper and lower bounds for each memory type.
We also illustrate the applicability of our tool by showing that the two available general characterisations of "memory@@objective" for special classes of "objectives", namely, the ones of Colcombet, Fijalkow and Horn~\cite{CFH14} for "topologically closed" "objectives", and of Dziembowski, Jurdzi\'nski and Walukiewicz~\cite{DJW97} for "Muller objectives", can both be understood as constructions of "monotone" "universal" graphs.

\paragraph{Closure properties.}
Finally, we discuss how our characterisations can be exploited for deriving closure properties on some classes of "objectives" (Section~\ref{sec:closure_properties}).
Apart from Ohlmann's result on "lexicographic products@@objectives" of "prefix-independent" "positional" "objectives"~\cite{Ohlmann23UnivJournal}, no such closure properties are known.
Extending Ohlmann's proof to our framework, we prove that if $W_1$ and $W_2$ are "prefix-independent" "objectives" with "$\eps$-memory" $m_1$ and $m_2$, then their "lexicographical product@@objectives" $W_1 \ltimes W_2$ has "$\eps$-memory" $\leq m_1 m_2$.
We also discuss a few implications of this result.

We then propose a new class of "objectives" with good properties, namely, "objectives" with "locally finite memory": for each "game", there exists a "strategy" which uses a finite (though possibly unbounded, even when the "game" is fixed) amount of "memory states" for each vertex.
These "objectives" are connected with the theory of "well-quasi orders" (wqo), since they correspond to "monotone" "universal" graphs which are "well-founded" and have finite "antichains".
We obtain from the fact that "wqo's" are closed under intersections, that intersections of "objectives" with finite "$\eps$-memory" have "locally finite memory"; an example is given by conjunctions of energy "objectives" which have unbounded finite memory even though energy "objectives" are "positional".
This hints at a general result, which is not implied by our characterisations but we conjecture to be true, that "objectives" with finite (possibly unbounded) memory are closed under intersection.

We end our paper by providing yet another application of our characterisation, establishing that prefix-independent $\kl{\Sigma_2^0}$ "objectives" with finite "memory@@objective" are closed under countable unions.
As of today, this is the only known (non-obvious) closure property pertaining to "objectives" with finite "memory@@objective".

%\newpage \tableofcontents \newpage

\section{Preliminaries}
\label{sec:preliminaries}

For a finite or infinite word $w\in C^*\cup C^\oo$ we denote by $w_i$ the letter at position $i$ and by $|w|$ its length.
For notations concerning order and set theory we refer the reader to Appendix~\ref{sec:appendix_set_theory}.
 
%\marginpar{\footnotesize See Appendix~A in the full version for basic definitions about set theroy.}

\subsection{Graphs and morphisms}
\paragraph{Graphs, paths and trees.}
\AP A ""$C$-pregraph"" $G$, where $C$ is a (potentially infinite) set of colours, is given by a set of vertices $\intro*\Verts{G}$, and a set of coloured directed edges $\intro*\Edges{G} \subseteq V(G) \times C \times V(G)$.
We write $v \re c v'$ for an edge $(v,c,v')$, say that it is outgoing from $v$, incoming in $v'$ and has colour $c$.
\AP A ""$C$-graph"" $G$ is a "$C$-pregraph" without sinks: from all $v \in \Verts{G}$ there exists an outgoing edge $v \re c v' \in \Edges{G}$.
We often say $c$-edges to refer to edges with colour $c$, and sometimes $C'$-edges for $C' \subseteq C$ for edges with colour in $C'$.

\AP A ""path"" in a "pregraph" $G$ is a finite or infinite sequence of edges of the form $\pi=(v_0 \re {c_0} v_1)(v_1 \re {c_1} v_2) \dots$,
which for convenience we denote by $\pi = v_0 \re {c_0} v_1 \re {c_1} \dots$.
We say that $\pi$ is a "path" from $v_0$ in $G$.
By convention, the empty "path" is a "path" from $v_0$, for any $v_0 \in \Verts{G}$.
If $\pi$ is a finite "path", it is of the form $v_0 \re {c_0} v_1 \re {c_1} \dots \re {c_{n-1}} v_n$, and in this case we say that it is a "path" from $v_0$ to $v_n$ in $G$. %, and that it has length $n$.
 \AP We let $\intro*\Ipath{v_0}{G} \subseteq \Edges{G}^\omega$ and $\intro*\Fpath{v_0}{G} \subseteq \Edges{G}^*$ respectively denote the sets of infinite and finite paths from $v_0$ in $G$.  

\AP Given a subset $X \subseteq V(G)$ of vertices of a "pregraph" $G$, we let $\intro{G|_X}$ denote the ""restriction@@graph"" of $G$ to $X$, which is the graph given by $V(G|_X)=X$ and $E(G|_X) = E(G) \cap (X \times C \times X)$. 
\AP Given a vertex $v\in V(G)$, we let $\intro*\treerooted{G}{v}$ denote the "restriction@@graph" of $G$ to vertices that are reachable from $v$.

\AP A ""$C$-tree"" (resp. ""$C$-pretree"") $T$ is a "$C$-graph" (resp. "$C$-pregraph") with an identified vertex $t_0 \in V(T)$ called its ""root"", with the property that for each $t \in V(T)$, there is a unique "path" from $t_0$ to $t$.
Note that since "graphs" have no sinks, trees are necessarily infinite.
We remark that $\treerooted{T}{t}$ represents the ""subtree rooted at $t$"" (if $T$ is a "tree", $\treerooted{T}{t}$ is also a "tree" with "root" $t$).

When it is clear from context, we omit $C$ and simply say ``a graph'' or ``a tree''.

\AP The ""size@@graph"" of a graph $G$ (and by extension, of a tree) is the cardinality of $\Verts{G}$.

\paragraph{Morphisms and unfoldings.}
\AP A ""morphism"" $\phi$ between two "graphs" $G$ and $H$ is a map $\phi\colon \Verts{G} \to \Verts{H}$ such that for each edge $v \re c v' \in \Edges{G}$ it holds that $\phi(v) \re c \phi(v') \in E(H)$.
We write $\phi : G \to H$ in this case, and sometimes say that $H$ "embeds" $G$.
Note that "morphisms" preserve paths: if $v_0 \re {c_0} v_1 \re{c_1} \dots$ is a "path" in $G$, then $\phi(v_0) \re {c_0} \phi(v_1) \re{c_1} \dots$ is a "path" in $H$.
\AP An ""isomorphism"" is a bijective "morphism" whose inverse is a "morphism"; two graphs are isomorphic if they are connected by an "isomorphism" (stated differently, they are the same up to renaming the vertices).
The composition of two "morphisms" is a "morphism".

\AP Given a "graph" $G$ and an initial vertex $v_0 \in G$, the ""unfolding"" of $G$ from $v_0$ is the "tree" $U$ with vertex set $\Verts{U} = \Fpath{v_0}{G}$ and edges
\[
\Edges{U} = \{(v_0 \re{c_0} \dots \re{c_{n-1}} v_n) \re {c_n} (v_0 \re{c_0} \dots \re{c_{n-1}} v_n \re {c_n} v_{n+1}) \mid v_n \re {c_n} v_{n+1} \in \Edges{G}\}.
\]
Note that the map $(v_0 \re{c_0} \dots \re{c_{n-1}} v_n) \mapsto v_n$ (with the empty "path" mapped to $v_0$) defines a "morphism" from $U$ to $G$.
  
\subsection{Valuations, games, strategies and memory}

\paragraph{Valuations and objectives.}
\AP A ""$C$-valuation"" is a map $\val : C^\omega \to X$, where $X$ is a complete linear order (that is, a total order in which all subsets have both a supremum and an infimum). %,for instance the real interval $[0,1]$, a non limit ordinal such as $\omega +1$, or the pair $\{\bot,\top\}$.
\AP The ""value@@graph"" $\intro*\val_G(v_0)$ of a vertex $v_0 \in \Verts{G}$ in a "graph" $G$ is the supremum value of infinite paths from $v$,
%\[
%\val_G(v_0) = \sup_{\pi \in \Ipath{v_0}{G}} \val(\pi),
%\]
where the value of an infinite "path" $\pi = v_0 \re {c_0} v_1 \re {c_1} \dots$ is defined to be $\val(\pi)=\val(c_0c_1 \dots)$.

\AP In the important special case where $X=\{\bot,\top\}$, $\bot < \top$, we identify\footnote{When considering an objective as a set of infinite words rather than a valuation $C^\omega \to \{\bot,\top\}$, we lose the information that $C$ is the set of colours that we are considering.
 This may be important in some cases, for instance $\emptyset \subseteq \{0\}^\omega$ and $\emptyset \subseteq \{1,2\}^\omega$ are not the same objective.
 However, it will always be clear from context what the set of colours is, and therefore, by a slight abuse, we avoid the hassle of defining objectives as tuples $(W,C)$.}
   {}
 $\val$ with $W=\val^{-1}(\bot) \subseteq C^\omega$, and say that $\val$ (or $W$) is an ""objective"".
\AP In a "graph" $G$, a "path" with value $\bot$ (equivalently, whose sequence of colours belongs to $W$) is said to ""satisfy"" $W$, and a vertex $v_0$ with "value@@graph" $\bot$ (equivalently, all paths from $v_0$ satisfy $W$) is also said to satisfy $W$. A "graph" is said to satisfy $W$ if all its vertices satisfy it.

\paragraph{Games.}
\AP A ""$C$-game"" is a tuple $\G=(G,\VE,v_0,\val)$, where $G$ is a "$C$-graph", $\intro*\VE$ is a subset of $\Verts{G}$, $v_0 \in \Verts{G}$ is an identified initial vertex, and $\val: C^\omega \to X$ is a "$C$-valuation".
\AP We interpret $\VE$ to be the set of vertices controlled by the first player, ""Eve"", and we will write $\intro*\VA = \Verts{G} \setminus \VE$ for the vertices controlled by her opponent, ""Adam"".
A game is played as follows: starting from $v_0$, successive moves are played where the player controlling the current vertex $v$ chooses an outgoing edge $v \re c v'$ and proceed to $v'$.
This interaction goes on forever, producing and infinite "path" $\pi$ from $v_0$.
"Eve"'s goal is to minimise the "value@@graph" of the produced "path" $\pi$, whereas "Adam" aims to maximise it.

In this paper, we are interested in questions of strategy complexity for Eve: if she wins, how much memory is required/sufficient?
Formally, these are independent of questions of determinacy (is there a winner?).
As a result, we will only ever consider strategies for Eve.

\paragraph{Strategies}
%This is formalised using the notion of strategies. 
\AP A ""strategy"" in the game $\G$ is a tuple $\S=(S,\pi_\S,s_0)$ where $S$ is a "graph", $\pi_\S$ is a "morphism" $\pi_S\colon S \to G$ called the ""$\S$-projection"" and $s_0\in \Verts{S}$ satisfying:
\begin{itemize}
	\item $\pi_\S(s_0)=v_0$,
	\item for all $v \in \VA$, all outgoing edges $v \re c v' \in \Edges{G}$ and all $s \in \pi_\S^{-1}(v)$, there is $s' \in \pi^{-1}(v')$ such that $s \re c s' \in E(S)$ (see Figure~\ref{fig:strategy}) .
\end{itemize}
Note that the requirements that $S$ is a "graph" and $\pi_\S$ a "morphism" impose that for all $v \in \VE$ and $s \in \pi_\S^{-1}(v)$, $s$ has an outgoing edge $s \re c s' \in E(S)$ satisfying $\pi_\S(s)=v \re c \pi_\S(s') \in\Edges{G}$.

\begin{figure}[h]
\begin{center}
\includegraphics[width=0.48 \linewidth]{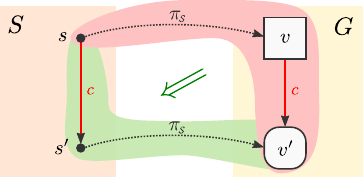}
\end{center}
\caption{Diagram illustrating the definition of a strategy. We use squares to represent vertices controlled by Adam and circles for vertices controlled by Eve. In this figure, it does not matter who controls $v'$.}\label{fig:strategy}
\end{figure}

We remark that we do not impose that for each $v\in \VE$ and $s \in \pi_\S^{-1}(v)$, $s$ has \emph{exactly one} outgoing edge. Stated differently, non-determinism is allowed in this definition of strategy. As the upcoming definition of "value@@strategy" of a "strategy" will clarify, we can interpret that "Adam" decides how to resolve this non-determinism.

On an informal level, a strategy $\S=(S,\pi_\S,s_0)$ from $v_0 \in G$ is used by Eve to play in the "game" $\G$ as follows:
\begin{itemize}
	\item whenever the game is in a position $v \in V(G)$, the strategy is in a position $s \in \pi_\S^{-1}(v)$;
	\item initially, the position in the game is $v_0$, and the position in the strategy is $s_0 \in \pi_\S^{-1}(v_0)$;
	\item if the position $v$ in the game belongs to $\VA$, and Adam chooses the edge $v \re c v'$ in $G$, then the strategy state is updated following an edge $s \re c s'$ in $S$ with $\pi_\S(s')=v'$, which exists by definition of $\S$ (if multiple options exist, Adam chooses one);
	\item if the position $v$ in the game belong to $\VE$, then the strategy specifies at least one successor $s \re c s'$ from the current $s \in \pi^{-1}(v)$, and the game proceeds along the edge $v \re c \pi(s')$ (if multiple options exist in the strategy, which corresponds to the non-determinism mentioned above, then "Adam" chooses one).
\end{itemize}

Note that infinite sequences of colours produced when playing as above are exactly labels of infinite paths from $s_0$ in $S$.

%Intuitively, a strategy $\S=(S,\pi_\S,s_0)$ from $v_0 \in G$ is used by Eve to play in the "game" $\G$ as follows:
%\begin{itemize}
%\item whenever the game is in a position $v \in V(G)$, the strategy is in a position $s \in \pi_\S^{-1}(v)$;
%\item at each step, positions in both the "game" and the strategy are updated along an edge of the same color;
%\item initially, the position in the game is $v_0$, and the position in the strategy is $s_0 \in \pi_\S^{-1}(v_0)$;
%\item if the position $v$ in the game belongs to $\VA$, and Adam chooses the edge $v \re c v'$ in $G$, then Eve follows an edge $s \re c s'$ in $S$ with $\pi_\S(s')=v'$, which exists by definition of $\S$;
%\item if the position $v$ in the game belongs to $\VE$, then Eve chooses an outgoing edge $s \re c s'$ in $S$, which exists since $S$ is sinkless, and progresses in $G$ along the edge $\pi_\S(s) \re c \pi_\S(s') = v \re c v' \in E(G)$.
%\end{itemize}
%
%Note that infinite paths from $v_0$ in $G$ which are visited when playing as above are exactly infinite paths from $s_0$ in $S$.

\AP The ""value@@strategy"" $\intro*\valStrat$ of a "strategy" $\S$ is $\val_S(s_0)$.
The ""value@@game"" $\intro*\valGame$ of a "game" is the infimum "value@@strategy" among its strategies.
If $\val$ is an "objective", we say that $\S$ is ""winning@@strategy"" if $\val_S(s_0)=\bot$, and we say that "Eve" ""wins"" a "game" $\G$ if $\val(\G)=\bot$.

The following observation is standard  (in fact, it   is usualy taken as the definition of a strategy).

\begin{lem}\label{lem:value_reached_on_trees}
The "value@@game" of a "game" is reached with "strategies" that are "trees".
%\marginpar{\footnotesize See Lemma~1 in the full version for the proof.}
\end{lem}

\begin{proof}
Let $\G$ be a "game" and $\S=(S,\pi_\S,s_0)$ a "strategy" over $\G$.
Consider the unfolding $U$ of $S$ from $s_0$, with "morphism" $\phi : U \to S$.
It is a direct check that $\U=(U,\pi_\U,\emptyword)$, where $\emptyword$ is the "root" of $U$ (represented by the empty "path"), and $\pi_\U = \pi_\S \circ \phi : U \to G$ is a "strategy".
Moreover, the fact that $\phi: U \to S$ is a "morphism" mapping $\emptyword$ to $s_0$ immediately yields $\val(\U) \leq \valStrat$.
\end{proof}

\paragraph{Memory.}
For a "strategy" $\S=(S,\pi_\S,s_0)$, we interpret the fibres $\pi_\S^{-1}(v)$ as memory spaces.
\AP Given a "cardinal" $\mu$, we say that $\S$ has ""memory@@strategy"" strictly less than $\mu$, (resp. less than $\mu$) if for all $v \in \Verts{G}$, $|\pi_\S^{-1}(v)| < \mu$ (resp. $|\pi_\S^{-1}(v)| \leq \mu$).
As it will appear later on, it is convenient for us to be able to use both strict and non-strict inequalities.
By means of clarity and conciseness, we usually simply write ``$\S$ has memory $< \mu$'' (resp. $\leq \mu$) instead of ``$\S$ has memory strictly less than $\mu$ (resp. less than $\mu$)''.

\AP We say that a "valuation" $\val$ has ""memory strictly less than"" $\mu$, or $< \mu$, (resp. less than $\mu$, or $\leq \mu$) if in all "games" with "valuation" $\val$, the "value@@game" is reached with "strategies" with "memory@@strategy" $<\mu$.

\AP Conversely, we say that $\val$ has ""memory at least"" $\mu$ (resp. strictly more than $\mu$), or $\geq \mu$ (resp. $>\mu$), if it does not have memory $<\mu$ (resp. $\leq \mu$): there exists a "game" with "valuation" $\val$ in which "Eve" cannot reach the "value@@game" with "strategies" with "memory@@strategy" $<\mu$ (resp. $\leq \mu$).

\AP Finally, if there exists\footnote{It might be that there is no cardinal $\mu$ such that $\val$ has memory exactly $\mu$ (intersections of energy objectives are an example, see Section~\ref{sec:combinations}).} $\mu$ such that $\val$ has memory $\geq \mu$, but memory $< \alpha$ for all $\alpha > \mu$, then we say that $\val$ has ""memory exactly"" $\mu$.
 
\AP We say that $\val$ is ""positional""\hspace*{1pt}\footnote{This is sometimes called \emph{half-positionality} in the literature.}   if it has "memory" $\leq 1$.

\paragraph{Product strategies, chromatic strategies.} \AP A "strategy" $\S = (S,\pi_\S,s_0)$ in the "game" $\G$ is a ""product strategy"" over a set $M$ if $\Verts{S} \subseteq \Verts{G} \times M$, with $\pi_\S(v,m) = v$.
%Intuitively, for product strategies, there is a uniform notion of \emph{memory states}.
We call the elements of $M$ ""memory states"".
Note that the "memory@@strategy" in a "product strategy" over $M$ is $\leq |M|$, since fibers are included in $M$.
\AP A "product strategy" is ""chromatic@@strategy"" if there is a map $\delta: M \times C \to M$ such that for all $(v,m) \re c (v',m') \in \Edges{S}$ we have $m' = \delta(m,c)$.
We say in this case that $\delta$ is the ""update function"" of $\S$.
In words, the update of the "memory state" in a "chromatic strategy" depends only on the current "memory state" and the colour that is read.
\AP A "valuation" $\val$ has ""chromatic memory"" $< \mu$ (resp. $\leq \mu$) if in all "games" with "valuation" $\val$, the "value@@game" is reached with "chromatic@@strategy" "strategies" with "memory@@strategy" $<\mu$ (resp. $\leq \mu$).

\paragraph{$\eps$-games and $\eps$-strategies.}
\AP Fix a set of colours $C$, a fresh colour $\eps \notin C$, and let $\intro*\Ceps=C \sqcup \{\eps\}$.
\AP The ""$C$-projection"" of an infinite sequence $w \in (\Ceps)^\omega$ is the (finite or infinite) sequence $w_C \in C^* \cup C^\omega$ obtained by removing all $\eps$'s in $w$.
Given a "$C$-valuation" $\val: C^\omega \to X$, define its ""$\eps$-extension"" $\intro*\vale$ to be given by
\begin{align*}
\vale(w)&=\begin{cases}
\val(w_C),& \tif |w_C| = \infty, \\
\inf\limits_{w' \in C^\omega} \val(w_C w'),& \tow.
\end{cases}
\end{align*}
It is the unique extension of $\val$ with $\eps$ as a strongly neutral colour, in the sense of Ohlmann~\cite{Ohlmann23UnivJournal}. In particular, if $W$ is an "objective" and $w\in C^*$, $w\eps^\oo \in \Weps$ unless $w$ has no winning continuation in $W$.  

\AP An ""$\eps$-game"" $\G$ is a $\Ceps$-game with "valuation" $\vale$.
An ""$\eps$-strategy"" over such a game is a "product strategy" $\S=(S,\pi_\S,s_0)$ over some set $M$ such that $(v,m) \re \eps (v',m') \in E(S) $ implies $m=m'$.
Intuitively, "Eve" is not allowed to update the state of the memory when an $\eps$-edge  is traversed.
\AP The ""memory of an $\eps$-strategy"" is defined to be $|M|$.
\AP A "valuation" $\val$ has ""$\eps$-memory <"" $\mu$ (resp. $\leq \mu$) if in all "$\eps$-games" with "valuation" $\vale$, the "value@@game" is attained by "$\eps$-strategies" with "memory@@strategy" $<\mu$ (resp. $\leq \mu$). \AP Having ""$\eps$-memory $\geq$""$\mu$, $>\mu$, and the ""exact $\eps$-memory"" is defined as before.

\begin{prop}\label{prop:non-strict-ineq-for-eps-memory}
	Let $W$ be an "objective". If $W$ has "$\eps$-memory" $<\mu$, for some "cardinal" $\mu$, then there is some "cardinal" $\aa < \mu$ such that $W$ has "$\eps$-memory" $\leq\aa$.
\end{prop}
Therefore, for "$\eps$-memory" and in the case of "objectives", we can restrict our study to non-strict inequalities without loss of generality. Moreover, the "exact $\eps$-memory" of an "objective" is always defined.

\begin{proof}[Proof of Proposition~\ref{prop:non-strict-ineq-for-eps-memory}]
	Suppose by contradiction that $W$ has "$\eps$-memory" $<\mu$ and that it has "$\eps$-memory" $>\aa$ for all $\aa < \mu$. By definition of having "$\eps$-memory" $>\aa$, for each $\aa < \mu$ there is a "game" in which "Eve" can "win", but she cannot do so with "strategies" with "$\eps$-memory@@strategy" $\leq\aa$. Let $\G_\aa=(G_\aa, V_{\mathrm{Eve},\aa}, v_{0,\aa}, \val)$ be such a "game", for $\aa < \mu$. We take the disjoint union of all these games and we let "Adam" choose the initial vertex among $v_{0,\aa}$. Formally, let $\G=(G, \VE, v_0, \val)$, where:
	\begin{itemize}
		\item $\Verts{G} = \bigsqcup_{\aa<\mu}\Verts{G_\aa} \cup \{v_0\}$,
		\item $\VE = \bigsqcup_{\aa<\mu}V_{\mathrm{Eve},\aa}$,
		\item $\Edges{G} = \bigsqcup_{\aa<\mu}\Edges{G_\aa} \cup \{v_0 \re \eps v_{0,\aa} \mid \aa < \mu\}$.
	\end{itemize}
	First, we remark that "Eve" "wins" this "game":  no matter Adam's choice, after the first $\eps$-move the play will take place in some game $\G_\aa$, where "Eve" can use a "winning@@strategy" "strategy".	
	Let $\S$ be a "winning@@strategy" "$\eps$-strategy" over some set $M$, $|M|=\aa<\mu$ (that exists since we have supposed that $W$ has "$\eps$-memory" $<\mu$). Let $s_0=(v_0,m_0)$. Since $(v_0,m_0) \re \eps (v_{0,\aa},m_0)$, and $\S$ is "winning@@strategy", all "paths" from $(v_{0,\aa},m_0)$ "satisfy" $W$. Therefore, the "restriction@@graph" of $\S$ to $\{(v,m) \mid v\in \Verts{G_\aa}\}$ is a "winning@@strategy" "$\eps$-strategy" with "$\eps$-memory@@strategy" $\leq \aa$, which contradicts the fact that "Eve" cannot "win" $\G_\aa$ using "strategies" with "$\eps$-memory@@strategy" $\leq\aa$.
\end{proof}

Note that by definition, a "chromatic strategy" over $M$ with "update function" $\delta$ is an "$\eps$-strategy" if and only if for all $m \in M$ it holds that $\delta(m,\eps) = m$.
\AP We call such a strategy an ""$\eps$-chromatic strategy"".
A "valuation" $\val$ has ""$\eps$-chromatic memory"" $<\mu$ (resp. $\leq \mu$) if in all "$\eps$-games" with "valuation" $\vale$, the "value@@game" is attained by "$\eps$-chromatic strategies" with "memory" $<\mu$ (resp. $\leq \mu$).  The ""exact $\eps$-chromatic memory"" is defined analogously. 

\AP Whenever we want to emphasise that we consider "games" (resp. "strategies", "memory@@objective") without $\eps$, we might add the adjective ""$\eps$-free"".

\subsection{Monotonicity and universality}
\paragraph{Monotonicity.}
\AP A partially ordered "graph" $(G,\leq)$ is ""monotone"" if
\[
u \geq v \re c v' \geq u' \text{ implies } u \re c u' \text{ in } G.
\]
\AP A partially ordered "graph" $(G,\leq)$ is called ""well-monotone"" if it is "monotone" and it is "well-founded" as a partial order.
\AP We say that the ""width"" of a partially ordered "graph" is $<\mu$ (resp. $\leq \mu$) if it does not contain "antichains" of size $\mu$ (resp. of size strictly greater than $\mu$).

\paragraph{$\eps$-separation.}
\AP An ""$\eps$-separated monotone graph over a set $M$"" is a $\Ceps$-graph $G$ such that $\re \eps$ defines a partial order making $G$ "monotone" ($v\leq v' \iff v'\re \eps v\in \Edges{G}$), and moreover $\Verts{G}$ is partitioned into $(V_m)_{m \in M}$ such that for all $m \in M$, $\re \eps$ induces a total order over~$V_m$, and there are no $\eps$-edges between different parts: $v \re \eps v' \in \Edges{G}$ implies that $v,v' \in V_m$ for some $m \in M$.
See Figure~\ref{fig:monotone_and_eps_separated}.
\AP We define the ""breadth"" of such a graph as $|M|$.

\begin{figure}[h]
\begin{center}
\includegraphics[width = 0.7 \linewidth]{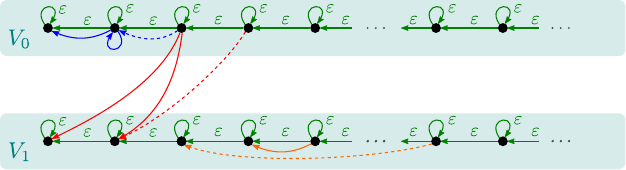}
\end{center}
\caption{An "$\eps$-separated" "chromatic@@epsGraph" "monotone" graph of "breadth" 2. Note that $\re \eps$ defines a total order on each $V_i$ (edges following from transitivity are not represented). Many edges which follow from "monotonicity" are not depicted, the dotted edges give a few examples.}\label{fig:monotone_and_eps_separated}
\end{figure}

\AP An "$\eps$-separated monotone graph" $G$ over $M$ is ""chromatic@@epsGraph"" if there is a map $\delta: M \times C \to M$ such that for all $v \re c v' \in \Edges{G}$ with $v \in V_m$ and $v' \in V_{m'}$ we have $m'=\delta(v,m)$.
We also say in this case that $\delta$ is the ""update function@@graph"" of $G$.

\paragraph{Universality.}
\AP Given a "$C$-valuation" $\val$, a "$C$-graph" $G$ and a "cardinal" $\kappa$, we say that $G$ is ""$(\kappa,\val)$-universal""\hspace*{1pt}\footnote{This definition is tailored to the general setting of quantitative "valuations", for which we are able to present most results. When specifying to "objectives" (more precisely, to "prefix-increasing" "objectives") the concept of universality can be simplified without loss of generality. This will be the object of Section~\ref{sec:prefix_increasing_objectives}.}
 if for all "$C$-trees" $T$ of cardinality $< \kappa$, there exists a "morphism" $\phi: T \to G$ such that
\[
\val_G(\phi(t_0)) \leq \val_T(t_0),
\]
where $t_0$ is the "root" of $T$.
\AP We say that $\phi$ ""preserves the value"" at the root to refer to this property (we remark that, in that case, $\val_G(\phi(t_0)) = \val_T(t_0)$, since the other inequality always holds).

\begin{rem}

	In the above definition, for a graph to be "universal", it needs to "embed" all \emph{trees} (up to a given cardinality bounds), and not all \emph{graphs} as in the case of "positionality"~\cite{Ohlmann23UnivJournal}.
	For (totally) well-ordered graphs, that is, in the case of "positionality", this does not make a difference; however for the current study of "memory@@objective", this difference is important. 
	An example where the definition with graphs is too constrained to capture "memory@@objective" is given in Proposition~\ref{prop:embed-tree-not-graphs}.
\end{rem}

\begin{rem}
	We remark that if $U$ is a "$(\kappa,\vale)$-universal" graph, then the graph $U'$ obtained by removing the edges labelled by $\eps$ is "$(\kappa,\val)$-universal". Moreover, if $U$ is an "$\eps$-separated" monotone graph of "breadth" $\mu$, then $U'$ is a monotone graph of "width" $\leq \mu$.
\end{rem}

\section{Main characterisation results}
\label{sec:main_results}

In this section, we state (Section~\ref{sec:statement-main_results}) and prove (Sections~\ref{sec:structure_to_finite_memory} and~\ref{sec:finite_memory_to_structure}) our two main results, Theorems~\ref{thm:characterisation_eps} and~\ref{thm:implication_non_eps}.
This is followed by additional general results (Section~\ref{sec:further_results}).

\subsection{Statement of the results}
\label{sec:statement-main_results}

We start with our characterisations of "$\eps$-memory" and "$\eps$-chromatic memory" via (chromatic) "$\eps$-separated" "universal" graphs.

\begin{thm}\label{thm:characterisation_eps}
Let $\val$ be a "valuation".
If for all "cardinals" $\kappa$ there exists an "$\eps$-separated" ("chromatic@@epsGraph") and "well-monotone" "$(\kappa,\vale)$-universal" graph of "breadth" $\leq \mu$, then $\val$ has $\eps$(-"chromatic@@epsMemory")-"memory@@eps" $\leq \mu$.
The converse holds if $\val$ is an "objective" (in both the chromatic and non-chromatic cases).
\end{thm}

As explained by Proposition~\ref{prop:non-strict-ineq-for-eps-memory}, strict inequalities, though they give more precise statements, are irrelevant for "$\eps$-memory".
Thus the use of non-strict inequalities in the statement above is not restrictive.

We state our second result in terms of strict inequalities, which is relevant in the case of "$\eps$-free" "memory@@epsFree", and allows for more precision.
However, we do not have a converse statement (as discussed in the introduction, the converse cannot hold, see also Figure~\ref{fig:summary_of_results} and Proposition~\ref{prop:infinite-antichains-no-eps-memory}).

\begin{thm}\label{thm:implication_non_eps}
Let $\val$ be a "valuation".
If for all "cardinals" $\kappa$ there exists a "well-monotone" "$(\kappa,\val)$-universal" "graph" of "width" $< \mu$, then $\val$ has "$\eps$-free memory" $< \mu$.
\end{thm}

As we will see in Section~\ref{sec:finitely_bounded_antichains}, the two results above collapse for finite cardinals $\mu$.

\begin{rem}\label{rmk:arena-independent-memory}
	We remark that we say that the "($\eps$-)chromatic memory" of an "objective" is $\leq\mu$ if for all "games", the "value@@game" can be attained with a "chromatic@@memory" "product strategy" over some structure $M$, $|M|\leq \mu$, with "update function" $\dd$.
	We could ask if it is possible to modify the order of the quantifiers in this definition, that is, if we could fix the structure $M$ and its "update function" in advance, regardless of the game. 
	The notion obtained in that way is called \AP ""arena-independent memory"" in the recent literature~\cite{BLORV22Journal}.
	
	Over "$\eps$-games", the size of a minimal "arena-independent memory" for an "objective" coincide with its "$\eps$-chromatic" memory (this is proved for the case of finite memory in~\cite[Proposition~8.9]{Kop08Thesis}).
	We note that this result can be easily derived from Theorem~\ref{thm:characterisation_eps} and its proof: the existence of an "$\eps$-separated" "chromatic@@epsGraph" "universal" graph over the structure $M$ implies that $M$ is an "arena-independent memory" (see Section~\ref{sec:structure_to_finite_memory}), and the existence of such a graph is guaranteed by the implication from right to left of this theorem.
	
	We do not know whether the sizes of a minimal "$\eps$-free" "arena-independent memory" and the "$\eps$-free chromatic" memory also coincide.
\end{rem}

\subsection{From structure to finite memory}
\label{sec:structure_to_finite_memory}

The goal of this section is to prove Theorem~\ref{thm:implication_non_eps} and the first implication in Theorem~\ref{thm:characterisation_eps}.
The two proofs are very similar; we start with Theorem~\ref{thm:implication_non_eps}.

\begin{proof}[Proof of Theorem~\ref{thm:implication_non_eps}]
	Let $\val : C^\omega \to X$ be a "valuation" , $\G = (G, \VE, v_0, \val)$ a "game" and $\T = (T,\pi_\T,t_0)$ be a "strategy" for $\G$ such that $T$ is a "tree".
	Our aim is to define a "strategy" with "memory@@strategy" $<\mu$ and "value@@strategy" $\leq \val(\T)$; this proves that $\val$ has memory $<\mu$ thanks to Lemma~\ref{lem:value_reached_on_trees}.
	
	Take a "well-monotone" "$(|T|,\val)$-universal" graph $(U,\leq)$ with "width" $<\mu$, and consider a "morphism" $\phi: T \to U$ "preserving the value" at the "root", $\val_T(t_0) = \val_{U}( \phi(t_0))$.
	For each $v \in \Verts{G}$, we consider the set $M_v \subseteq V(U)$ of minimal elements of $\phi(\pi_\T^{-1}(v))$ (see Figure~\ref{fig:structure_to_finite_memory}).

	\begin{figure}[h]
		\begin{center}
			\includegraphics[width = 0.8 \linewidth]{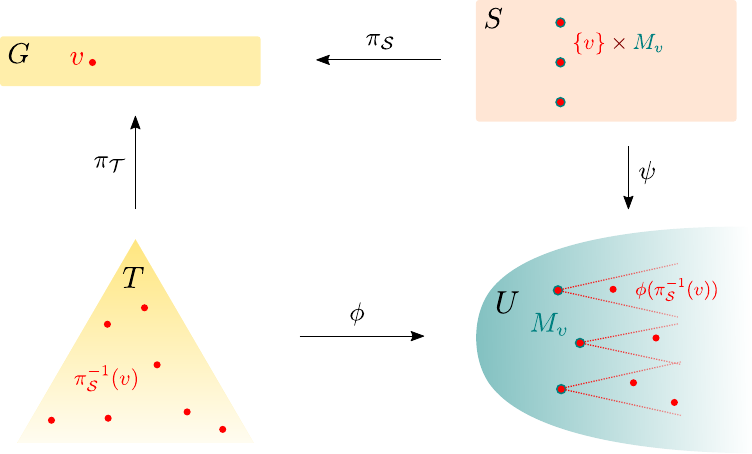}
		\end{center}
		\caption{An illustration for the construction of the bounded-memory strategy $\S$ in the proof of Theorem~\ref{thm:implication_non_eps}.}\label{fig:structure_to_finite_memory}
	\end{figure}

	We define our "strategy" $\S=(S,\pi_\S,s_0)$ over 
	\[
	V(S) = \bigsqcup_{v \in \Verts{G}} \{v\} \times M_v,
	\]
	with "projection@@strategy" $\pi_\S: (v,m) \mapsto v$, and let $s_0=(v_0,m_0)$ where $m_0 \in M_{v_0}$ is an element below $\phi(t_0)$ in $V(U)$.
	Note that for all $v \in \Verts{G}$, $M_v$ is an antichain of $V(U)$ and therefore $|\pi_S^{-1}(v)| = |M_v| < \mu$, as required.
	
	For each element $(v,m) \in V(S)$, fix a choice of a $t_{(v,m)} \in \pi_\T^{-1}(v)$ such that $\phi(t)=m$.
	We now let
	\[
	E(S) = \{(v,m) \re c (v',m') \mid \exists t' \in \pi_{\T}^{-1}(v'), t_{(v,m)} \re c t' \in E(T) \tand \phi(t') \geq m'\},
	\]
	which concludes the definition of $\S$.
	
	Let us verify that $\S$ is indeed a "strategy" over $\G$.
	It is clear that $\pi_\S(s_0)=v_0$.
	Now observe that for any $(v,m) \in V(S)$, and any edge $t_{(v,m)} \re c t' \in E(T)$, if we denote $v'=\pi_\T(t')$, there is an element $m' \leq \phi(t')$ in $M_{v'}$.
	This induces an edge $(v,m) \re c (v',m') \in E(S)$.
	This implies, since $T$ is a graph (it has no sink), that $S$ is a graph.
	Moreover, for all $v \in \VA$ and outgoing edge $v \re c v' \in \Edges{G}$, since $T$ is a "strategy" $t_{(v,m)}$ has an outgoing edge in $T$ towards some $t'$ with $\pi_\T(t') = v'$, thus by the above observation, $(v,m)$ has an outgoing edge in $S$ towards an element $(v',m')$ (which has "projection@@strategy" $\pi_\S(v',m')=v'$, as required) and $\S$ is a "strategy".
	
	There remains to see that $\valStrat \leq \val(\T)$.
	We will in fact prove that $\psi : (v,m) \mapsto m$ is a "morphism" from $S$ to $U$, which implies that 
	\[
	\valStrat = \val_S(s_0) \leq \val_U(\psi(s_0)) = \val_U(m_0) \leq \val_U(\phi(t_0)) = \val(\T),
	\]
	the wanted result.
	Let $(v,m) \re c (v',m') \in E(S)$, we aim to prove that $m \re c m' \in E(U)$.
	Let $t'$ be such that $\pi_\T(t')=v'$, $t_{(v,m)} \re c t'$ and $\phi(t') \geq m'$.
	Since $\phi$ is a "morphism" we have in $U$
	\[
	m = \phi(t_{(v,m)}) \re c \phi(t') \geq m',
	\]
	thus by monotonicity, $m \re c m' \in E(U)$.
\end{proof}

The proof of the first implication in Theorem~\ref{thm:characterisation_eps} is essentially the same, with a few minor adjustments.
We spell it out for completeness.

\begin{proof}[Proof of $\implies$ in Theorem~\ref{thm:characterisation_eps}]
	Let $\val : C^\omega \to X$ be a "valuation" , $\G=(G,\VE,v_0,\vale)$ an $\eps$-game and $\T = (T,\pi_T,t_0)$ a "strategy" for $\G$ such that $T$ is a "tree".
	Our aim is to define an "$\eps$-strategy" with "memory@@strategy" $\leq \mu$ and "value@@strategy" $\leq \val(\T)$.
	Take an "$\eps$-separated" "well-monotone" "$(|T|,\vale)$-universal" graph $(U,\re \eps)$ with partition $(U_m)_{m \in M}$ of "width" $|M| \leq \mu$, and consider a "morphism" $\phi: T \to U$ "preserving the value" at the root.
	We define the "product strategy" $\S=(S,\pi_\S,s_0)$ by
	\[
	V(S) = \{(v,m) \in \Verts{G} \times M \mid \phi(\pi_\T^{-1}(v)) \cap U_m \neq \emptyset\},
	\]
	with $s_0 = (v_0,m_0)$, where $m_0$ is such that $\phi(t_0) \in U_{m_0}$,
	and with "projection@@strategy" $\pi_\S: (v,m) \mapsto v$.
	To define $E(S)$, we pick for each $(v,m) \in V(S)$ an element $t_{(v,m)} \in V(T)$ such that $\phi(t_{(v,m)})=\min\{\phi(\pi_\T^{-1}(v)) \cap U_m\}$, and let
	\[
	E(S) = \{(v,m) \re c (v',m') \mid \exists t' \in \pi_\T^{-1}(v'), t_{(v,m)} \re c t' \in E(T) \tand \phi(t') \in U_{m'}\}.
	\]
	
	We verify that $\S$ is indeed a "strategy" over $\G$.
	By definition, we have $\pi_\S(s_0)=v_0$.
	Observe that for any $(v,m) \in V(S)$, and any edge $t_{(v,m)} \re c t' \in E(T)$, there is an edge $(v,m) \re c (v',m') \in E(S)$ where $v' = \pi_\T(t')$ and $m'$ is such that $\phi(t') \in U_{m'}$.
	This implies that $\S$ is a "strategy" since $\T$ is.
	
	We now prove that $\psi:(v,m) \mapsto \min\{\phi(\pi_\T^{-1}(v)) \cap U_m\}$ is a "morphism" from $S$ to $U$.
	This implies
	\[
	\valStrat = \val_S(s_0) \leq \val_U(\psi(v_0,m_0))  \leq \val_U(\phi(t_0)) = \val(\T),
	\]
	the wanted result.
	Let $(v,m) \re c (v',m') \in E(S)$, and let $t'$ be such that $\pi_\T(t')=v'$, $t_{(v,m)} \re c t' \in E(T)$ and $\phi(t') \in U_{m'}$.
	We have by definition $\phi(t_{(v,m)})=\psi(v,m)$ and $\phi(t') \geq \psi(v',m')$, therefore we conclude by monotonicity of $U$ that $\psi(v,m) \re c \psi(v',m')$.
	Finally, remark that if $c = \eps$, since $\psi(v,m) \in U_m$ and $\psi(v',m') \in U_{m'}$ and there are no $\eps$-edges in $U$ between different partitions, it must be that $m=m'$ which concludes our proof for the non-chromatic case: $\S$ is indeed an "$\eps$-strategy".
	
	For the "chromatic@@memory" case, it suffices to show in the construction above that if $U$ is in fact chromatic, then so is the constructed "strategy" $\S$.
	For this, we observe that the "morphism" $\psi$ above maps $(v,m) \in V(S)$ to a vertex in $U_m$, therefore if $(v,m) \re c (v',m') \in E(S)$ and $\delta$ is the "update function@graph" of $U$, it must be that $\delta(m,c)=m'$.
	We conclude that $\S$ is indeed a "chromatic strategy" with "update function" $\delta$.
\end{proof}

\subsection{From finite memory to structure}
\label{sec:finite_memory_to_structure}

In this section, we prove the converse implication in Theorem~\ref{thm:characterisation_eps}.
The main difficulty lies in proving the following result, which holds at the level of "valuations", and which we refer to as a structuration lemma for $C^\eps$-trees.

\begin{lem}[Structuration of $C^\eps$-trees]\label{lem:structuration}
	Let $\val: C^\omega \to X$ be a "valuation"  with $\eps$(-"chromatic@@epsMemory")-"memory@@eps" $\leq \mu$ and let $T$ be a $\Ceps$-tree with root $t_0 \in V(T)$.
	There exists an "$\eps$-separated" "well-monotone" ("chromatic@@epsGraph") "graph" $U$ of "breadth" $\leq \mu$ and a "morphism" $T \to U$ preserving the "value@@graph" at the "root". 
\end{lem}

Before proving the lemma, we show that it implies the Theorem.

\begin{proof}[Proof of $\impliedby$ in Theorem~\ref{thm:characterisation_eps} assuming Lemma~\ref{lem:structuration}]
	We consider an "objective" $W \subseteq C^\omega$ which has $\eps$-memory $\leq \mu$, and fix a "cardinal" $\kappa$.
	We consider the disjoint union of all $C^\eps$-trees of cardinality $< \kappa$ whose roots satisfy $\Weps$, up to "isomorphism", and we let $T$ be the tree with root $t_0$ obtained from this disjoint union by adding an $\eps$-edge from $t_0$ to the root of each tree (see Figure~\ref{fig:big_tree}).
	Note that $t_0$ satisfies $\Weps$.

	\begin{figure}[h]
		\begin{center}
			\includegraphics[width = 0.53 \linewidth]{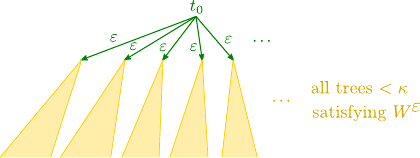}
		\end{center}
		\caption{The tree on which Lemma~\ref{lem:structuration} is applied.} \label{fig:big_tree}
	\end{figure}
	
	We now apply Lemma~\ref{lem:structuration} to $T$ and obtain an "$\eps$-separated" "well-monotone" ("chromatic@@epsGraph") graph $U$ of "breadth" $\leq \mu$ with a "morphism" $\phi: T \to U$ such that $\phi(t_0)$ "satisfies" $\Weps$ in $U$.
	There remains to prove that $U$ is $(\kappa,\Weps)$-"universal".
	Consider a $\Ceps$-tree $T'$ of cardinality $< \kappa$ and whose root "satisfies" $\Weps$.
	By definition of $T$, there is $t'$ in $T$ with $t_0 \re \eps t'$ such that the tree rooted at $t'$ in $T$ is isomorphic to $T'$.
	We then obtain a "morphism" $\phi': T' \re {} U$ simply as a restriction of $\phi$ (composed with the "isomorphism").
	Since $\phi(t_0)$ satisfies $\Weps$ in $T$, so does $\phi(t')$, and therefore $\phi'$ preserves the "value@@graph" at the root, as required.
	
	To accommodate trees whose root do not satisfy $W$, in the non-chromatic case it suffices to add an additional vertex $\top$ (in any chosen part $U_m$) with $c$-edges towards all $U_m$ (including itself) for all $c \in C^\eps$.
	This preserves being "$\eps$-separated" "well-monotone" of "breadth" $\leq \mu$, does not increase the "value@@graph" of vertices $\neq \top$, and allows to embed (while preserving the "value@@graph" at the root) any tree $T'$ whose root does not satisfy $W$ simply by mapping everything to $\top$.
	
	The "chromatic@@memory" case requires being slightly more careful.
	Let $\delta$ be the "update function@graph" of $U$.
	For each $m \in M$ we add a vertex $\top_m \in U_m$, with $c$-edges towards all $U_{m'}$ (including $\top_{m'}$) whenever $\delta(m,c)=m'$.
	This preserves being "$\eps$-separated", "well-monotone", "chromatic@@epsGraph" and of "breadth" $\leq \mu$, and does not increase the "value@@graph" of vertices $\notin \{\top_m \mid m \in M\}$.
	Now, if $T'$ is a "tree" whose "root" $t'_0$ does not "satisfy" $W$, we easily embed it in a top-down fashion, by mapping $t'_0$ to $\top_{m_0}$ (for any choice of $m_0$), and mapping $t' \in V(T)$ to $\delta^*(m_0,w)$, where $w$ is the label of the unique "path" from $t'_0$ to $t'$ in $T$.
\end{proof}

We now prove Lemma~\ref{lem:structuration}; our proof extends the one of~\cite[Theorem~3.3]{Ohlmann23UnivJournal}.

\begin{proof}[Proof of Lemma~\ref{lem:structuration}]
	Let $\val:C^\omega \to X$ be a "valuation"  with $\eps$(-"chromatic@@epsMemory")-"memory@@eps" $\leq \mu$ and $T$ be a $C^\eps$-tree with root $t_0$.
	We consider the "$\eps$-game" $\G=(G,\VE,v_0,\vale)$ obtained by adding an "Eve" vertex for each non-empty set $A$ of vertices of $T$, and $\eps$-edges back and forth from $t$ to $A$ whenever $t \in A$, with the control given to "Adam" over $V(T)$.
	Formally, it is given by
	\[
	\begin{array}{rcl}
	\Verts{G} & = & V(T) \cup \powne(V(T)) \\
	\VE & = & \powne(V(T)) \\
	\Edges{G} & = & E(T) \cup \{ t \re \eps A \mid t \in A \} \cup \{ A \re \eps t \mid t \in A \},
	\end{array}
	\]
	and $v_0 = t_0$. See Figure~\ref{fig:choice_arena} for an illustration.

	\begin{figure}[h]
		\begin{center}
			\includegraphics[width = 0.6 \linewidth]{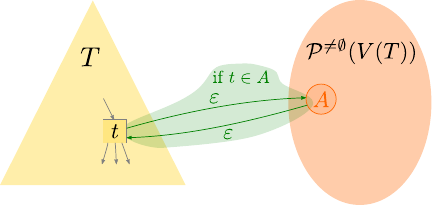}
		\end{center}
		\caption{The game $\G$.}\label{fig:choice_arena}
	\end{figure}
	
	We claim that the "value@@game" of $\G$ is $\leq \val_T(t_0)$.
	Indeed, consider the "strategy" for "Eve" which, whenever arriving at $A \in \VE$ via an edge $t \re \eps A$, follows the edge $A \re \eps t$ back towards $t$.
	Consider an infinite "path" $\pi$ from $t_0$ in that "strategy", and let $\pi'$ be obtained from $\pi$ by removing all occurrences of $t \re \eps A \re \eps t$.
	Note that $\pi'$ defines a "path" from $t_0$ in $T$.
	There are two cases.
	\begin{itemize}
		\item If $\pi'$ is infinite, then by neutrality of $\eps$ it has the same "value@@graph" as $\pi$.
		\item If $\pi'$ is finite, then any continuation of $\pi'$ in $T$ has "value@@graph" $\geq \vale(\pi)$ by definition of $\vale$.
	\end{itemize}
	This proves that for each infinite "path" $\pi$ from $t_0$ in the "strategy", there exists an infinite "path" of "value@@graph" $\geq \vale(\pi)$ from $t_0$ in $T$, and thus $\vale(\G) \leq \val_T(t_0)$.
	
	Since $\val$ has $\eps$(-"chromatic@@epsMemory")-"memory@@eps" $\leq \mu$, there exists an $\eps$(-"chromatic@@epsStrategy") "strategy@@eps" $\S=(S,\pi_\S,s_0)$ over $\G$ with "value@@strategy" $\vale(\S)=\vale(\G)$ and "memory@@eps" $\leq \mu$.
	By definition we have $V(S) \subseteq \Verts{G} \times M$ with $|M| \leq \mu$, $\pi_\S : (v,m) \mapsto v$ and $(v,m) \re \eps (v',m') \in E(S)$ implies $m=m'$.
	In particular, we have $s_0=(t_0,m_0)$ for some $m_0 \in M$.

	For each $(t,m) \in V(S)$ with $t \in V(T)$, and each edge $t \re c t' \in E(T)$, it holds that $t \re c t' \in \Edges{G}$ and $t \in \VA$ therefore there is $(t',m') \in V(S)$ with $(t,m) \re c (t',m') \in E(S)$ since $S$ is a "strategy".
	This allows to define a "morphism" $\phi:T \to S$ by proceeding top-down: we set $\phi(t_0) = s_0 = (t_0,m_0)$, and assuming $\phi(t) = (t,m)$ is defined and $t \re c t' \in E(T)$ we let $\phi(t')=(t',m')$ with $(t,m) \re c (t',m') \in E(S)$.
	Since $\vale(S)=\vale(\G)$, it holds that $\phi$ "preserves the value" at the root; moreover, note that the image of $\phi$ is included in $V(T) \times M \subseteq V(S)$.
	
	Observe that for each $(t,m) \in V(S)$ with $t \in V(T)$, and each $A \ni t$, since $t \re \eps A \in \Edges{G}$ and $t \in \VA$, the edge $(t,m) \re \eps (A,m)$ belongs to $E(S)$.
	Moreover, for each $(A,m) \in V(S)$, with $A \in \powne(V(T))$ there is an element $t_{(A,m)} \in A$ such that $(A,m) \re \eps (t_{(A,m)},m) \in E(S)$; we fix such a $t_{(A,m)}$ for each $(A,m)$.
	Combining these two observations, we have for each $(t,m) \in V(S)$ with $t \in V(T)$ and each $A \ni t$, the edges
	\[
	(t,m) \re \eps A \re \eps (t_{(A,m)},m)
	\]
	in $E(S)$.
	
	We now let $U^{(0)}$ be the "graph" over $V(U^{(0)}) = V(S) \cap (V(T) \times M)$ given by
	\[
	E(U^{(0)}) = E(S) \cap \big[V(U^{(0)}) \times C \times V(U^{(0)})\big] \cup \{(t,m) \re \eps (t_{A,m},m) \mid t \in A\}.
	\]
	In words, the graph $U^{(0)}$ is obtained by first restricting $S$ to $V(T) \times M$, and then adding all edges $(t,m) \re \eps (t_{(A,m)},m)$.
	Note that $\phi: T \to S$ defined above restricts to a "morphism" $\phi^{(0)} : T \to U^{(0)}$.
	Moreover, any "path" $\pi$ from $s_0$ in $U^{(0)}$ can be turned to a "path" $\pi'$ from $s_0$ in $S$ by replacing each occurrence of edges $(t,m) \re \eps (t_{(A,m)},m)$ by $(t,m) \re \eps (A,m) \re \eps (t_{(A,m)},m)$.
	Since the "path" $\pi'$ obtained in this way has the same "value@@graph" as $\pi$, we have $\val_{U^{(0)}}(s_0) \leq \val_{S}(s_0) = \val_T(t_0)$; stated differently $\phi^{(0)}$ "preserves the value" at the "root".
	Since it is the case in $S$, and we added only $\eps$-edges which preserve the "memory state" $m$, it holds that $(t,m) \re \eps (t',m') \in E(U^{(0)})$ implies $m=m'$.
	
	Note that for each $(t,m) \in V(U^{(0)})$ it must be that $t_{(\{t\},m)} = t$ (since by definition $t_{(A,m)} \in A$), and thus there is a loop $(t,m) \re \eps (t,m) \in E(U^{0})$.
	We then let $U^{(1)}$ be given by $V(U^{(1)}) = V(U^{(0)})$ and
	\[
	E(U^{(1)}) = \{u \re c u' \mid \exists v,v' \in V(U^{(0)}), u \rp{\eps^*} v \re c v' \rp{\eps^*} u' \tin U^{(0)}\},
	\]
	where the notation $x \rp{\eps^*} y$ means that there exists a "path" of $\eps$-edges from $x$ to $y$.
	By the observation above, it holds that $E(U^{(0)}) \subseteq E(U^{(1)})$ or stated differently the identity is a "morphism" from $U^{(0)}$ to $U^{(1)}$; we thus obtain a "morphism" $\phi^{(1)} : T \to U^{(1)}$ by composition.
	We now argue that $\phi^{(1)}$ "preserves the value" at the root: any "path" $\pi$ from $s_0$ in $U^{(1)}$ can be transformed into a "path" $\pi'$ in $U^{(0)}$ with same "value@@graph" by replacing occurrences of $u \re c u'$ by $u \rp{\eps^*} v \re c v' \rp{\eps^*} u'$, thus $\val_{U^{(1)}}(s_0) \leq \val_{U^{(0)}}(s_0) \leq \val_T(t_0)$.
	Moreover, $\eps$-edges in $U^{(1)}$ cannot modify the "memory state" $m$ since this is the case of $\eps^*$-paths in $U^{(0)}$.
	
	Observe now that it holds that $u \re \eps v \re c v' \re \eps u'$ in $U^{(1)}$ implies $u \re c u' \in E(U^{(1)})$.
	Applying to $c = \eps$ gives transitivity of $\re \eps$.
	Moreover, defining the partition of $V(U^{(1)})$ by $(V^{(1)}_m)_{m \in M}$ with $V^{(1)}_m = V(U^{(1)}) \cap (V(T) \times \{m\})$, we have that for each $m \in M$ and each non-empty subset $A \times \{m\}$ of $V^{(1)}_m$, for each $(t,m) \in A \times \{m\}$ there is an $\eps$-edge in $E(U^{(1)})$ towards $(t_{(A,m)},m)$.
	This implies that $\re \eps$ induces a "well-founded" total preorder over $V^{(1)}_m$, satisfying the monotonicity axiom.
	
	The only remaining caveat is that $\re \eps$ is not necessarily antisymmetric over $V(U^{(1)})$.
	However, in $U^{(1)}$, vertices $v,v'$ such that both $v \re \eps v'$ and $v' \re \eps v$ have the same incoming and outgoing edges.
	Defining such vertices to be $\sim$-equivalent, we thus let $U^{(2)}$ be given over $V(U^{(2)}) = V(U^{(1)}) / \sim$ by
	\[
	E(U^{(2)}) = \{[v] \re c [v'] \mid v \re c v' \in E(U^{(1)})\},
	\]
	where $[v]$ is the $\sim$-class of $v$; note that this is well defined since $v \re c v' \in E(U^{(1)})$ does not depend on the choices of representatives $v$ and $v'$ in $[v]$ and $[v']$.
	It is easy to verify that the "morphism" $v \mapsto [v]$ "preserves all values" from $U^{(1)}$ to $U^{(2)}$, and that $U^{(2)}$ is an "$\eps$-separated" monotone graph of "width" $\mu$, with the partition $(V^{(2)}_m)_{m \in M}$ defined by $V^{(2)}_m = V^{(1)}_m/\sim$.
	This concludes the proof in the "non-chromatic@@memory" setting.
	
	For the "chromatic@@memory" case, there remains to verify that $U^{(2)}$ is "chromatic@@epsGraph".
	We let $\delta: M\times C \to M$ be the "update function" of $\S$. Let $[u] \re c [u'] \in E(U^{(2)})$, we will show that $\delta$ witnesses the fact that $U^{(2)}$ is chromatic.
	Unraveling the definitions, we obtain that $u \re c u' \in E(U^{(1)})$, and in turn $u \rp {\eps^*} v \re c v' \rp{\eps^*} u'$ in $U^{(0)}$ for some $v,v' \in V(U^{(0)})$.
	Since in $U^{(0)}$, $\eps$-edges preserve the "memory state", we get that $u$ and $v$, as well as $u'$ and $v'$ have the same "memory state"; let us write them $m$ and $m'$. We aim to show that $m' = \delta(m,c)$.
	If $c = \eps$, there is nothing to prove, we already know that $\eps$-edges preserve the "memory state" in $U^{(2)}$.
	Otherwise, by definition of $U^{(0)}$ we get that $v \re c v' \in E(S)$, which yields $m'=\delta(m,c)$ as required.
\end{proof}

\subsection{Further results}\label{sec:further_results}

Before going to applications in subsequent sections, we prove a few further general results that are useful for constructing "universal" graphs.
We start by proving (Section~\ref{sec:finitely_bounded_antichains}) that in the case of memory $\leq m$ for some finite $m \in \N$, and with some further technical assumptions, our two notions of universal structures ("well-monotone" graphs with bounded "antichains" on one hand, and "$\eps$-separated" "well-monotone" graphs with bounded "breadth") collapse.

We proceed to show how our definitions instantiate in the important special cases of "prefix-increasing" (Section~\ref{sec:prefix_increasing_objectives}) and "prefix-independent" (Section~\ref{sec:prefix_independent_objectives}) "objectives" (these are defined later).
Last, we show (Section~\ref{sec:almost_universality}) how the convenient notion of "almost universality" (which serves as a lever for deriving universality results) from~\cite{Ohlmann23UnivJournal} adapts to the setting at hands.

We urge the reader to jump to Section~\ref{sec:examples} and come back to~\ref{sec:further_results} when required.

\subsubsection{Finitely bounded antichains determine the $\ee$-memory}
\label{sec:finitely_bounded_antichains}

Dilworth's Theorem (c.f. Appendix~\ref{sec:appendix_set_theory}) states that if the size of the "antichains" of
an "ordered set" $(P,\leq)$ is bounded by a finite number $k$, then $P$ can be decomposed in $k$ disjoint "chains"~\cite{Dilworth50}.
Therefore (assuming "well-foundedness" of the set of values), this allows to construct "$\eps$-separated" "universal" structures from arbitrary "monotone" ones, whenever we have a finite bound on the "width".

\begin{prop}\label{prop:non_eps-implies-eps}
	Let $\val: C^\omega \to X$ be a "valuation", and $m \in \N$; we further assume that $X$ is "well-founded".
	If for all "cardinals" $\kappa$ there exists a "well-monotone" graph which is "$(\kappa,\val)$-universal" and has "width" $\leq m$, then for all "cardinals" $\kappa$ there is also an "$\eps$-separated" "well-monotone" "$(\kappa,\vale)$-universal" graph of "breadth" $\leq m$, and therefore $\val$ has "$\eps$-memory" $\leq m$.
\end{prop}

Unfortunately, proving this proposition requires dealing with some slight technical complications arising from creation of sinks when contracting $\eps$'s in an infinite tree.
This is what leads to the assumption that $X$ is "well-founded", we do not know whether it can be dropped.
Note however that "objectives" are "valuations" with $X=\{\bot,\top\}$, which is "well-founded", and moreover many other interesting examples of valuations have well-founded sets of values (for instance, energy valuations over $\N$).

Proposition~\ref{prop:non_eps-implies-eps} is very useful in practice (see examples in Section~\ref{sec:examples}) for establishing finite "$\eps$-memory": it suffices to construct universal structures with bounded "width", which is often easier in practice than "$\eps$-separated" structures.
One can also see the result in a negative light: for finite bounds (for instance, $\omega$-regular objectives), one cannot use Theorem~\ref{thm:implication_non_eps} to derive "$\eps$-free" "memory@@epsFree" upper bounds smaller than the "$\eps$-memory".

\begin{proof}
	Let $(G,\leq)$ be a "well-monotone" "($\kappa,\val$)-universal" "$C$-graph" of "width" $\leq m$.
	Applying Dilworth's Theorem yields a partition of $V(G)$ into $(V_j)_{j \in m}$ so that the restriction of $\leq$ to each $V_j$ is a "total order".
	We let $G^\eps$ be the graph over $V(G)$ defined by adding $\re \eps$'s according to this decomposition, that is,
	\[
		E(G^\eps) = E(G) \cup \{v \re \eps v' \mid v \geq v' \tin G \tand \exists j \in m \text{ such that } v,v' \in V_j\}.
	\]
	Note that $G^\eps$ is indeed an "$\eps$-separated" "monotone" graph over $m$, as required.
	We first prove that values in $G^\eps$ are the same as in $G$, that is, for any $v \in V(G)= V(G^\eps)$ it holds that
	\[
		\val_G(v) = \vale_{G^\eps}(v).
	\]
	We remark that $\val_G(v) \leq \vale_{G^\eps}(v)$, since $G$ is a subgraph of $G^\eps$. For the other inequality, let $v \in V(G)$ and consider a path $\pi: v_0 \re {c_0} v_1 \re{c_1} \dots$ from $v=v_0$ in $G^\eps$, our aim is to construct a path from $v$ in $G$ with "value@graph" larger than $\pi$; for this we proceed in two steps.
	First, we replace in $\pi$ any block of the form
	\[
		v_i \re {\eps} v_{i+1} \re{\eps} \dots \re{\eps} v_{j-1} \re {c} v_j,
	\]
	where $c \in C$, by
	\[
		v_i \re {c} v_j.
	\]
	This does not increase the $\vale$-value by definition, and yields a path $\pi'$ in $G^\eps$ by "monotonicity".
	Now if the original path $\pi$ had infinitely many occurrences of colours in $C$, we are done; otherwise $\pi'$ is of the form $\pi'_0 \pi'_1$, where $\pi'_0$ is a finite path avoiding $\eps$-edges whereas $\pi'_1$ is an infinite path comprised only of $\eps$-edges.
	Note that $\pi'_0$ is thus a finite path from $v$ in $G$, let $v'$ denote its endpoint.
	Now append to $\pi'_0$ any infinite path starting from $v'$ in $G$, which yields a path $\pi''$ in $G$ with value $\geq \vale(\pi)$, by definition of $\vale$.

	We now proceed to proving "$(\kappa,\vale)$-universality" of $G^\eps$: let $T^\eps$ be a "$C^\eps$-tree" of cardinality $< \kappa$ and let $t_0 \in V(T^\eps)$ denote its root.
	We first remove $\re \eps$'s from $T^\eps$ by contracting them, formally we let $T$ be the "$C$-pretree" given over
	\[
		V(T) = \{t \in V(T^\eps) \mid \text{ the unique path from $t_0$ to $t$ in $T^\eps$ does not end with an $\eps$-edge}\}
	\]
	by
	\[
		E(T) = \{t \re c t' \mid t \rp {\eps^*} t'' \re c t' \tin T\}.
	\]
	Note that $T$ is rooted at $t_0 \in V(T)$, and that there may be sinks in $T$, namely, the vertices from which all paths visit only $\eps$-edges in $T^\eps$; let
	\[
		\begin{array}{lcl}
		S &=& \{t \in V(T) \mid t \text{ is a sink in } T\} \\
		&=& \{t \in V(T) \mid \text{ all paths from $t$ in $T^\eps$ see only $\eps$-edges} \}.
		\end{array}
	\]
	For each $s \in S$, let $u_s \in C^*$ be the coloration of the unique path from $t_0$ to $s$ in $T$, and let $w_s \in C^\omega$ be an infinite word such that
	\[
		\vale(u_s \eps^\omega) = \val(u_s w_s),
	\]
	whose existence is guaranteed by "well-foundedness" of $X$ and the definition of $\vale$.

	We then append to each sink $s \in S$ an infinite path with label $w_s$, formally we let $T'$ be the $C$-tree over
	\[
	V(T') = V(T) \cup (S \times \NN)
	\]
	given by
	\[
		E(T') = E(T) \cup \{(s,i) \re{w_{s,i}} (s,i+1) \mid s \in S \tand i \in \N\},
	\]
	where it is understood that $(s,0)=s$ and we write $w_s=w_{s,0} w_{s,1} \dots$.
	By construction, we get that $\val_{T'}(t_0) = \vale_{T^\eps}(t_0)$; moreover, $T'$ has cardinality $< \kappa$ (unless $\kappa$ is finite, in which case there is no tree with cardinality $< \kappa$ and the proof is vacuous).
	There is a "morphism" $\phi': T' \to G$ "preserving the value" at the "root" by "($\kappa,\val$)-universality" of $G$.

	Finally, we define a map $\phi^\eps : T^\eps \to G^\eps$ by letting $\phi^\eps(t) = \phi(t')$, where $t'$ is the unique vertex in $V(T)$ such that $t' \rp{\eps^*} t$ is a path in $T^\eps$.
	It is a direct check that $\phi^\eps$ is a "morphism", since $G^\eps$ includes $\eps$-loops around all vertices.
\end{proof}

\subsubsection{The case of prefix-increasing objectives}
\label{sec:prefix_increasing_objectives}

\AP A "$C$-valuation" $\val$ is ""prefix-increasing"" (resp. prefix-decreasing) if adding a prefix can only increase (resp. decrease) values, meaning that for all $u \in C^*$ and $w \in C^\omega$ we have $\val(uw) \geq \val(w)$ (resp. $\val(uw) \leq \val(w)$).
\AP We say that $\val$ is ""prefix-independent"" if it is both "prefix-increasing" and prefix-decreasing, that is, for all $u \in C^*$ and $w \in C^\omega$, $\val(uw)=\val(w)$.
An "objective" $W$ is thus "prefix-increasing" (resp. deacreasing, independent) if for all $c\in C$, $cW \supseteq W$ (resp. $\subseteq$, $=$).

Just as in~\cite{Ohlmann23UnivJournal}, we may simplify the notions under study when the "objective" has such properties. First, note that for a "prefix-increasing" "objective" $W$ and a "tree" $T$, it is equivalent that the "root" of $T$ "satisfies" $W$, and that $T$ itself (meaning, all vertices in $T$) "satisfies" $W$. 

Now fix a "prefix-increasing" objective $W \subseteq C^\omega$ and consider a "well-monotone" graph~$U$.
Consider moreover the "restriction@@graph" $U'$ of $U$ to vertices which "satisfy" $W$ (note that $U$ is "well-monotone", as is any restriction of a "well-monotone" graph).
Last, let $U^\top$ be the "well-monotone" graph obtained from $U'$ by appending an additional fresh vertex $\top$, with all possible outgoing edges (and only incoming edges from itself); formally $V(U^\top) = V(U') \sqcup \{\top\}$ and $E(U^\top) = E(U') \cup \{\top\} \times C \times V(U^\top)$. 
The following lemma states that the (hypothetical) "universality" of $U$ transfers to $U^\top$.

\begin{lemC}[{\cite[Lemma~3.9]{Ohlmann23UnivJournal}}]
    Let $\kappa$ be a "cardinal".
    The following conditions are equivalent:
    \begin{enumerate}[(i)]
        \item $U$ is "$(\kappa,W)$-universal";
        \item $U^\top$ is "$(\kappa,W)$-universal";
        \item all "$C$-trees" of cardinality $< \kappa$ "satisfying" $W$ have a "morphism" into $U'$.
    \end{enumerate}
\end{lemC}

Intuitively, the lemma states that in the case of a "prefix-increasing" "objective" and when looking for a universal structure, vertices which do not satisfy the "objective" are irrelevant, and can simply be replaced by $\top$.
Observe moreover that "antichains" are not larger in $U'$ or $U^\top$ than they are in the original graph $U$.

\AP In this way, we can simplify without loss of generality the definition of universality when dealing with "prefix-increasing" "objectives". In the remainder of the paper, if $W$ is a "prefix-increasing" "objective", we will say that a graph $U$ is ""$(\kappa,W)$-universal for prefix-increasing objectives"" if:
\begin{itemize}
	\item $U$ "satisfies" $W$; and
	\item it embeds all trees of cardinality $< \kappa$ that "satisfy" $W$.
\end{itemize}
When it is clear from the context that $W$ is "prefix-increasing", we will just say "$(\kappa,W)$-universal@@prefixIncreasing".

That is, we may always disregard vertices of universal graphs not satisfying the objective under consideration. We note that the definition of universality that we have just given coincides with the one introduced (for "prefix-independent" objectives) by Colcombet and Fijalkow~\cite{CF18}.

\subsubsection{The case of prefix-independent objectives}
\label{sec:prefix_independent_objectives}

Recall that an "objective" $W$ is "prefix-independent" if for all $u \in C^*$ and $w \in C^\omega$,
\[
	uw\in W \Leftrightarrow w\in W.
\]
When dealing with "prefix-independent" objectives, it is often more natural to consider "pretrees", which leads to a stronger definition of universality that may lend itself better to inductive arguments (see for example Sections~\ref{sec:muller_games} and~\ref{sec:lexicographic_products}).
\AP We say that a vertex in a "pregraph" ""satisfies@@pregraph"" an "objective" if all infinite paths from the vertex "satisfy" the "objective" (regardless of finite paths), and that a pregraph \emph{satisfies} an "objective" if all its vertices do.
This may be unsatisfactory for modelisation purposes, for instance, in the case of a safety condition, since this definition allows for non-safe finite paths; however it poses no issue in the context of "prefix-independent" objectives for which finite paths are indeed irrelevant.

\AP Given a "prefix-independent" "objective" $W$, we say that a "graph" $U$ is ""$(\kappa,W)$-universal for prefix-independent objectives"" if
\begin{itemize}
		\item $U$ "satisfies" $W$; and
		\item $U$ embeds all "pretrees" of cardinality $< \kappa$ that "satisfy@@pregraph" $W$.
\end{itemize}
When it is clear from the context that $W$ is "prefix-independent", we will just say that $U$ is "$(\kappa,W)$-universal@@prefixIndependent".

We prove that for "prefix-independent" objectives, this stronger definition of universality can in fact be used without loss of generality. First, we remark that as "prefix-independent" objectives are a special case of "prefix-increasing" ones, all remarks from the previous subsection apply.

\begin{lem}
Let $W \subseteq C^\omega$ be a nonempty "prefix-independent" "objective", let $U$ be a "$C$-pregraph" and let $\kappa$ be an infinite "cardinal".
The following are equivalent:
\begin{enumerate}[(i)]
    \item all "trees" of cardinality $<\kappa$ which "satisfy" $W$ embed in $U$;
    \item all "pretrees" of cardinality $<\kappa$ which "satisfy@@pregraph" $W$ embed in $U$.
\end{enumerate}
\end{lem}

\begin{proof}
    The implication $(ii) \implies (i)$ is trivial and therefore we concentrate on the other one.
    Fix an infinite word $w=w_0w_1\dots \in W$ and consider a pretree $T'$ of cardinality $< \kappa$ which "satisfies" $W$.
    Let $S \subseteq V(T')$ be the set of sinks in $T'$.
    Now let $T$ be the tree obtained by appending a path labelled with $w$ to all sinks in $T'$, formally, $V(T)=V(T') \cup (S \times \NN)$, and
    \[
        E(T) = E(T') \cup \{(s,i) \re {w_i} (s,i+1) \mid s \in S \tand i \in \N\};
    \]
    where it is understood that we identify $(s,0)$ with $s$ for all $s \in S$.
    Paths in $T$ are either paths in $T'$, or their label end with $w$; thus $T$ "satisfies" $W$ by "prefix-independence".
    Thus there is a "morphism" $T \to U$, whose restriction to $V(T')$ is then a "morphism" $T' \to U$, and the lemma is proved.
\end{proof}

\subsubsection{Almost universality}
\label{sec:almost_universality}

In this section, we show how the technically convenient notion of almost universality defined by Ohlmann~\cite{Ohlmann23UnivJournal} adapts to our setting.
Recall that $G[v]$ denotes the "restriction@@graph" of $G$ to vertices reachable from $v$.

\AP For a "prefix-independent" "objective" $W$, we say that a "graph" $U$ is ""almost $(\kappa,W)$-universal"" if
\begin{itemize}
    \item $U$ "satisfies" $W$; and
    \item all "pretrees" $T$ "satisfying@@pregraph" $W$ have a vertex $t$ such that $\treerooted{T}{t} \to U$.
\end{itemize}

The following technical result allows us to build "well-monotone" "universal@@prefixIndependent" "graphs" from "almost universal" graphs, without any blowup on the size of "antichains".
Given a "well-monotone" graph $U$ and an "ordinal" $\alpha$, we let\footnote{Using the vocabulary from Section~\ref{sec:lexicographic_products}, $U \ltimes \alpha$ is the "lexicographic product@@graphs" of $U$ and the edgeless pregraph over $\alpha$; this explains the common notation.} $U \ltimes \alpha$ be the "well-monotone" graph given by $V(U \ltimes \alpha) = V(U) \times \alpha$ and
\[
E(U \ltimes \alpha) = \{(v, \lambda) \re c (v', \lambda') \mid \lambda > \lambda' \tor [\lambda = \lambda' \tand v \re c v' \in E(U)]\};
\]
it is illustrated in Figure~\ref{fig:repeat_graph}.
\begin{figure}[h]
\begin{center}
\includegraphics[width = 1 \linewidth]{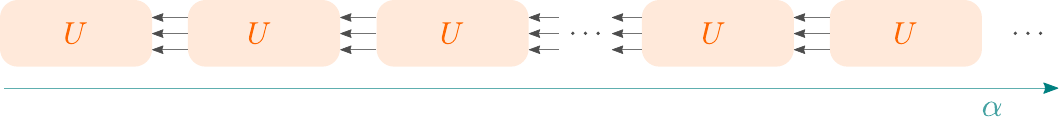}
\end{center}
\caption{An illustration of the graph $U \ltimes \alpha$.}
\label{fig:repeat_graph}
\end{figure}

\begin{lem}\label{lem:rongeur_de_croute}
Let $W$ be a "prefix-independent" "objective", $\kappa$ a "cardinal", and assume that $U$ is "almost $(\kappa,W)$-universal".
Then $U \ltimes \kappa$ is "$(\kappa,W)$-universal@@prefixIndependent" (for "prefix-independent" objectives).
\end{lem}

The proof is directly adapted from~\cite[Lemma~4.5]{Ohlmann23UnivJournal} to this setting.

\begin{proof}
Consider an infinite path $(u_0,\lambda_0) \re {c_0} (u_1,\lambda_1) \re {c_1} \dots$ in $U \ltimes \kappa$.
Since $\lambda_0 \geq \lambda_1 \geq \dots$, it must be that this sequence is eventually constant by "well-foundedness".
Therefore, some suffix $u_i \re {c_i} u_{i+1} \re{c_{i+1}} \dots$ defines a path in some copy of $U$, which implies that $c_i c_{i+1} \dots \in W$.
We conclude by "prefix independence" that $U \ltimes \kappa$ indeed "satisfies" $W$.

Let $T$ be a tree of cardinality $\kappa$ which "satisfies" $W$.
We construct by "transfinite recursion" an "ordinal" sequence of vertices $\{v_\alpha\}_{\alpha < \lambda_0} \in V(T)$ (for some $\lambda_0 < \kappa$) where for each $\beta < \lambda$, $v_\lambda$ is not reachable from $v_\beta$ in $T$, together with a "morphism" $\phi_\lambda: T_\lambda \to U$, where $T_\lambda$ is the restriction of $T$ to vertices reachable from $v_\lambda$ but not from $v_\beta$ for $\beta < \lambda$.

Assuming the $v_\beta$'s for $\beta < \lambda$ already constructed (this assumption is vacuous for the base case $\lambda=0$), there are two cases.
If all vertices in $T$ are reachable from some $v_\beta$, then the process stops.
Otherwise, we let $T_{\geq \lambda}$ be the "restriction@@graph" of $T$ to vertices not reachable from any $v_\beta$ for $\beta < \lambda$.
It is a "pretree" of cardinality $< \kappa$. By "almost $(\kappa,W)$-universality" of $U$, there exists some $t\in T_{\geq \lambda}$ such that $\treerooted{T_{\geq \lambda}}{t}$ has a "morphism" towards $U$. We let $v_\lambda= t$ and $\phi_\lambda$ be this "morphism".

Since all the $T_\lambda$'s are nonempty, the process must terminate in $\lambda_0$ steps for some "ordinal" $\lambda_0$ satisfying $\lambda_0 \leq \card{V(T)} < \kappa$.
Now observe that any edge in $T$ is either from $T_\beta$ to itself, for some $\beta \leq \lambda_0 < \kappa$, or from $T_\beta$ to $T_{\beta'}$ for $\beta' < \beta \leq \lambda_0< \kappa$.
This proves that the map $\phi : V(T) \to V(U \ltimes \kappa)$ defined by $\phi(v)=(\phi_\lambda(v),\lambda)$, where $\lambda$ is so that $v \in V(T_\lambda)$, is a "morphism" from $T$ to $U\ltimes\kappa$.
\end{proof}

\section{Examples}
\label{sec:examples}

In this section we show how Theorems~\ref{thm:characterisation_eps} and~\ref{thm:implication_non_eps} can provide upper bounds on the "memory" of different "objectives" by constructing "well-monotone" "universal" "graphs".
In general, proving tight bounds for the "memory" of "objectives" is a hard task, and only the memory of a few classes of "objectives" has been characterised, notably, for "topologically closed" "objectives"~\cite{CFH14} and "Muller objectives"~\cite{DJW97}.

As a warm-up and to illustrate our tool, we start (Section~\ref{sec:other-examples}) with a few concrete examples.
We then turn our focus to "topologically closed objectives" (Section~\ref{sec:generalised_safety}) for which we derive a variant of the result of~\cite{CFH14}.
Finally, we show how the upper bound of~\cite{DJW97} for the "memory@@eps" of "Muller objectives" can be understood in our framework (Section~\ref{sec:muller_games}).

In Table~\ref{table:examples}, we compile the examples appearing throughout the paper and their "exact memory" requirements for the different notions of memory that we consider.
\AP For an infinite word $w\in C^\omega$ we write $\intro*\minf(w)=\{ c\in C \mid w_i=c \text{ for infinitely many } i \}$.
\AP For a word $u\in C^*$, we write $\intro*\infOften(u) = \{w\in C^\oo \mid w \text{ contains infinitely many factors } u\}$ and we let $\intro*\finOften(u)$ denote its complement.
We let $C^{\geq n} = C^n C^*$.

\newcolumntype{M}[1]{>{\centering\arraybackslash}m{#1}} %Center content and set a fixed width. Also, replace p with m
\newcolumntype{N}{@{}m{0pt}@{}}
\renewcommand{\thefootnote}{\fnsymbol{footnote}}
\newcommand{\mycolwidth}{1.3cm}
%\hdashline[10pt/2pt] %Instead of \hline, if dashed wanted
\begin{table}[h]
	\begin{tabular}{|M{3.2cm}||M{1.7cm}|M{2.1cm}|M{1.45cm} M{2.15cm}|M{2.1cm}|N} 
		\hline
		%\diagbox[innerwidth=31mm]{Examples}{Memories} &
		\textbf{Objective} & \begin{tabular}{@{}c@{}}"$\eps$-free" \\ "memory@@epsFree"\end{tabular}  & "$\eps$-memory" & \multicolumn{1}{c|}{\begin{tabular}{@{}c@{}}"$\eps$-free" \\ "chromatic@@memory"\end{tabular}}  & "$\eps$-chromatic" & \begin{tabular}{@{}c@{}}Minimal \\ det. parity  \\ automaton\end{tabular} \\[3mm]
		\hline
		\hline
		
$\infOften(a) \cap \infOften(b)$  &    2    & 2     & \multicolumn{1}{M{1.3cm}|}{2}         & 2 & 2       &\\[3mm]
		\hline

$\forall i, \, w_i \neq w_{i+1}$  & 2 (Prop.~\ref{prop:eps-memory-greater-than-memory})   &   $\card{C}$ &   \multicolumn{1}{M{\mycolwidth}|}{$\card{C}$\footnotemark[3]}   &  $\card{C}$\footnotemark[3] &  $\card{C}+2$\footnotemark[2]  & \\[3mm]
\hline

$(C^*a)^mC^{\geq n} a C^\oo$	&  $n+1$  &   $n+1$  &   \multicolumn{1}{M{\mycolwidth}|}{$n+1$}  &  $n+1$  &  $m + n + 2$\footnotemark[2] &\\[3mm]
\hline

$\infOften(bb) \; \cup $ \newline $( \finOften(b) \cap \finOften(aa))$	&   2  & 2 &  \multicolumn{1}{M{\mycolwidth}|}{2}   &  2   &   3\footnotemark[3]          &\\  
	\hline
	
$|\minf(w)| = 2$  & 2   &   2  \cite{DJW97}    &    \multicolumn{1}{M{\mycolwidth}|}{$\card{C}$ \cite{Casares22}}   & $\card{C}$ \cite{Casares22} & $\card{C}(\card{C}+1)$ \cite{Casares22, CCFL24FromMtoP}  & \\[3mm] \hline

%$C^*a^n C^\oo$	&  2       &   $n$    &  $n$  & $n$   &   $n+1$  &\\[3mm]
%\hline
	
Topologically closed objectives  (Section~\ref{sec:generalised_safety})	&   Unknown     &  "Width" of "left quotients" \cite{CFH14}    &   \multicolumn{1}{M{\mycolwidth}|}{Unknown}       &  $\mathtt{NP}$-complete\footnotemark[5] \cite{BFRV23Regular}   &   "Left quotients"\footnotemark[2] &\\[3mm]
\hline

Muller objective $\F$ \newline
	(Section~\ref{sec:muller_games})	&   Unknown     &  $\memory(\F)$ \cite{DJW97}   &  \makecell{\hspace{4mm} $\mathtt{NP}$-complete\footnotemark[5] \\
	\hspace{1mm}Both notions coincide\\ \cite{Casares22}}        &    &     Leaves of the "Zielonka tree" \cite{CCFL24FromMtoP}       &\\[3mm]
		\hline                      
	\end{tabular}
\caption{Examples of objectives appearing in the paper and their memory requirements. We also include sizes of minimal parity automata, which give upper bounds to the $\eps$-chromatic memory.
}
\label{table:examples}

\end{table}

\subsection{Concrete objectives}
\label{sec:other-examples}

We start by illustrating the notions presented until now and some methods to derive "universality" proofs with a few simple concrete examples of "objectives".
%We fix the set of colours to be $C=\{a,b\}$.

\paragraph*{Objective $W_1  = \{ w\in \{a,b\}^\oo \mid \normalfont{ a \text{ and } b \text{ occur infinitely often in } w }\} = \infOften(a) \cap \infOften(b) $.}

Objective $W_1$ is an example of a "Muller objective" ($W_1 = \Muller{\{a,b\}}$; see Section~\ref{sec:muller_games} for details).
It is known that its "$\eps$-memory" is exactly $2$~\cite{DJW97}.
We show, for each "cardinal" $\kappa$, an "$\eps$-separated" "chromatic@@graph" and "well-monotone" "$(\kappa, W_1^\eps)$-universal@@prefixIndependent" graph of "breadth" $2$.
(Since $W_1$ is "prefix-independent", we use the corresponding notion of "universality@@prefixIndependent", from Section~\ref{sec:prefix_independent_objectives}).
By Theorem~\ref{thm:characterisation_eps}, this implies that the "$\eps$-chromatic memory" of $W_1$ is exactly $2$.
 
Fix a "cardinal" number $\kappa$ and consider the graph $U$ from the left hand side of Figure~\ref{fig:univGraph-muller-a-b}.
It is easy to check that $U$ is an "$\eps$-separated monotone graph over" the set $M = \{a,b\}$ and that it is indeed "chromatic@@graph" and "satisfies" $W$.
We sketch a "universality" proof; formal details are given for general "Muller objectives" in Section~\ref{sec:muller_games}.

%Move footnotes some paragraphs below if we want them to appear in the same page as the table
\footnotetext[2]{Since these objectives are topologically closed or topologically open, they can be recognised by a weak automaton, and the size of a minimal deterministic parity (resp. weak) automaton recognising them is given by the number of "left quotients" of the objective.}

\footnotetext[3]{The proof of these claims can be found in Appendix~\ref{sec:appendix_bounds_memory}.}

\footnotetext[5]{When the objective $W$ is $\oo$-regular (it has a finite number of "left quotients") the decision problem is: given a deterministic parity automaton recognising $W$ and $k\in \NN$, decide whether the "$\eps$-chromatic memory" of $W$ is $\leq k$.}
\renewcommand{\thefootnote}{\arabic{footnote}}

\begin{figure}[h]
	\begin{center}
		\includegraphics[width =  \linewidth]{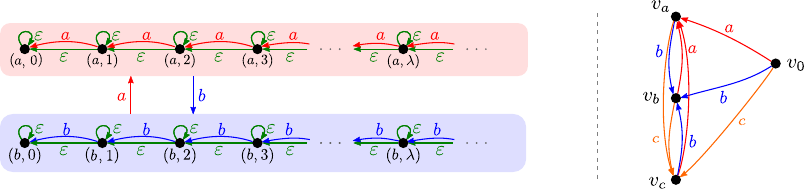}
	\end{center}
	\caption{Universal graphs for $W_1$ (on the left) and $W_2$ (on the right, for $C=\{a,b,c\}$). For the graph on the left (as required by the definition of $\eps$-separated graphs), the order coincides with $\re \eps$; on the right, it is given by $v_0 \geq v_a,v_b,v_c$ and $v_a,v_b,v_c$ incomparable.
	Edges following from monotonicity are not represented. An edge between boxes indicates that all edges are put between vertices in the respective boxes.}\label{fig:univGraph-muller-a-b}
\end{figure}

Let $T$ be a "$C$-tree" of "size@graph" $< \kappa$ which "satisfies" $W$, and let $t_0$ be its root.
Note that all paths from $t_0$ eventually visit a $b$-edge; there is in fact an "ordinal" $\lambda_0<\kappa$ (defined by "induction@@transfinite") which counts the maximal amount of $a$-edges seen from $t_0$ before a $b$-edge is seen; we set $\phi(t_0)$ to be $(a,\lambda_0)$.

Then for each edge $t_0 \re c t \in E(T)$ we proceed as follows.
\begin{itemize}
	\item If $c \in \{a, \eps\}$, we iterate exactly the same process on $t$, but the "ordinal" count on the number of $a$'s will have decreased (or even strictly decreased if $c=a$) from $t_0$ to $t$, which guarantees that $\phi(t_0) \re a \phi(t)$ is indeed an edge in $U$.
	\item If $c = b$, then we iterate the same process of $t$ but inverting the roles of $a$ and $b$; thus $\phi(t)$ is of the form $(b,\lambda_b)$ for some $\lambda_b < \kappa$, and the edge $\phi(t_0) \re b \phi(t)$ belongs to $U$, as required.
\end{itemize}
This concludes the top-down construction of $\phi$ and the universality proof.

It is not difficult to find lower bounds to see that the "$\eps$-free" memory of $W_1$ (and therefore all the other notions of memory) is $\geq 2$. For example, a "game" with just one vertex controlled by "Eve" where she can choose to produce $a$ or $b$ provides this lower bound.
Therefore, the exact memory of $W_1$ is $2$, for all the different notions of memory.

%Surprisingly, if the colours of a "game" were put in the vertices instead of the edges, objective $W_1$ would be positional~\cite{Zielonka98}, that is, its $\eps$-free memory would be $1$.

\paragraph*{Objective $W_2 = \{w_0w_1w_2\dots \in C^\oo \mid \forall i,\, w_i \neq w_{i+1}\}$.}%\{ w\in C^\oo \mid w \normalfont{ \text{ does not contain two identical consecutive colours}}\}$.}

Note that $W_2$ is "prefix-increasing", and therefore we use the definition of "universality@@prefixIncreasing" from Section~\ref{sec:prefix_increasing_objectives}.
Consider the graph $U$ with vertices $V(U) = \{v_0\} \cup \{v_c \mid c \in C\}$ and edges
\[
	E(U) = \{v \re c v_c \mid v \in V(U),c \in C\};
\]
see right hand side of Figure~\ref{fig:univGraph-muller-a-b}.
With the order with maximal element $v_0$ and otherwise no comparable elements, the graph $U$ is well-monotone of width $|C|$.
We prove that it is $W_2$-universal, which implies, by Theorem~\ref{thm:implication_non_eps}, that the "$\eps$-memory" of $W_2$ is $\leq |C|$.
To do so, it suffices to remark that for any "tree" $T$ "satisfying@@graph" $W$, mapping the "root" to $v_0$ and every other node $t$ is to $v_c$, where $c$ is the colour of the unique edge towards $t$, defines a "morphism". 

Proposition~\ref{prop:non_eps-implies-eps} implies the existence of an "$\eps$-separated" "well-monotone" "$(\kappa, W_2^\eps)$-universal" graph of "breadth" $2$.
In fact, an "$\eps$-separated" "graph" given by the proof of Proposition~\ref{prop:non_eps-implies-eps} can be obtained just by adding $\eps$-edges $v_0 \re \eps v_c$ for $c \in C$ and $\eps$-loops over all vertices.
Since the graph obtained in this way is "chromatic@@graph", we get that the "$\eps$-chromatic memory" of $W_2$ is also $\leq |C|$.
Proposition~\ref{prop:eps-memory-greater-than-memory} below proves that the "$\eps$-memory" is exactly $|C|$, and that the "$\eps$-free" memory is in fact just $2$.

\paragraph*{Objective $W_{3} = (C^*a)^mC^{\geq n} a C^\omega$ with $C=\{a,b\}$ and $m,n\geq 1$.}
%The minimal automaton recognising the language $L= (C^*a)^mC^{\geq n} a$ has $n+m+2$ states;
We provide a "universal graph" of "width" $n+1$ which proves that the "$\eps$-memory" is $\leq n+1$.
A matching lower bound on the "$\eps$-free memory" follows from the "game" depicted on Figure~\ref{fig:w3_lowerbound}.
We remark that from the minimal automaton for the regular language $L= (C^*a)^mC^{\geq n} a$ we only obtain a straightforward upper bound of $n+m+1$ on the "memory".

\begin{figure}[h]
	\begin{center}
		\includegraphics[width=0.45\linewidth]{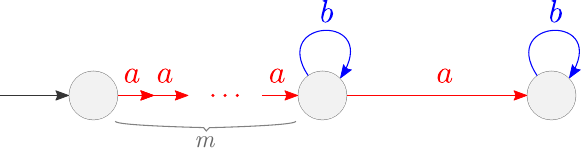}
	\end{center}
	\caption{A "game" where "Eve" requires memory $n+1$ to ensure "objective" $W_3$.}\label{fig:w3_lowerbound}
\end{figure}

Fix a "cardinal" $\kappa$, and consider the "graph" $U$ with vertex set
\[
	\Verts{U} = \{q_0,\dots,q_{m-1},p_{n}\} \times \kappa \cup \{p_0,\dots,p_{n-1}\} \cup \{\top\},
\]
and edges as in Figure~\ref{fig:constr_w3}.

\begin{figure}[h]
	\begin{center}
		\includegraphics[width=0.9\linewidth]{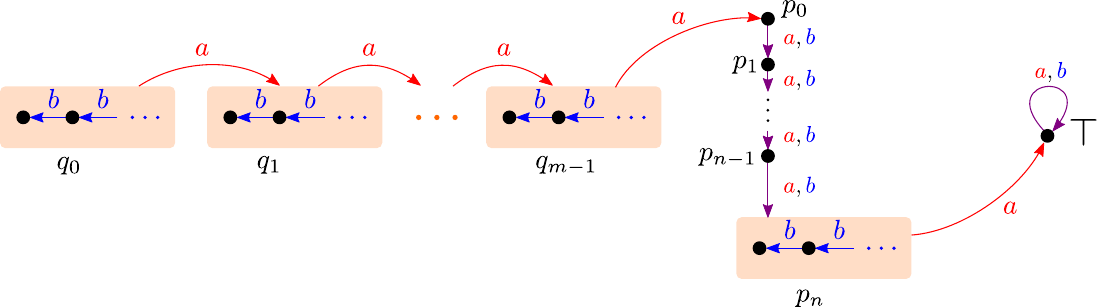}
	\end{center}
	\caption{The "well-monotone" "graph" $U$ which has "width" $n+1$ and is "universal" for $W_3$. Edges between boxes represent all possible edges between vertices from these boxes.
	For readability, second coordinates of vertices are not displayed.
	The order is as follows: $\top$ is maximal, the $p_i$'s are pairwise incomparable and greater than the $q_j$'s, vertices are ordered within boxes, and the $q_i$'s are ordered.
	Many edges that follow from "monotonicity" (for instances, $a$'s pointing down, and edges from $p_i$'s to $q_j$'s) are omitted for clarity.
	}\label{fig:constr_w3}
\end{figure}

Let us sketch a proof of "universality".
Observe that the vertices "satisfying" $W$ in $U$ are exactly those of the form $(q_0,\lambda)$.
Consider a tree $T$ whose root $t_0$ satisfies $W_3$; we aim to build a "morphism" $T \to U$ mapping $t_0$ to one of the $(q_0,\lambda)$'s.
Given a vertex $t \in V(T)$, let $w_t \in C^*$ denote the unique word labelling a path from $t_0$ to $t$.

A vertex $t \in V(T)$ such that $w_t$ has $j<m$ occurrences of $a$ is mapped to a vertex of the form $(q_j,\lambda) \in V(U)$, where $\lambda$ is an "ordinal" capturing the distance until the next $a$ in $T$.
Then a vertex $t \in V(T)$ such that $w_t$ is of the form $w_t=w'au$, where $w'$ has exactly $m-1$ occurrences of $a$ and $|u|=i < n$ is mapped to $p_i$.
A vertex $t \in V(T)$ as above with $|u| \geq n$ is mapped to a vertex of the form $(p_{n},\lambda)$, where $\lambda$ captures the distance to the next $a$ (which must occur since $t_0$ satisfies $W_3$).
Finally, remaining vertices $t \in V(T)$ "satisfy" $w_t \in (C^* a)^m C^{\geq n} a C^*$, and we map them to $\top$.
It is easy to verify that the map constructed above indeed defines a "morphism".

One may make the "graph" $U$ "$\eps$-separated" without blowing up its "width" (for instance, using Proposition~\ref{prop:non_eps-implies-eps}); however the obtained graph is not "chromatic@@graph".
Nevertheless, with a slightly more involved construction depicted in Figure~\ref{fig:constr_chromatic_w3}, we obtain a "chromatic@@graph" "$\eps$-separated" "graph" of "breadth" $n+1$, yielding an upper bound of $n+1$ also on the "$\eps$-chromatic memory".
We omit a proof of "universality" as it follows roughly the same lines as the one above. 

\begin{figure}[h]
	\begin{center}
		\includegraphics[width=0.9\linewidth]{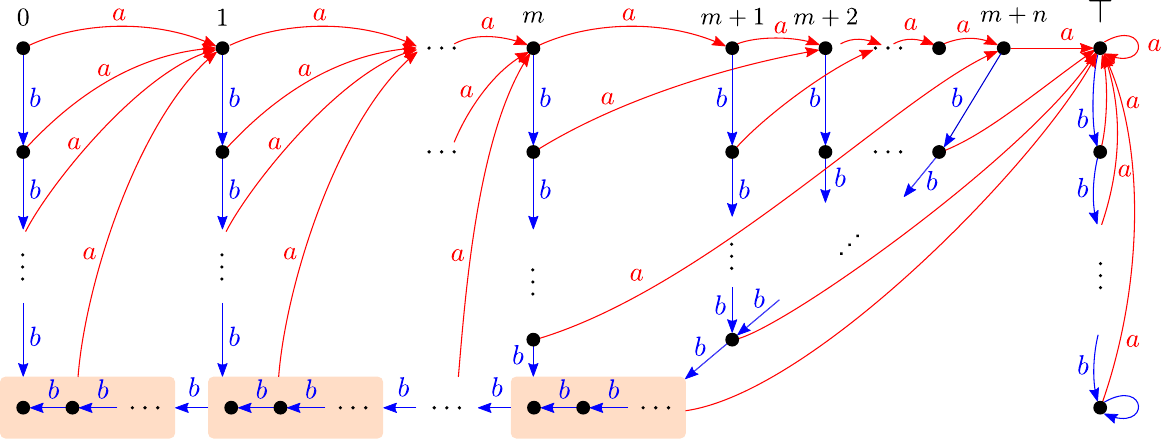}
	\end{center}
	\caption{An "$\eps$-separated" "chromatic@@graph" "well-monotone" graph of "breadth" $n+1$ which is "universal" for $W_3$.}
	\label{fig:constr_chromatic_w3}
\end{figure}

\paragraph*{Objective $W_{4} = \infOften(bb) \; \cup ( \finOften(b) \cap \finOften(aa))$ over $C=\{a,b,c\}$.}
Note that $W_4$ is "prefix-independent".
Figure~\ref{fig:constr_w4} depicts a deterministic "parity automaton"\footnotemark $\A$ of size 3 recognising $W_4$ (it is shown in Appendix~\ref{sec:appendix_bounds_memory} that there is no smaller automaton for~$W_4$); so this yields an upper bound of~$3$ on the "memory" of $W_4$. 
We claim that "memory"~$2$ is actually sufficient.
The "game" depicted on the right witnesses that "Eve" requires "$\eps$-free memory" $\geq 2$: "positional strategies" are "losing@@strategy", but she "wins" by answering $b$ to $b$ and $a$ to $c$.

\footnotetext{\AP We recall that a ""parity automaton"" is an automaton over infinite words with transitions labelled by natural numbers called ""priorities"". A run in the automaton is accepting if the maximum of the priorities produced infinitely often is even.}

\begin{figure}[h]
	\begin{center}
		\includegraphics[width=\linewidth]{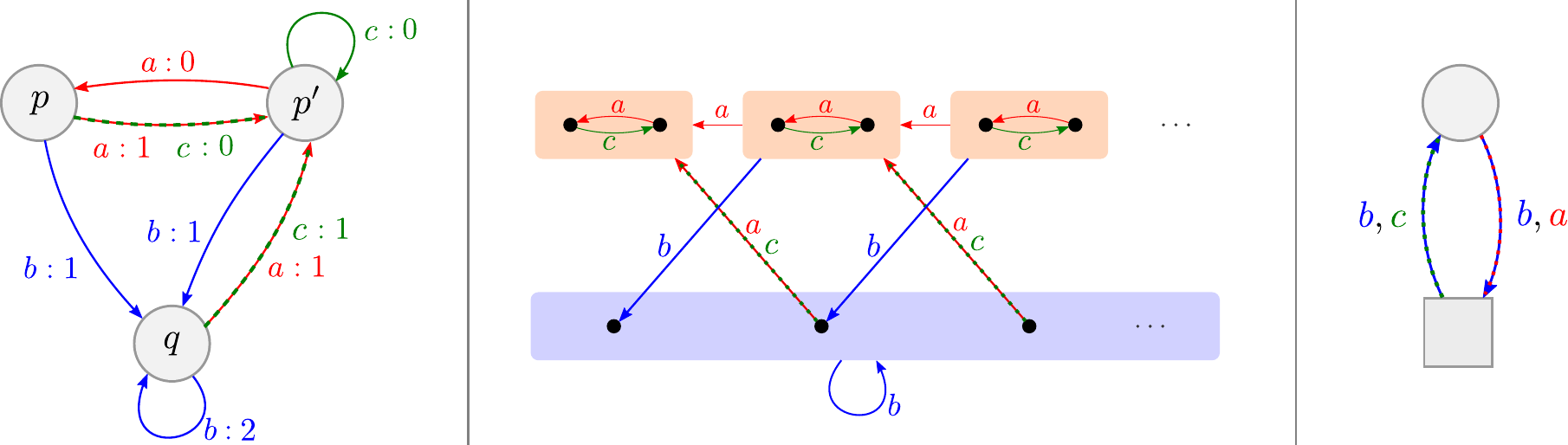}
	\end{center}
	\caption{On the left, a deterministic "parity automaton" $\A$ with three states recognising $W_4$ (we use max-parity semantics).
	In the middle, an "$\eps$-separated" "chromatic@@graph" "universal" "graph" $U$ of "breadth" $2$ for $W_4$; as always, edges following from "monotonicity" are omitted. On the right, a "game" witnessing that "Eve" requires "$\eps$-free memory" $\geq 2$.}\label{fig:constr_w4}
\end{figure}

Consider the "graph" $U$ depicted in the middle of Figure~\ref{fig:constr_w4}; formally it is defined over $V(U) = \{q,p,p'\} \times \kappa$ by the order given by two chains, as in the figure, all $\eps$-edges following the order, and all edges of the form
\begin{enumerate}[(1.)]
\item $(q,\lambda) \re b (q,\lambda')$;
\item $(r,\lambda) \re b (q,\lambda')$ with $r \in \{p,p'\}$ and $\lambda > \lambda'$;
\item $(q,\lambda) \re d r$ with $d \in \{a,c\}, r \in \{p,p'\}$ and $\lambda > \lambda'$;
\item $(r,\lambda) \re d (r',\lambda')$ with $d \in \{a,c\}, r,r' \in \{p,p'\}$ and $\lambda > \lambda'$;
\item $(r,\lambda) \re c (r',\lambda)$ with $r,r' \in \{p,p'\}$; and
\item $(p',\lambda) \re a (p,\lambda)$.
\end{enumerate}
Note that $U$ is "well-monotone", "$\eps$-separated" and "chromatic@@graph" and that it has "breadth" $2$.

To prove that $U$ is "universal" for $W_4$, which implies that $W_4$ has "$\eps$-chromatic memory" $\leq 2$, we proceed as follows.
Take a "tree" $T$ of "cardinality" $< \kappa$ "satisfying@@graph" $W_4$, and label it top-down by $\rho : V(T) \to \{p,p',q\}$ by following a run in the deterministic automaton $\A$, say, starting from state $q$ (this choice does not matter).
Since $T$ "satisfies@@graph" $W$, every branch corresponds to an accepted run, thus on each branch the maximal "priority" appearing infinitely often is even.
To obtain a morphism into $U$, it suffices to append to $\rho(t)$ an ordinal $\lambda \in \kappa$ capturing the number of $1$'s appearing before the next $2$ on paths starting from $t$.

\paragraph*{Objective $W_{5} = \{w\in C^\oo \mid |\minf(w)| = 2\}$.}

In~\cite{Casares22}, Casares uses this "Muller objective" to provide a separation between "chromatic@@memory" and "non-chromatic memory".
By the characterisation of the "$\eps$-memory" of "Muller objectives"~\cite{DJW97} (see also Section~\ref{sec:muller_games} below), we know that the "$\eps$-memory" of $W_5$ is exactly $2$.
However, the size of the alphabet is a lower bound for the "$\eps$-free chromatic memory" (and therefore, also for the "$\eps$-chromatic memory"~\cite{Casares22}).
Figure~\ref{fig:constr_w5} depicts two universal graphs for $W_5$ which give the two upper bounds ("$\eps$-memory" 2 and "$\eps$-chromatic memory" $|C|$).

\begin{figure}[h]
	\begin{center}
		\includegraphics[width=\linewidth]{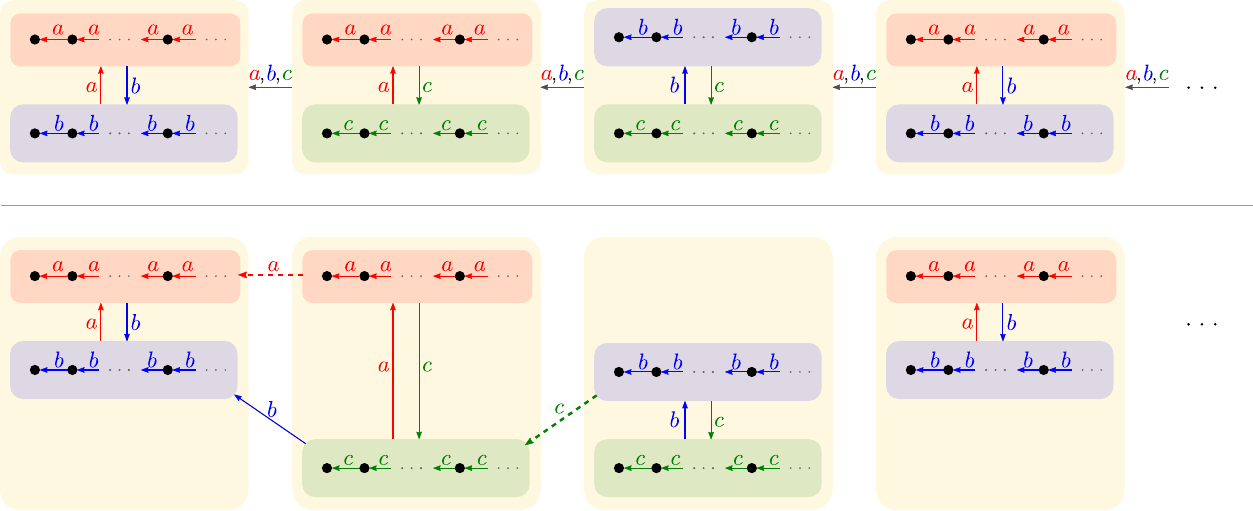}
	\end{center}
	\caption{Two "well-monotone" "universal" graphs for $W_5$ in the case where $C=\{a,b,c\}$.
	The one at the top has "width" $2$ and gives the "$\eps$-memory" upper bound (it coincides with the graph built in Section~\ref{sec:muller_games}) and the one at the bottom is "$\eps$-separated" and "chromatic" and gives the bound on the "$\eps$-chromatic memory". Many edges which follow from "monotonicity" (such as the dashed ones) are omitted.}\label{fig:constr_w5}
\end{figure}

\subsection{Topologically closed objectives}
\label{sec:generalised_safety}

Let $C$ be a set of colours and $L\subseteq C^*$ be a language of finite words. \AP The ""safety objective associated to $L$"" is defined by
\[
	\intro*\Safe{L}= \{ w\in C^\oo \mid w \text{ does not contain any prefix in } L\}.
\]
\AP An "objective" $W$ is ""topologically closed"" if $W=\Safe{L}$ for some $L\subseteq C^*$. (This notation is justified since "objectives" of the form $\Safe{L}$ are exactly the closed subsets of $C^\oo$ for the Cantor topology.)
Colcombet, Fijalkow and Horn~\cite{CFH14} characterised the memory\footnote{Although the authors do not explicitly mention $\eps$-transitions, the lower bound of \cite[Lemma 5]{CFH14} makes implicit use of games with $\eps$-transitions.}
for "topologically closed" "objectives" using the notion of "left quotient". We show next how to recover a variant of their result by applying Theorem~\ref{thm:characterisation_eps}. It has been recently proven~\cite{BFRV23Regular} that given a finite automaton recognising a regular language $L$ and a number $k\in \NN$, it is $\mathtt{NP}$-complete to decide whether the "$\eps$-chromatic memory" of $\Safe{L}$ is $\leq k$ .

Let $W\subseteq C^\oo$ be an "objective" and let $u\in C^*$. \AP We define the ""left quotient"" of $W$ with respect to $u$ by
\[
	\intro*\lquot{u}{W} = \{ w\in C^\oo \mid uw\in W\}.
\] 
\AP We denote $\intro*\Res(W)$ the set of "left quotients" of $W$, and we consider it ordered by inclusion.
We will also write $[u]=\lquot{u}{W}$ for $u\in C^*$, whenever $W$ is clear from the context.
We remark that $[u]\subseteq [v]$ implies $[uc]\subseteq [vc]$ for every $c\in C$.

The following result is a version of \cite[Theorem 6]{CFH14}, but the two statements differ in some slight assumptions\footnote{In \cite{CFH14}, authors only consider finite branching graphs and "objectives" over finite alphabets. Nonetheless, they do not need to suppose that $\ResW$ is well-founded. In this respect, the two results are incomparable.}.

\begin{thm}%[Adapted from {\cite[Theorem 6]{CFH14}}]
	Let $W\subseteq C^\oo$ be a "topologically closed" "objective". Suppose that $(\ResW,\subseteq)$ is "well-founded" of "width" $< \mu$.
	Then $W$ has "$\eps$-free" "memory@@epsFree" $<\mu$.	
	Moreover, if $\mu$ is finite, objective $W$ has "$\eps$-memory exactly" $\mu$.
\end{thm}

\begin{rem}
	As shown in Section~\ref{sec:unbounded_antichains-not-epsilon-memory}, if $\mu$ is infinite, we cannot deduce anything about the "$\eps$-memory" of $W$ by showing "$(\kappa,W)$-universal" graphs of "width" 
	%with "antichains" of size 
	$< \mu$ for $W$. 
\end{rem}

Let $W\subseteq C^\oo$ be a "topologically closed" "objective" such that $(\ResW,\subseteq)$ is "well-founded" of width $<\mu$.
We prove the theorem by giving a construction of a "well-monotone" "$(\kappa,W)$-universal" graph of "width" $<\mu$.
Let $(U,\leq)$ be the partially ordered "graph" given by
\begin{itemize}
	\item $V(U)= \ResW\setminus \{\emptyset\} \cup \{ \top\}$ (where $\top$ is a fresh element).
	\item For $[u],[v]\in \ResW$ we define $[u]\leq [v]$ if $[u]\subseteq [v]$. We let $x\leq \top$ for all $x\in V(U)$.
	\item $[u]\re{c}[v] \in E(U)$ for all $[v]\leq [uc]$. %, if $[uc]\neq \emptyset$. 
	Also, $\top \re{c} x$ for all $x\in V(U)$ and all $c\in  C$.
\end{itemize}

\begin{lem}\label{lem: [u]_satisfies_u-1W}
	 A vertex $[u]\in V(U)\setminus\{\top\}$ "satisfies" the "objective" $\lquot{u}{W}$. In particular, vertex~$[\eps]$ "satisfies" $W$.
\end{lem}

\begin{proof}
	Let $L\subseteq C^*$ be a language such that $W=\Safe{L}$. Let $w\in C^\oo$ be a word labelling an infinite "path" from $[v_0]=[u]$ in $U$:
	\[ [v_0]\re{w_0}[v_1] \re{w_1}[v_2] \re{w_2} \dots .\] We need to show that for any finite prefix $w'$ of $w$, $uw'\notin L$. We remark that this is equivalent to $[uw']\neq \emptyset$. We prove by induction that $[v_i]\leq [uw_0\cdots w_{i-1}]$. By definition of $E(U)$, $[v_i]\leq [v_{i-1}w_{i-1}]$. By induction hypothesis $[v_{i-1}]\leq [uw_0\cdots w_{i-2}]$, so $[v_{i-1}w_{i-1}]\leq [uw_0\cdots w_{i-2}w_{i-1}]$ and by transitivity, $[v_i]\leq [uw_0\cdots w_{i-1}]$. Therefore, for any finite prefix $w'=w_0\dots w_{k-1}$ it holds that $\emptyset < [v_k] \leq [uw_0\dots w_{k-1}]$, which concludes the proof.
\end{proof}

\begin{prop}\label{prop:universal-graph-safety}
	For all "cardinals" $\kappa$, $(U,\leq)$ is a "well-monotone" "$(\kappa,W)$-universal" graph of "width" $<|\mu|$.
\end{prop}
\begin{proof}
	By the hypothesis of well-foundedness and on the size of the antichains of $\ResW$, graph $(U,\leq)$ is "well-founded" and has width $<|\mu|$.
	
	For the "monotonicity", suppose that $x\geq y \re{c} y' \geq x'$. If $x=\top$, then $x\re{c}x'$ by definition. If not, $x=[u]$, $y=[v]$, $y'=[v']$ and $x'=[u']$ for some words $u,v,v',u'$. By definition of $E(U)$, $[v']\leq[vc]$. Since $[v]\leq [u]$ implies $[vc]\leq [uc]$, we deduce by transitivity that $[u']\leq [uc]$ and therefore $u \re{c} u'$.
	
	Finally, we prove the "$(\kappa, W)$-universality" of $(U,\leq)$. Let $T$ be a "$C$-tree"  with "root" $t_0$. If $t_0$ does not "satisfy" $W$, the map $\phi(t)=\top$ for all $t\in V(T)$ is a "morphism" that "preserves the value" at $t_0$. If $t_0$ satisfies $W$, we define a "morphism" $\phi:T \to U$ satisfying that $\phi(T)\subseteq \ResW\setminus \{\emptyset\}$ in a top-down fashion: $\phi(t)=[u]$, for $u\in C^*$ the unique word labelling a "path" from $t_0$ to $t$ in $T$. In particular, $\phi(t_0)=[\eps]$, so by Lemma~\ref{lem: [u]_satisfies_u-1W} $\phi$ "preserves the value" at $t_0$. Finally, we verify that $\phi$ is a "morphism": let $t\re{c} t'$ be an edge in $T$. If $u$ is the word labelling the "path" from $t_0$ to $t$, the word labelling the "path" from $t_0$ to $t'$ is $uc$, so $\phi(t)=[u]$ and $\phi(t')=[uc]$. By definition, $[u]\re{c}[uc]\in E(U)$, so $\phi$ is a "morphism". 
\end{proof}

 %\begin{example}
 %	Objective $W_2 = \Safe{C^*(aa+bb)}$ from Section~\ref{sec:other-examples} is an example of a "topologically closed" "objective". The graph presented in the right-hand side of Figure~\ref{fig:univGraph-muller-a-b} is the graph of "left quotients" obtained by applying the procedure above: $v_0$ corresponds to $\lquot{\eps}{W_2}$, $v_1$ to $\lquot{a}{W_2}$, and $v_2$ to $\lquot{b}{W_2}$.
 %\end{example}

 %\po{C'est pas tout à fait vrai car le graphe en question n'a pas $\top$. C'est OK dans ce cas car il se trouve que $\Safe{C^*(aa+bb)}$ est prefix-increasing, mais c'est pas le cas en général pour $\Safe{L}$. Comme tout ca est un peu chiant à expliquer et pas très intéressant, je vote pour enlever completement cet exemple.}

\subsection{Muller objectives}
\label{sec:muller_games}

Recall that for an infinite word $w\in C^\omega$ we let $\minf(w)=\{ c\in C \mid w_i=c \text{ for infinitely many } i \}$.
A ""Muller objective"" over a finite set of colours $C$ is given by a family $\F \subseteq \powne(C)$ of non-empty subsets of $C$ and defined by
\[
\intro*\Muller{\F} = \{ w \in C^\omega \mid \minf(w) \in \F\}. 
\]
By a slight abuse, we will say that $\F$ is a "Muller objective" over $C$.

The exact "$\eps$-memory" for "Muller objectives" was characterised by Dziembowski, Jurdzi\'nski and Walukiewicz~\cite{DJW97} using the notion of Zielonka trees, introduced by Zielonka to study the positionality of "Muller objectives"~\cite{Zielonka98}. It has been recently shown that the "$\eps$-memory" of a "Muller objectives" also coincides with the minimal size of a good-for-games\footnote{A good-for-games automaton is a non-deterministic automaton for which the non-determinism can be resolved based on the input processed so far.}
 Rabin automaton recognising it~\cite{CCL22SizeGFG}.
Concerning their "$\eps$-free" and "chromatic@@memory" memory, Casares showed~\cite{Casares22} that (1) there are "Muller objectives" whose "$\eps$-free memory" is strictly smaller than its "$\eps$-memory", (2) for all Muller objectives the exact "$\eps$-chromatic memory" and the exact "$\eps$-free chromatic memory" coincide and (3) deciding if the "chromatic memory" of a "Muller objective" is $\leq k$ is $\mathtt{NP}$-complete.

In this section we will focus on the study of the "$\eps$-memory", and we will show how to recover the upper bound presented in~\cite{DJW97} by means of "well-monotone" "universal" graphs.
We now present the necessary definitions to recall their characterisation.

\AP We say that a "Muller objective" $\F \subseteq \powne(C)$ over $C$ is ""positive"" if $C \in \F$, and that it is ""negative"" otherwise.
Given a subset $C' \subseteq C$ of colours, we define the ""restriction@@Muller"" $\intro*\restr \F {C'}$ of $\F$ to $C'$ to be the "Muller objective" over $C'$ given by
\[
\restr \F {C'} = \{F \in \F \mid F \subseteq C'\}.
\]
\AP The ""children"" of a "positive" (resp. "negative") "Muller objective" $\F$ are the "restrictions@@Muller" $\restr{\F} {C'}$ of $\F$ to maximal subsets $C' \subseteq C$ of colours such that $C' \notin \F$ (resp. $C' \in \F$).
"Muller objectives" with no children are called ""basic"", they are exactly those of the form $\F=\emptyset$ or $\F = \pow(C)$ over~$C$.

Note that "children" of a "non-basic" "positive" "Muller objective" are "negative" and vice-versa, and that they are defined over strictly smaller sets of colours.
Observe finally that a "Muller objective" defined over a singleton set of colours is necessarily "basic".
\AP The "$\eps$-memory" of a "Muller objective" can be computed bottom-up from its ""Zielonka tree"": a structure displaying the parenthood relation for the "children" of the condition and all its descendants (defined recursively).
Proposition~\ref{prop: optimalMemoryDJW} details this computation.
For a formal exposition of the Zielonka tree and its uses, see~\cite{DJW97, Horn07PhDThesis, CCL22SizeGFG}.

\begin{propC}[\cite{DJW97}]\label{prop: optimalMemoryDJW}
	\AP Let $\F$ be a "Muller objective".
	The "exact $\eps$-memory" of $\Muller{\F}$ is given by
	\[
	\intro*\memory(\F) = \begin{cases}
	1 & \tif \F \text{ is "basic"},\\
	\sum\limits_{\F' \text{ child of } \F} \memory(\F') &\tif \F \text{ is "positive" "non-basic"}, \\
	\max\limits_{\F' \text{ child of } \F} \memory(\F') &\tif \F \text{ is "negative" "non-basic"}.
	\end{cases}
	\]
\end{propC}

\begin{rem}
	As remarked by Casares~\cite{Casares22}, this characterisation no longer holds for "$\eps$-free" memories or "chromatic@@memory" ones.
\end{rem}

As an example, let $C= \{a,b,c\}$ and consider the "Muller objective" given by
\[
	\F = \{\{a,b\}, \{a,c\}, \{b\}\}.
\]
In Figure~\ref{fig:ZielonkaTree-example} we show the set of colours of the descendants of $\F$ arranged in a "Zielonka tree". 
For this objective, $\memory(\F)=2$.

\begin{figure}[ht]
	\begin{center}
		\includegraphics[width=0.28 \linewidth]{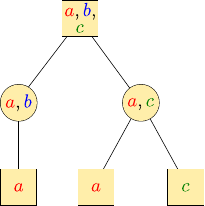}
	\caption{Zielonka tree for 
		$\F=\{ \{a,b\}, \{a,c\}, \{b\}\}$.
		A subtree rooted at a circle (resp. square) node labelled $C'$ corresponds to a "positive" (resp. "negative") "Muller objective" over $C'$.}
	\label{fig:ZielonkaTree-example}
	\end{center}
\end{figure}

The remainder of this section is devoted to obtaining a construction of a "well-monotone" "$(\kappa, \Muller{\F})$-universal@@prefixIndependent" "graph" of width $\leq \memory(\F)$ for a "Muller objective" $\F$ over a finite set of colours $C$ and a "cardinal" $\kappa$.
As always for "prefix-independent" objectives (see Section~\ref{sec:prefix_independent_objectives}), recall that being $\kappa$-universal means satisfying the objective, and embedding all $C$-"pretrees" of cardinality $< \kappa$ whose infinite branches satisfy the objective.

We start with positive and negative "basic" objectives, which are dealt with separately.
\begin{description}
	\item[If $\F$ is "positive" "basic", $\F=\pow(C)$] In this case, the objective is trivially winning:\linebreak$\Muller{\F}=C^\omega$.
	It is easy to see that the graph consisting in just one vertex with a self loop for each colour in $C$ is "well-monotone" "$(\kappa, \Muller{\F})$-universal@@prefixIndependent" and of "width" 1, as required.
	
	\item[If $\F$ is "negative" "basic", $\F=\emptyset$] In this case, the objective is trivially losing: $\Muller{\F}=\emptyset$.
	Let us define $U$ over $V(U) = \kappa$ by
	\[
	x \re c y \in E(U) \iff x > y;
	\]
	it is a well-monotone pregraph with "width" $1$.
Note that graphs satisfying $\Muller{\F}$ are exactly those without infinite paths, which is the case of $U$.
Now any "$C$-pretree" $T$ of cardinality $<\kappa$ without infinite branches can be embedded in $U$ by a "morphism" $\phi$ defined in a bottom-up fashion: if $t \in V(T)$ is a sink then $\phi(t)=0$, and otherwise $\phi(t)= \sup\{\phi(t') \mid t \re{c} t' \in E(T)\}$.
\end{description}

We now assume that $\F$ is "non-basic", with "children" $\F_1, \dots, \F_s$ respectively over $C_1, \dots,$ $C_s \subsetneq C$.
For each $i$, we obtain by induction a "well-monotone" $(\kappa,\Muller {\F_i})$-universal graph $U_i$ with "width" $\leq \memory(\F_i)$.
For convenience, we assume that the $V(U_i)$'s are pairwise disjoint.
\begin{description}	
	\item[If $\F$ is "positive" "non-basic"] We define the desired graph $(U,\leq)$ by putting the $U_i$ parallel to each other, and adding edges in between in a cycling fashion. See the left-hand side of Figure~\ref{fig:universal_graph_muller}.
	
	\begin{figure}[h]
		\begin{center}
		\includegraphics[width=\linewidth]{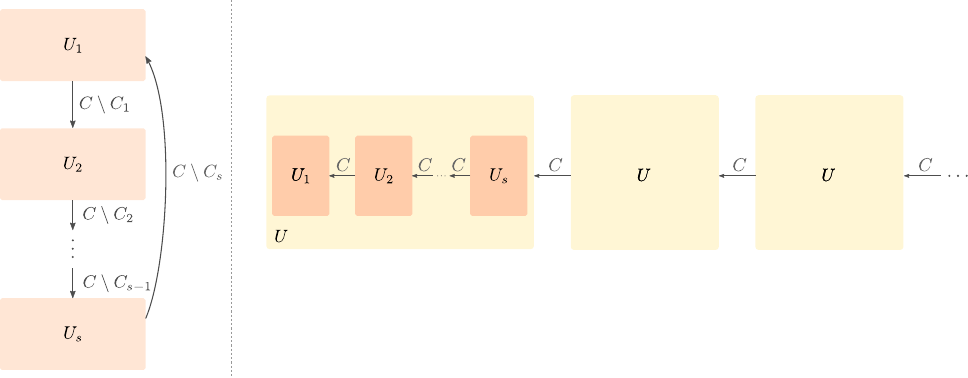}
		\end{center}
		\caption{On the left, the construction for "positive" "Muller objectives" (putting $U_i$'s in parallel); on the right, the construction for "negative" "Muller objectives" (putting $U_i$'s in series).}\label{fig:universal_graph_muller}
	\end{figure}

	Formally, we put $V(U) = \bigcup\limits_{i=1}^s V(U_i)$, and set
	\[
	E(U) = \bigcup_{i=1}^s E(U_i) \cup \bigcup_{i=1}^{s} \{v \re c v' \mid v \in V_i, v' \in V_{i+1} \tand c \notin C_i\}.
	\]
	where it is understood that $s+1$ is identified with $1$.
	The partial order on $U$ is given by
	\[
	v \geq v' \tin U \iff \exists i, [v,v' \in V(U_i) \tand v \geq v' \tin U_i].
	\]
	There remains to prove the following claim.
	
	\begin{clm}
		Graph $(U,\leq)$ is "well-monotone", "$(\kappa, \Muller{\F})$-universal@@prefixIndependent" and has "width" $\leq \memory(\F)=\sum_{i=1}^s \memory{\F_i}$.
	\end{clm}
	
	\begin{claimproof}
		"Well-monotonicity" of $U$ follows directly from "well-monotonicity" of the $U_i$'s.
		The bound its "width" $U$ is also direct.
			
		We now prove that $U$ is "$(\kappa,\Muller{\F})$-universal@@prefixIndependent".
		First, we show that $U$ satisfies $\Muller{\F}$.
		Infinite paths that eventually remain in some $U_i$ satisfy $\Muller{\F_i} \subseteq$ $\Muller{\F}$.
		Colour sequences $w \in C^\omega$ of other infinite paths have infinitely many occurrences of colours not in $C_i$, for each $i$, and therefore $\minf(w)$ is not a subset of any of the $C_i$'s: it has to belong to $\F$ by maximality of the $C_i$'s.
		
		We now let $T$ be a "pretree" of cardinality $<\kappa$ with root $t_0$ satisfying $\Muller{\F}$.
		Let us first label vertices of $T$ by integers from $\{1,\dots,s\}$ in a top-down fashion; these labels $\ell : V(T) \to \{1,\dots,s\}$ will determine in which $U_i$ will the vertices be mapped, and are defined as follows:
		\begin{itemize}
			\item we set $\ell(t_0)=1$, and
			\item for each $t \re c t' \in E(T)$, assuming $\ell(t)$ is defined, if $c \in C_{\ell(t)}$ we let $\ell(t')=\ell(t)$ and otherwise we let $\ell(t')=\ell(t)+1$ (or $1$ if $\ell(t)=s$).
		\end{itemize}
		For each $i$, we let $G_i$ be the pregraph obtained as the "restriction@@graph" of $T$ to $\ell^{-1}(i)$.

		Observe that $G_i$ is a $C_i$-"pregraph" satisfying $\Muller{\F}$, and therefore it satisfies $\Muller{\F_i}$ since $\F_i$ is the "restriction@@Muller" $\restr{\F}{C_i}$ of $\F$ to $C_i$.
		Moreover, $G_i$ is a disjoint union of "pretrees" (as a "restriction@graph" of a "pretree"), each of which is of cardinality $< \kappa$.
		Thus by induction we define, for each $i$, a "morphism" $\phi_i : G_i \to U_i$.
		
		Observe finally that each edge in $E(T)$ either belongs to $E(G_i)$ for some $i$, or is of the form $v \re c v'$ with $\ell(v)=i$, $\ell(v')=i+1$ (or $1$ if $i=s$) and $c \notin C_i$.
		This precisely ensures that the sum of the $\phi_i$'s defines a "morphism" $T \to U$, as required.
	\end{claimproof}

	\item[If $\F$ is "negative" "non-basic"] 
	We will in this case construct an almost universal graph "well-monotone" graph $U$ and conclude thanks to Lemma~\ref{lem:rongeur_de_croute}; it is defined by putting the $U_i$'s in series (see right-hand side of Figure~\ref{fig:universal_graph_muller}).
	Formally, we let $U$ be given by $V(U) = \sum_{i=1}^s V(U_i)$, ordered by
	\[
	v \geq v' \tin U \iff i > i' \tor [i=i' \tand v \geq v' \tin U_i],
	\]
	where $v \in V(U_i)$ and $v' \in V(U_{i'})$, and with edges given by
	\[
	E(U) = \sum_{i=1}^s E(U_i) \cup \{v \re c v' \mid v \in V(U_i) \tand v' \in V(U_{i'}) \text{ with } i > i'\}.
	\]
	We now concentrate on the following claim which, together with Lemma~\ref{lem:rongeur_de_croute}, implies that $U \ltimes \kappa$ is $(\kappa, \Muller \F)$-universal, which concludes the proof.

	\begin{clm}
	The "well-monotone" graph $U$ is almost $(\kappa,\Muller{\F})$-universal and has "antichains" bounded by $\memory(\F)=\max_{i} \memory(\F_i)$.
	\end{clm}

	\begin{claimproof}
		"Well-monotonicity" as well as the bound on the "width" are both a direct proof; we focus on almost universality.
		First, observe that an infinite path in $U$ eventually remains in some $U_i$ and therefore satisfies $\Muller{\F_i} \subseteq \Muller{\F}$.
		
		We now let $T$ be a $C$-pretree of cardinality $< \kappa$ which satisfies $\Muller{\F}$, and aim to show that there is a vertex $t \in V(T)$ such that $\treerooted{T}{t}$, embeds in some $U_i$ (and thus, by composition with the inclusion morphism, in $U$).
		Since such $\treerooted{T}{t}$'s are pretrees of cardinality $< \kappa$ satisfying $\Muller{\F}$, it suffices by induction that for some $t$ and some $i$, all colours appearing in $\treerooted{T}{t}$ belong to $C_i$.
		
		Towards a contradiction, assume otherwise.
		Then we will construct an infinite path 
		\[
		\pi = t_0 \rp {w_0} t_0' \re {c_1} t_1 \rp{w_1} t_1' \re {c_2} \dots
		\]
		in $T$ as follows.
		Assume $\pi$ constructed up to $t_j$ with $j=ks + i$.
		Since $T[t_j]$ contains an edge colour $c_{j+1} \notin C_j$, we let $t_{j}' \re {c_j} t_{j+1}$ be such an edge, and extend $\pi$ with the path $t_j \rp {w_j} t_j'$ followed by the above edge.
		Since there are finitely many colours this implies that $\minf(w)$ is not a subset of any of the $C_i$'s, and therefore $\pi$ does not satisfy $\Muller{\F}$, the wanted contradiction.
	\end{claimproof}
\end{description}

Figure~\ref{fig:example_univgraph_muller} depicts the universal graph obtained using this construction for the Muller condition from Figure~\ref{fig:ZielonkaTree-example}.

\begin{figure}[h]
\begin{center}
\includegraphics[width = \linewidth]{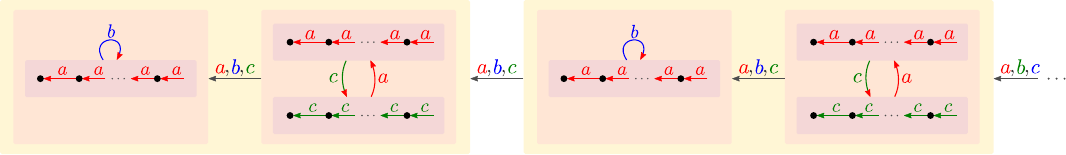}
\end{center}
\caption{The "universal" "graph" obtained for the "Muller objective" given by $\F=\{ \{a,b\}, \{a,c\}, \{b\}\}$ (the corresponding "Zielonka tree" is given in Figure~\ref{fig:ZielonkaTree-example}).} \label{fig:example_univgraph_muller}
\end{figure}

"Objective" $W_1$ from Section~\ref{sec:other-examples} is another simpler example of a "Muller objective". The graph shown in Figure~\ref{fig:univGraph-muller-a-b} coincides with the one obtained by following the above procedure.

Thanks to Proposition~\ref{prop:non_eps-implies-eps} and Theorem~\ref{thm:characterisation_eps}, we conclude with the upper bound in Proposition~\ref{prop: optimalMemoryDJW}; for the lower bound we refer to~\cite{DJW97}.

\section{Counterexamples}
\label{sec:counter-examples}

A few counterexamples which set the limits of our approach are given in this section.

\subsection{No structuration theorem for \texorpdfstring{$\eps$}{ε}-free memory}
%{No converse implication in Theorem~\ref{thm:implication_non_eps}}

%In Section~\ref{sec:main_results} we presented the two main results relating the different types of "memories@@objective" of "objectives" and "well-monotone" "universal" "graphs". Theorem~\ref{thm:characterisation_eps} covers the case of the "$\eps$-memory", whereas Theorem~\ref{thm:implication_non_eps} treats the "$\eps$-free" case. 
%These two results are incomparable, since the exact "memory@@objective" of an "objective" differs in these two cases~\cite{Casares22, Kozachinskiy22InfSeparation} (see also Section~\ref{sec:other-examples}).
%In the $\eps$-case, Theorem~\ref{thm:characterisation_eps} provides an exact characterisation of the "$\eps$-memory" of an "objective": an "objective" has "$\eps$-memory" at most $\mu$ if an only if there is a family of "$\eps$-separated" "well-monotone" "universal" graphs. However, in the case of the "$\eps$-free" "memory@@epsFree", only the implication from left to right has been stated.
In this section, we show that the converse of Theorem~\ref{thm:implication_non_eps} does not hold, even in the case of "objectives".
The counterexample we provide is a generalisation to infinite cardinals of the example proposed by Casares~\cite{Casares22} for showing that the "$\eps$-memory" can be strictly smaller than the "$\eps$-free" one.

\begin{prop}\label{prop:eps-memory-greater-than-memory}
	For each "cardinal" $\mu$ there is a set of colours $C_\mu$ and an "objective" $W_\mu\subseteq C_\mu^\omega$ satisfying:
	\begin{enumerate}
		\item\label{item: eps-free_at_most_2} The "$\eps$-free" "memory@@less" of $W_\mu$ is $\leq 2$.
		\item\label{item: eps-at_least_mu} The "$\eps$-free" "memory@@less" of $\kl{W_\mu^\eps}$ is $\geq \mu$; and therefore the "$\eps$-memory" of $W_\mu$ is $\geq \mu$.
		\item\label{item: univGraph-at_least_kappa} There is $\kappa$ such that any "monotone" "$(\kappa,W_\mu)$-universal" "graph" has "width" $\geq \mu$.%, for $\kappa=\max\{\alephno,\mu\}$.
	\end{enumerate}
\end{prop}

Note that combining the two first items with Proposition~\ref{prop:non_eps-implies-eps}, %and Theorem~\ref{thm:characterisation_eps},
we already obtain that the converse of Theorem~\ref{thm:implication_non_eps} fails. %: for every infinite $\mu$ there is an "objective" $W_\mu\subseteq C_\mu^\omega$ with "$\eps$-free" "memory at most" $2$, but such that for some $\kappa$ there is no "well-monotone" "$(\kappa,W_\mu)$-universal" "graph" with finitely bounded "antichains".
We include the third item, as it slightly strengthens this result, and we provide a direct proof of it.

\begin{proof}[Proof of Proposition~\ref{prop:eps-memory-greater-than-memory}]
	Let $C_\mu= \mu$ and 
	$ W_\mu= \{ w_0w_1 \dots \in C_\mu^\oo \mid \forall i, w_i\neq w_{i+1}\}$.\footnotemark
	\footnotetext{Condition $W'_\mu= \{ w\in C_\mu^\oo \mid \nexists c \in C_\mu, \nexists u\in C_\mu^* \text{ such that } w=uc^\oo \}$ also verifies the desired property.}
	
	\begin{enumerate}
		\item We first prove that the "$\eps$-free" "memory@@epsFree" of $W_\mu$ is $\leq 2$. Let $\G=(G,\VE,v_0,W_\mu)$ be a "game" that is won be "Eve".
		Since $W_\mu$ is "prefix-increasing", we can assume without loss of generality that for all $v\in \Verts{G}$, "Eve" "wins" the game $\G_v=(G,\VE,v,W_\mu)$.
		Therefore we may assume that for each $v\re{c}v'\in \Edges{G}$, if $v'\in \VE$, $v'$ has an outgoing edge not labelled by $c$, and if $v' \in \VA$ then $v'$ has no outgoing edge labelled with $c$.
		
		We define a "strategy" implementing the following idea: for each vertex $v\in \VE$, "Eve" fixes two outgoing edges labelled by $c_{v,1}\neq c_{v,2}$ (if possible). When she has to play from $v$, the strategy just remembers if the colour produced in the preceding action was $c_{v,1}$ (in which case she chooses to play $c_{v,2}$) or not (in which case she can play $c_{v,1}$ safely).
		
		We define this "strategy" formally.  %, we single out some edges: 
		For each $v\in \VE$, if $v$ has at least two outgoing edges labelled with two different colours, we choose two of them: $v\re{c_{v,1}}v'_1$, $v\re{c_{v,2}}v'_2$, $c_{v,1}\neq c_{v,2}$ (if the same colour labels all outgoing edges of $v$, we take $c_{v,1}=c_{v,2}$, $v'_1=v'_2$).
		For $v \in \VA$, we let $c_{v,1}$ be a fresh colour not in $C_\mu$ (so that it is different from all $c \in C_\mu$ in the conditions below). 
		We define $\S=(S,\pi_\S,s_0)$ as follows:
		\begin{itemize}
			\item $V(S) = \Verts{G}\times \{1\} \sqcup \VE \times \{2\}$.
			\item $s_0= (v_0,1)$.
			\item If $v \in \VA$ and $v\re{c}v'\in \Edges{G}$, if $c\neq c_{v',1}$,  $(v,1)\re{c} (v',1)\in E(S)$. If $c= c_{v',1}$,  $(v,1)\re{c} (v',2)\in E(S)$. 
			\item For $v\in \VE$ and $i\in\{1,2\}$, if $c_{v,i}\neq c_{v'_i,1}$ we let $(v,i)\re{c_{v,i}} (v'_i,1)\in E(S)$. If $c_{v,i}= c_{v'_i,1}$ we let $(v,i)\re{c_{v,i}} (v'_i,2)\in V(S)$.
			\item $\pi_\S(v,i)=v$.
		\end{itemize}
		
		As we have supposed that after visiting a $c$-edge "Eve" have some option not labelled with $c$, the above "strategy" only contains "paths" "satisfying" $W_\mu$.
		
		\item  We now prove that the "memory" of $\kl{W_\mu^\eps}$ is\footnote{In fact, if $\mu$ is infinite we will prove that the "memory" of $W_\mu$ is $> \mu$.} $\geq \mu$.
		Consider the "game" $\G=(G,\VE,v_0,W_\mu)$ over $V(G)=\{v_0\} \sqcup \mu$, where $v_0 \notin \mu$ is a fresh element, by $\VE = \{v_0\}$ and
		\[
			E(G) = \{v_0 \re \eps x \mid x \in \mu\} \cup \{x \re y v_0 \mid y \neq x\}.
		\]
		"Eve" "wins" this game using the following strategy: whenever "Adam" picks an edge $x \re y v_0$, she sends him to vertex $y$ (from where he cannot pick colour $y$ again).
		The game $\G$ is depicted in Figure~\ref{fig:counter_example1}.

		\begin{figure}[h]
			\begin{center}
				\includegraphics[width = 0.5 \linewidth]{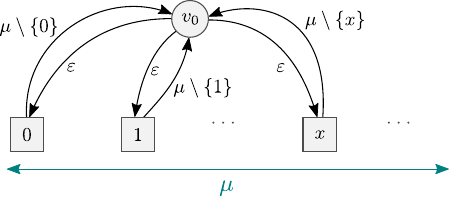}
			\end{center}
			\caption{The game $\G$ in the proof on the second item of Proposition~\ref{prop:eps-memory-greater-than-memory}.}\label{fig:counter_example1}
		\end{figure}
		In this way, the produced sequence does not contain two consecutive colours that are equal, thus "Eve" "wins".
		
		Let $\S=(S,\pi_\S,s_0)$ be an "strategy" such that $|\inv{\pi_\S}(v_0)| < \mu$.
		We will show that $\S$ is not "winning@@strategy".
		For each $s \in \inv{\pi_\S}(v_0)$, choose $x_s \in V(S)$ such that $s \re{\eps} x_s \in E(S)$.
		Let $y \in \mu \setminus \{\pi_\S(x_s) \mid s\in \inv{\pi_\S}(v_0)\}$ be the image under $\pi_\S$ of an element which is never chosen (which exists since $|\inv{\pi_\S}(v_0)| < \mu$).
		The strategy $\S$ contains the following losing "path":
		\[ s_0\re{\eps} x_{s_0} \re{y} s_1\re{\eps} x_{s_1} \re{y} s_2\re{\eps} x_{s_2} \re{y} s_3 \dots  \]

		\item Let $\kappa=\max\{\alephno,\mu\}$. We define a "$C_\mu$-tree" $T$ of cardinality $\kappa$ that cannot be embedded in any "monotone" "$C_\mu$-graph" with "width" $<\mu$ by a "morphism" $\phi$ that "preserves the value" at the root $t_0$. We define $T$ over
		\[
			V(T) = \{w_0\dots w_k \in \mu^* \mid w_i \neq w_{i+1} \text{ for all } 0\leq i < k \},
		\]
		by
		\[
			E(T) = \{ w_0\dots w_k \re{c} w_0\dots w_kc \mid c \in \mu\setminus \{w_k\}\},
		\]
		with root $t_0=\eps$.
		In words, $T$ consists of finite sequences of elements in $\mu$ whose pairwise consecutive elements differ, and the successors of a vertex $t \in V(T)$ are those obtained by adding a colour different from the one in the last position.
		By construction, all paths from $t_0$ satisfy $W_\mu$.
		Moreover, since $\kappa$ is infinite it holds that $|T|=\kappa$ as claimed.
		
		Let $(G,\leq)$ be a "monotone" "$C$-graph" with "antichains" of cardinality $<\mu$ and let $\phi \colon T \to G$ be a "morphism". We show that $\phi(t_0)$ does not satisfy $W_\mu$ in $G$. Consider the set of vertices at the first level of the tree, that is, $V_1=\{t \in V(T) \mid t \in \mu\}$.
		As any antichain of $G$ has cardinality $<\mu$, there are two different elements $t,t'\in V_1 \subseteq \mu$ such that $\phi(t)$ and $\phi(t')$ are comparable; we assume without loss of generality that $\phi(t')\leq \phi(t)$.
		Observe that $t \re{t'} tt' \in E(T)$, and therefore $\phi(t)\re{t'}\phi(tt') \in \Edges{G}$ since $\phi$ is a morphism.
		By monotonicity it follows that $\phi(t') \re{t'} \phi(tt') \in \Edges{G}$.
		We deduce that $G$ contains a "path" from $\phi(t_0)$ starting by 
		\[
			\phi(t_0)\re{t'} \phi(t') \re{t'}\phi(tt') \dots,
		\]
		and thus $\phi(t_0)$ does not satisfy $W_\mu$. \qedhere
	\end{enumerate}	
\end{proof}

\subsection{Universal graphs with antichains of unbounded size do not determine the \texorpdfstring{$\eps$}{ε}-memory}\label{sec:unbounded_antichains-not-epsilon-memory}

As observed by Perles~\cite{Perles63}, Dilworth's Theorem (c.f. Theorem~\ref{prop:Dilworth-theorem}) does not hold if the upper bound on the "width" is infinite.
More precisely, he proved that for any "cardinal" $\kappa$, all antichains of the coordinate-wise order $\kappa \times \kappa$ are finite, but it cannot be decomposed in less than $\kappa$ disjoint "chains".

In this section we show that Proposition~\ref{prop:non_eps-implies-eps} does not hold if the bound on the size of the "antichains" of the graph is not finite: the existence of a "well-monotone" "$(\kappa,\val)$-universal" graph of "width" $< \mu$ does not provide any information on the "$\eps$-memory" of $\vale$ if $\mu$ is infinite (even if $\val$ is an "objective").

\begin{prop}\label{prop:infinite-antichains-no-eps-memory}
	For any infinite "cardinal" $\mu$, there exists an "objective" $W_\mu$ such that
	\begin{itemize}
		\item for all cardinals $\kappa$ there exists a "well-monotone" "$(\kappa,W_\mu)$-universal" graph whose "antichains" have cardinality $<\aleph_0$; and
		\item there is an $\eps$-game with "objective" $W_\mu^\eps$ in which "Eve" cannot reach the "value@@game" with "$\eps$-memory"~$<\mu$.
	\end{itemize}
\end{prop}

The rest of this section is devoted to the proof of Proposition~\ref{prop:infinite-antichains-no-eps-memory}.
Fix an infinite "cardinal"~$\mu$.
Let $C_\mu= \mu \times \mu$ and let $W_\mu$ be the "objective":
\[
	W_\mu = \{ (w,w') \in C_\mu^\oo \mid \nexists i \text{ such that } w_i<w_{i+1} \text{ and } w'_i < w'_{i+1}\}.
\]
In words, Eve wins as long as at each step, one of the two coordinates does not increase.
Clearly, "objective" $W_\mu$ is topologically closed, therefore thanks to Proposition~\ref{prop:universal-graph-safety}, it suffices to study its "left quotients" in order to construct well-monotone universal graph.
We now prove that the "left quotients" of $W_\mu$ form a "well-quasi order".

\begin{lem}\label{lem:W_mu-finite-antichains}
	The partial order $(\Res(W_\mu),\subseteq)$ is a "well-quasi order" (wqo).
\end{lem}

\begin{proof}
	We will prove that $(\Res(W_\mu)\setminus\{\emptyset\},\subseteq)$ is order-isomorphic to $(\mu\times \mu, \leq)$ ordered coordinatewise:
	\[
		(x,y) \leq (x',y') \iff x \leq x' \tand y \leq y',
	\]
	which is well-known to be a "wqo".
	
	First, observe that for $(u,v) = (u_0\dots u_n,v_0\dots v_n) \in C_\mu^*$ such that $\lquot{(u,v)}{W_\mu}\neq \emptyset$, it holds that $\lquot{(u,v)}{W_\mu}=\lquot{(u_n,v_n)}{W_\mu}$ (that is, the last letters determine the left quotient).
	We aim to prove that for all $(x,y),(x',y')\in \mu\times \mu$ it holds that
	\[
		(x,y)\leq (x',y') \; \Longleftrightarrow \; \lquot{(x,y)}{W_\mu} \subseteq \lquot{(x',y')}{W_\mu}.
	\]
	If $(x,y)\leq (x',y')$, then for any $(w,w')=(w_0w_1\dots, w'_0w_1'\dots)\in C_\mu^\oo$, if $(xw,yw')\in W_\mu$, then in particular $x\geq w_0$ or $y\geq w'_0$; and there is not $i\in \oo$ such that $w_i<w_{i+1}$ and $w'_i<w'_{i+1}$. Therefore, $x'\geq w_0$ or $y'\geq w'_0$ and $(x'w,y'w')\in W_\mu$.
	Conversely, if $(x,y)\nleq (x',y')$, we suppose without loss of generality that $x>x'$.
	Then $((x'+1)^\oo,y^\oo)\in \inv{(x,y)} W_\mu$ but $((x'+1)^\oo,y^\oo)\notin \inv{(x',y')} W_\mu$.	
\end{proof}

Applying Proposition~\ref{prop:universal-graph-safety} then yields the first item in Proposition~\ref{prop:infinite-antichains-no-eps-memory}.

We now define an $\eps$-game won by "Eve", but in which she needs to use an $\ee$-strategy with memory at least $\mu$; we start with a formal definition, the intuition is explained below.
We write $\powtwo(\mu \times \mu)$ to denote the set of subsets of $\mu \times \mu$ of size $2$.
For $(x,y)\in \mu\times \mu$ we write $(x,y)^{>} =\{ (x',y') \in \mu\times \mu \mid (x,y)< (x',y')\}$.

Let $\G_\mu=(G,\VE, v_0, W_\mu^\ee)$ be the game defined as follows.
\begin{itemize}
	\item $\Verts{G} = \{v_0\} \cup \mu\times \mu \cup \powtwo(\mu \times \mu)$.
	\item $\VE=\powtwo(\mu \times \mu)$.
	\item $\Edges{G}$ contains the following edges:
	\begin{itemize}
		\item $v_0 \re{(x,y)} (x,y)$ for all $(x,y)\in \mu \times \mu$.
		\item For $(x,y)\in \mu\times \mu$, $(x,y) \re{\ee} A$ for all $A\in \powtwo(\mu \times \mu)$ such that $A\nsubseteq (x,y)^{>}$.
		\item For $A\in \powtwo(\mu \times \mu)$, $A \re{(x,y)} (x,y)$ for $(x,y)\in A$.
%		\item $(x,y) \re {\ee} \choice$.
%		\item $\choice \re{\ee} (x,y)$ for all $(x,y)\in \mu\times \mu$.
	\end{itemize}
\end{itemize}

That is, in the game "Adam" and "Eve" alternate moves as follows: "Eve" picks an element $(x,y)\in \mu \times \mu$ amongst two options (the two elements in some set $A$) and then "Adam" can choose what are going to be the options in Eve's next move, as long as she has at least one non-losing move. The first move is done by "Adam", he can choose the first element of the sequence.
Note the similarity with the gadget from the proof of Section~\ref{sec:finite_memory_to_structure}. %Moreover, at any point "Adam" can decide to go to a special vertex $\choose$, from where Eve has to decide from which "Adam"'s vertex the play will continue.

In the game $\G_\mu$, "Eve" has a strategy ensuring that a "path" satisfying $W_\mu^\ee$ will be produced: whenever "Adam" sends her to a vertex $A\in \powtwo(\mu \times \mu)$, she can choose an element that ``keeps her alive'' (she does not produce an increasing pair). %If "Adam" takes an edge $(x,y)\re{\ee} \choice$, she just has to come back to $(x,y)$ with an $\ee$-transition. 
We are now ready to prove the second item in Proposition~\ref{prop:infinite-antichains-no-eps-memory}.

\begin{lem}
	In the game $\G_\mu=(G,\VE, v_0, W_\mu^\ee)$ "Eve" cannot win using an $\ee$-strategy with memory $<\mu$.
\end{lem}

\begin{proof}
	Let $\S=(S,\pi_\S,s_0)$ be an $\ee$-strategy over a set $M$ such that $|\inv{\pi}_\S(v)|<\mu$ for every $v\in V(G)$.
	We will prove that $\S$ contains a losing path from $s_0=(v_0,m_0)$.
	
	For each $q\in \mu \times \mu$ we pick $m_q\in M$ such that $(v_0,m_0) \re{q} (q,m_q) \in E(S)$, and we let $M'=\{m_q\in M \mid q\in \mu\times \mu\}$.
	We let $p_1=(1,0)$ and $p_2=(0,1)$.
	Observe that for any $q \in \mu \times \mu$, $q\re{\ee} \{p_1,p_2\} \in \Edges{G}$, and therefore	
	for any $q\in \mu\times \mu$, since $\S$ is an $\eps$-strategy, $q \in \VA$, $q \re \ee \{p_1,p_2\} \in E(G)$ and $(q,m_q) \in V(S)$, it holds that $(q,m_q) \re \eps (\{p_1,p_2\},m_q) \in E(S)$.	
	Thus there is a different element for each $m_q$ in the fiber above $\{p_1,p_2\}$, and we deduce that $|M'|<\mu$, and we can find two different elements $q_1,q_2\in \mu \times \mu$ such that $m_{q_1}= m_{q_2}=m$.
	Moreover, since $\mu \times \mu$ cannot be decomposed in less than $\mu$ disjoint chains~\cite{Perles63}, we can pick $q_1$ and $q_2$ incomparable. Pick $t_1,t_2\in \mu\times \mu$ satisfying
	\begin{itemize}
		\item $q_1 \nless t_1$, $q_2<t_1$,
		\item $q_2\nless t_2$, $q_1<t_2$.
	\end{itemize}
	Therefore, when "Adam" plays $\{t_1,t_2\}$ from $q_1$, "Eve" should play $t_1$, and when "Adam" plays it from $q_2$ she should choose $t_2$. However, the strategy $\S$ contains the edge $(\{t_1,t_2\},m) \re{t'} t'$, for both $t'\in \{t_1,t_2\}$, and therefore $\S$ contains infinite paths starting by the following two possibilities
	\[
	\begin{array}{rcl}
	(v_0,m_0) \re{q_1} (q_1,m) \re{\ee} (\{t_1,t_2\},m) \re{t'} (t',m'),\\
	(v_0,m_0) \re{q_2} (q_2,m) \re{\ee} (\{t_1,t_2\},m) \re{t'} (t',m'),
	\end{array}
	\]
	one of which violates the "objective" $W_\mu^\ee$.	
\end{proof}

\subsection{Universal graphs need to embed just trees, not graphs}\label{sec:embedding_trees-vs-graphs}
There is a discrepancy between the notions of "universality" used in Ohlmann's characterisation of "positionality"~\cite{Ohlmann23UnivJournal} (which comes from the work of Colcombet and Fijalkow~\cite{CF19}) and the one introduced in this paper: in Ohlmann's paper~\cite{Ohlmann23UnivJournal}, for a graph $U$ to be "universal" it must embed all \emph{"graphs"} (of a given cardinality) via a "morphism" "preserving the value" of all its vertices. However, in this work this condition is relaxed; we only require that $U$ embeds all \emph{"trees"} (of a given cardinality) via a "morphism" "preserving the value" of its "root". % to be "universal".

In the study of positionality, the two definitions (embedding graphs or trees) can be seen to be equivalent: to embed a graph $G$ in a well-monotone (totally ordered) graph $U$, one may first unfold it, then embed the obtained tree, and then obtain a morphism by considering minimal images for each $v \in V(G)$.
For a formal exposition, see~\cite[Corollary 1.1]{Ohlmann21Thesis}.
Therefore notions from this paper indeed collapse with those from~\cite{Ohlmann23UnivJournal} in the case of "positionality" (that is, "memory@@eps" $\mu=1$).

When $U$ is not "totally ordered", as in this paper, the two notions however differ; ours (embedding trees) is a strict relaxation of the previous one (embedding graphs).
We show in this section that this relaxation is in fact necessary: Lemma~\ref{lem:structuration} (and thus, the converse implication in Theorem~\ref{thm:characterisation_eps}) fails when using the stronger definition of universality.

\begin{prop}\label{prop:embed-tree-not-graphs}
	There exits a "prefix-independent" "objective" $W$ with "$\eps$-memory" $\leq 2$ such that for all $m \geq 1$, there is a graph $G_m$ "satisfying" $W$ such that any "monotone" "graph" "satisfying" $W$ and "embedding" $G_m$ has "width" $\geq m$.
\end{prop}

\begin{proof}
	Let $C=\{a,b\}$ and 
	\[
		W = \{ w\in C^\oo \mid w  \text{ has infinitely many occurrences of both } a \text{ and } b \} = \Muller{\{a,b\}}.
	\]
	For $m \geq 1$ we let $G_m$ be the "$C$-graph" over $V(G)=m=\{0, \dots, m-1\}$ given by
	\[
		E(G_m) = \{i \re a j \mid i < j\} \cup \{j \re b i \mid i<j\}.
	\]
Graph $G_m$ is represented in Figure~\ref{fig:non-orderable-graph}.
Note that it indeed "satisfies" $W$.

\begin{figure}[h]
	\begin{center}
		\includegraphics[width = 0.53 \linewidth]{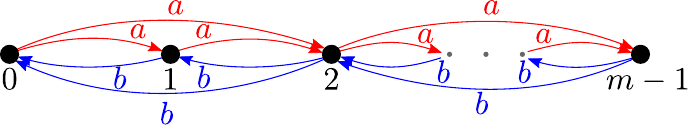}
	\end{center}
	\caption{The graph $G_m$ from the proof of Proposition~\ref{prop:embed-tree-not-graphs}. As often, some edges are omitted for clarity.}\label{fig:non-orderable-graph}
\end{figure}

Let $(U, \leq)$ be a "monotone" "$C$-graph" "satisfying" $W$ with a "morphism" $\phi : G_m \to U$; aiming for a contradiction, we assume that $U$ has width $< m$. 
Since $G_m$ has size $m$, $\phi(V(G_m))$ cannot be an "antichain", thus there are $0 \leq i<j \leq m-1$ such that $\phi(i)$ and $\phi(j)$ are "comparable" in $U$.

Assume first that  $\phi(i) \leq \phi(j)$.
Then we have in $U$ that
\[
	\phi(i) \re a \phi(j) \geq \phi(i)
\]
which gives $\phi(i) \re a \phi(i)$ by "monotonicity".
But then $\phi(i) \re a \phi(i) \re a \dots$ defines a path in $U$ with colouration $a^\omega$, which contradicts that $U$ "satisfies" $W$.
Then case $\phi(i) \geq \phi(j)$ is dealt with symmetrically by constructing a path of colouration $b^\omega \notin W$.	
\end{proof}

\section{Closure properties}
\label{sec:closure_properties}

\subsection{Lexicographical products}
\label{sec:lexicographic_products}

In this section, we prove that "lexicographic products@@objectives" of "objectives" are well-behaved with respect to "memory"; thus extending the result of~\cite{Ohlmann23UnivJournal} about "positionality".
We will only be working with "prefix-independent" "objectives", thus we adopt the definition of "universality for prefix-independent objectives" (see Section~\ref{sec:prefix_independent_objectives}).
\paragraph{Lexicographical products.}\label{sec:lexicographical_products}
We provide a study of "lexicographical products", as introduced by Ohlmann~\cite{Ohlmann23UnivJournal}, whose result we generalize to finite memory bounds.

\AP Given two "prefix-independent" "objectives" $W_1$ and $W_2$ over disjoint sets of colours $C_1$ and $C_2$, we define their ""lexicographical product@@objectives"" $W_1 \ltimes W_2$ over $C= C_1 \sqcup C_2$ by
\[ 
    W_1 \ltimes W_2 = \{w \in C^\omega \mid [w^2 \text{ is infinite and in } W_2] \tor [w^2 \text{ is finite and } w^1 \in W_1]\},
\]
where $w^1$ (resp. $w^2$) is the (finite or infinite) word obtained by restricting $w$ to occurrences of letters from $C_1$ (resp. $C_2$) in the same order.
Note that if $w^2$ is finite then $w^1$ is infinite, which is why the product is well defined.

Note that "lexicographical products@@objectives" are not commutative: informally, more importance is given to $W_2$ and to colours from $C_2$.
They are however associative.

\AP As a well-known example, the ""parity condition""
\[
    \{w \in [0,2h]^\omega \mid \limsup(w) \text{ is even}\},
\]
can be rewritten as a "lexicographical product@@objectives"
\[
    \TW(0) \ltimes \TL(1) \ltimes \TW(2) \ltimes \dots \ltimes \TL(2h-1) \ltimes \TW(2h),
\]
\AP where $\TW(c)$ and $\TL(c)$ are respectively the ""trivially winning"" and ""trivially losing""  "objectives" over $C=\{c\}$, that is
\[
\intro*\TW(c) = \{c^\omega\} \subseteq C^\omega  \tand \intro*\TL(c) = \emptyset \subseteq C^\omega.
\]

\AP Given two "partially ordered sets" $(U_1,\leq_1)$ and $(U_2,\leq_2)$, their ""lexicographical product@@po"" $\leq$ is defined over $U=U_1 \times U_2$ by
\[
    (u_1,u_2) \leq (u_1', u_2') \iff u_2 < u_2' \tor [u_2 = u_2' \tand u_1 \leq u_1'].
\]
If both $\leq_1$ and $\leq_2$ are "well-founded", then so is their "lexicographical product@@po" $\leq$.
The following simple property relates antichains in $\leq_1$ and $\leq_2$ to those in their "product@@po".

\begin{lem}\label{lem:antichains_lexico}
A set $A \subseteq U_1 \times U_2$ defines an "antichain" in $\leq$ if and only if its projection on $U_2$ is an "antichain" with respect to $\leq_2$ and for each fixed $u_2\in U_2$, $\{u_1 \mid (u_1,u_2) \in A\}$ is an "antichain" in $U_1$ with respect to $\leq_1$.
Thus, if $\mu_1$ and $\mu_2$ are upper bounds to the size of "antichains" in $\leq_1$ and $\leq_2$, then $\mu_1 \mu_2$ is an upper bound to the size of "antichains" in $\leq$. 
\end{lem}

\AP We now define the ""lexicographical product@@graphs"" $(U,\leq)$ of two ordered "graphs" $(U_1,\leq_1)$ and $(U_2,\leq_2)$.
Intuitively, each vertex in $U_2$ is replaced by a copy of $U_1$ (see also Figure~\ref{fig:lexico}.

\begin{figure}[h]
	\begin{center}
		\includegraphics[width = \linewidth]{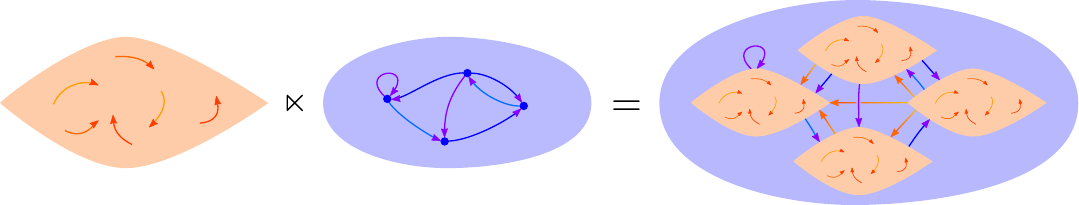}
	\end{center}
	\caption{Illustration of the lexicographical product of two ordered graphs.}\label{fig:lexico}
\end{figure}

Formally $U$ is defined over the "lexicographical product@@po" of $(\Verts{U_1},\leq_1)$ and $(\Verts{U_2},\leq_2)$, that is $V(U)=V(U_1) \times V(U_2)$ and $\leq$ is as above.
Its edges are:
\[
    \begin{array}{lcl}
    E(U) & = & \{(u_1, u_2) \re {c_1} (u_1',u_2') \mid c_1 \in C_1 \tand (u_2 >_2 u_2' \tor [u_2 = u_2' \tand u_1 \re {c_1} u_1' ]) \} \\
    && \cup \quad \{(u_1,u_2) \re {c_2} (u_1',u_2') \mid c_2 \in C_2 \tand u_2 \re {c_2} u_2'\}.
    \end{array}
\]

We denote this product by $U= U_1 \ltimes U_2$; it is very robust with respect to the notions under study.
\begin{lem}
If $(U_1,\leq_1)$ and $(U_2,\leq_2)$ are monotone, then so is their "lexicographical product@@graphs"~$U$.
\end{lem}

\begin{proof}
    Let 
    \[
    (v_1,v_2) \geq (u_1,u_2) \re c (u_1',u_2') \geq (v_1',v_2')
    \]
    in $U$. There are two cases.
    \begin{itemize}
        \item If $c \in C_1$, then again there are two cases.
        \begin{itemize}
            \item If $u_2 > u_2'$, then we have $v_2 \geq u_2 > u_2' \geq v_2'$ which concludes.
            \item Otherwise we have $u_2 = u_2'$ and $u_1 \re c u_1'$.
            If $v_2 > v_2'$ we conclude immediately.
            Otherwise, we have $v_2 = v_2'$ and $v_1 \geq u_1 \re c u_1' \geq v_1'$ and the result follows from monotonicity in $U_1$.
        \end{itemize}
        \item If $c \in C_2$, then by definition, $u_2 \re c u_2'$.
        Since moreover it holds that $v_2 \geq u_2$ and $u_2' \geq v_2'$, we conclude thanks to monotonicity in $U_2$. \qedhere
    \end{itemize}
\end{proof}

We may now state our main result in this section, which is a direct extension of \cite[Theorem~5.2]{Ohlmann23UnivJournal}.

\begin{thm}\label{thm:lexico}
Let $W_1$ and $W_2$ be two "prefix-independent" "objectives" over disjoint sets of colours $C_1$ and $C_2$.
Let $\kappa$ be a "cardinal" and let $(U_1,\leq)$ and $(U_2,\leq)$ be "monotone" graphs which are respectively $(\kappa,W_1)$ and "$(\kappa,W_2)$-universal@@prefixIndependent".
Then $U_1 \ltimes U_2$ is "monotone" and "$(\kappa, W_1 \ltimes W_2)$-universal@@prefixIndependent".
\end{thm}

Combining with Theorems~\ref{thm:characterisation_eps} and~\ref{thm:implication_non_eps} together with Proposition~\ref{prop:non_eps-implies-eps} and Lemma~\ref{lem:antichains_lexico}, we get the following result.

\begin{cor}\label{cor:closure_lexico}
    Let $W_1$ and $W_2$ be two "prefix-independent" "objectives" over disjoint sets of colours $C_1$ and $C_2$, and assume that $W_1$ (resp. $W_2$) has "$\eps$-memory" $\leq n_1 \in \N$ (resp. $\leq n_2$).
    Then, their "lexicographical product@@objectives" $W_1 \ltimes W_2$ has "$\eps$-memory" $\leq n_1 n_2$.
\end{cor}

\paragraph{Products with trivial conditions.}
Before moving on to its proof, we discuss a basic but interesting application of Corollary~\ref{cor:closure_lexico}, namely, that products with "trivial conditions" preserve $\eps$-memory.
Let $W \subseteq C^\omega$ be an objective with finite $\eps$-memory $\leq m$, let $a \notin C$ and denote $C^a=C \sqcup\{a\}$.
Consider the four conditions $W_1,W_2,W_3$ and $W_4$ over $C^a$ defined by
\[
    \begin{array}{lclcl}
        W_1 &=& W \ltimes \TL(a) &=& \{w \in (C^a)^\omega \mid |w_a| < \infty \tand w_{C} \in W\} \\
        W_2 &=& W \ltimes \TW(a) &=& \{w \in (C^a)^\omega \mid |w_a| = \infty \tor w_{C} \in W\}\\
        W_3 &=& \TL(a) \ltimes W &=& \{w \in (C^a)^\omega \mid w_{C} = \infty \tand w_{C} \in W\} \\
        W_4 &=& \TW(a) \ltimes W &=& \{w \in (C^a)^\omega \mid w_{C} < \infty \tor w_{C} \in W\}.
    \end{array}
\]
By Corollary~\ref{cor:closure_lexico}, since $\TL(a)$ and $\TW(a)$ are "positional", each of these four objectives has "$\eps$-memory" $\leq m$.
The first two objectives have sometimes been called respectively $W \land \mathrm{CoBuchi}(a)$ and $W \lor \mathrm{Buchi}(a)$ in the literature, and it was known from the work of Kopcy\'nski~\cite{Kop08Thesis} that these operations preserve "positionality".
However, the stronger result we establish (preservation of "$\eps$-memory") is new, as far as we are aware.

\paragraph{Proof of Theorem~\ref{thm:lexico}.}
The proof of Theorem~\ref{thm:lexico} is similar to that of Ohlmann~\cite[Theorem 5.2]{Ohlmann23UnivJournal}, we give full details for completeness.
The remainder of the section is devoted to the proof.

Fix $W_1,W_2,C_1,C_2,\kappa,(U_1,\leq_1)$ and $(U_2,\leq_2)$ as in the statement of Theorem~\ref{thm:lexico}, and let $C = C_1 \sqcup C_2$, $U=U_1 \ltimes U_2$ and $W=W_1 \ltimes W_2$.
We need to show that $U$ "satisfies" $W$, and that it embeds all "pretrees" of cardinality $< \kappa$ which "satisfy" $W$.

\begin{clm}
The graph $U$ "satisfies" $W$.
\end{clm}

\begin{claimproof}
Consider an infinite path
\[
    \pi = u^0 \re {c^0} u^1 \re {c^1} \dots
\]
in $U$, and for each $i$ let us denote $u^i=(u^i_1,u^i_2)$.
Assume first that there are only finitely many $c^i$'s which belong to $C_2$, and let $i_0$ be such that $c^i \in C_1$ for all $i \geq i_0$.

Then by definition of $C_1$-edges in $U$, it holds that 
\[
    u^{i_0}_2 \geq_2 u^{i_0+1}_2 \geq_2 u^{i_0+2}_2 \geq_2 \dots.
\]
Thus by "well-foundedness" of $\leq_2$, the $u^{i_0+i}_2$ are constant after some point, say for $i \geq i_1$.
Thus it holds that
\[
    u^{i_1}_1 \re{c^{i_1}} u^{i_1+1}_1 \re{c^{i_1+1}} \dots
\]
is a path in $U_1$, and therefore $c^{i_1} c^{i_1+1} \dots \in W_1$.
We conclude in this case that $\pi$ indeed "satisfies" $W$ by "prefix-independence".

Hence we now assume that there are infinitely many $c^i$'s which belong to $C_2$, and let $i_0 < i_1 < \dots$ denote exactly these occurrences.
Then we have for all $j$ that all $c^i$'s with $i \in [i_j +1, i_{j+1}-1]$ belong to $C_1$ and thus by definition of $U$ it holds that $u_2^{i_j+1} \geq u_2^{i_{j+1}}$.
Hence we have in $U_2$
\[
    u_2^{i_0} \re{c^{i_0}} u_2^{i^0+1} \geq_2 u_2^{i^1} \re{c^{i_1}} u_2^{i_1+1} \geq_2 \dots, 
\]
and thus by "monotonicity" of $U_2$,
\[
    u_2^{i_0} \re{c^{i_0}} u_2^{i^1} \re{c^{i_1}} \dots
\]
is a path in $U_2$.
Since $U_2$ "satisfies" $W_2$, we conclude that $\pi$ "satisfies" $W$.
\end{claimproof}

We now show that $U$ embeds all "$C$-pretrees" of cardinality $<\kappa$ which "satisfy" $W$; let $T$ be such a "pretree", and let $t_0$ denote its "root".
Let us partition $V(T)$ according to colours of incoming edges, that is, we let
\[
    \begin{aligned}
    V_2 = \{t_0\} \cup \{t \in V(T) \mid \exists t' \in V(T) \tand c_2 \in C_2 \text{ with } t' \re{c_2} t \in E(T)\} \\ \tand \quad V_1 = \{t \in V(T) \mid \exists t' \in V(T) \tand c_1 \in C_1 \text{ with } t' \re{c_1} t \tin E(T)\}.
    \end{aligned}
\]
Note that indeed we have $V(T) = V_1 \sqcup V_2$.
For each $t \in V(T)$, we moreover define the $V_2$-ancestor of $t$, denoted $\anc(t) \in V_2$, to be the closest ancestor of $t$ belonging to $V_2$, that is, the unique $t' \in V_2$ with a path of $C_1$-edges towards $t$ in $T$; note that for $t \in V_2$ we have $\anc(t)=t$.

We now define a "$C_2$-pretree" $T_2$ rooted at $t_0$ by contracting the $C_1$-edges in $T$, formally we let $V(T_2) = V_2$ and
\[
    E(T_2) = \{t \re {c_2} t' \mid t \rp{C_1^*c_2} t'\}.
\]

\begin{clm}
The $C_2$-"pretree" $T_2$ "satisfies" $W_2$.
\end{clm}

\begin{claimproof}
Let $\pi = t_0 \re {c^0} t_1 \re {c^1} \dots$ be an infinite path in $T_2$; note that the $c^i$'s belong to $C_2$.
Then by definition of $T_2$ we have an infinite path of the form
\[
    \pi':t_0 \re {C_1^* c^0} t_1 \re{C_1^* c^1} \dots
\]
in $T$.
Since $T$ "satisfies" $W$ and $\pi'$ has infinitely many occurrences of colours in $C_2$ (namely, exactly the $c^i$'s), we get that $c^0 c^1 \dots \in W_2$ thus $\pi$ "satisfies" $W_2$. 
\end{claimproof}

Since moreover $|T_2| \leq |T| < \kappa$, there is a "morphism" $\phi_2 : T_2 \to U_2$.
We now partition $V(T)$ according to which element of $U_2$ is assigned to the $2$-ancestor of each vertex, formally for each $u_2 \in U_2$, we define
\[
    V^{u_2} = \{t \in V(T) \mid \phi_2(\anc(t)) = u_2\}.
\]
Note that some $V^{u_2}$'s may be empty, that they partition $V(T)$, and that for each $t \in V_2$ we have $t \in V^{\phi_2(t)}$ since $\anc(t)=t$.

For each $u_2 \in U_2$, we define $T_1^{u_2}$ to be the restriction of $T$ to vertices in $V^{u_2}$ and to $C_1$-edges.
Note that $T_1^{u_2}$ is a disjoint union of "pretrees" (as is any restriction of a "pretree"), and that it has colours in $C_1$.

\begin{clm}
For any $u_2 \in U_2$, it holds that $T_1^{u_2}$ "satisfies" $W_1$.
\end{clm}

\begin{claimproof}
Since $T_1^{u_2}$ is a restriction of $T$, any path in $T_1^{u_2}$ is also a path in $T$; the result follows because $T$ "satisfies" $W$ and $W \cap C_1^\omega \subseteq W_1$ (this is actually an equality).
\end{claimproof}

Since moreover $|T_1^{u_2}| \leq |T| < \kappa$, there exists, for each $u_2 \in U_2$, a "morphism" $\phi_1: T_1^{u_2} \to U_1$.
We are finally ready to define $\phi : V(T) \to U$ to be given by
\[
    \phi(t) = (\phi_1^{u_2}(t), u_2),
\]
where $u_2$ is such that $t \in V_{u_2}$ (that is, $u_2 = \phi_2(\anc(t))$).
The following claim concludes the proof of Theorem~\ref{thm:lexico}.

\begin{clm}
The map $\phi: T \to U$ defines a "morphism".
\end{clm}

\begin{claimproof}
We should check that any edge in $T$ is mapped to an edge in $U$; there are two cases depending on the colour of the edge.
\begin{itemize}
\item Consider an edge $t \re {c_1} t' \in E(T)$ with $c_1 \in C_1$.
Then $t$ and $t'$ have the same $2$-ancestor, and therefore $t \re {c_1} t'$ is an edge in $T_1^{u_2}$ for $u_2=\phi_2(\anc(t))$.
Since $\phi_1^{u_2} : T_1^{u_2} \to U_1$ is a "morphism", it follows that $\phi_1^{u_2}(t) \re{c_1} \phi_1^{u_2}(t') \in E(U_1)$.
Hence by definition of $U$, it indeed holds that $\phi(t) \re {c_1} \phi(t') = (\phi_1^{u_2}(t), u_2) \re {c_1} (\phi_1^{u_2}(t),u_2) \in E(U)$.
\item Consider now an edge $t \re {c_2} t' \in E(T)$ with $c_2 \in C_2$.
Then $t' \in V_2$.
Let $t_0 = \anc(v)$, and observe that $t_0 \re {c_2} t' \in E(T_2)$.
Thus since $\phi_2 : T_2 \to U_2$ is a "morphism", it follows that $\phi_2(t_0) \re {c_2} \phi_2(t') \in E(U_2)$ thus by definition of $U$ we get that $\phi(t) \re {c_2} \phi(t') = (u_1,\phi_2(t_0)) \re {c_2} (u_1',\phi_2(t)) \in E(U)$(regardless of $u_1$ and $u_1'$).
\end{itemize}
This concludes the proof.
\end{claimproof}

\subsection{Combining objectives with locally finite memory}
\label{sec:combinations}

\AP In this section, we investigate properties of "objectives" which have "$\eps$-free" "memory" $< \alephno$.
This means that for any "$\eps$-free" "game" there is an optimal "strategy" $\S$ such that for all vertices $v$, the amount of "memory@@strategy" used at $v$ (that is, the "cardinality" of $\pi_\S^{-1}(v)$) is finite; however it may be that there is no uniform finite bound on the $|\pi_\S^{-1}(v)|$'s, even when the "game" is fixed. 
We call this property ""locally finite memory"".
An example is discussed in Figure~\ref{fig:product_energy}.

We remark that this notion is only interesting in the case of "$\eps$-free memory", and not in that of "$\eps$-memory". By Proposition~\ref{prop:non-strict-ineq-for-eps-memory}, if the "$\eps$-memory" of an objective is $< \alephno$, then it is $\leq n$ for some $n\in \NN$.

\begin{figure}[h]
\begin{center}
\includegraphics[width=0.75\linewidth]{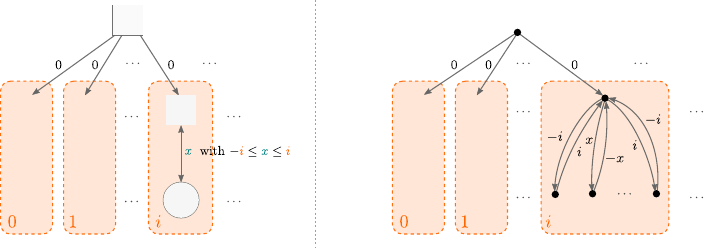}
\end{center}
\caption{A "game" where initially, "Adam" chooses an upper bound $i$, then the players alternate in choosing integers in $[-i,i]$. "Eve" "wins" if the partial sums of the weights remain bounded both from above and below (bi-boundedness objective). She can ensure a win by simply playing the opposite of "Adam" in each round (this "strategy" is represented on the right-hand side), which requires unbounded but locally finite memory.
Since bi-boundedness objectives are intersections of two "positional objectives" (being bounded from above and from below), our results in this section ensure that any "game" with a bi-boundedness objective has optimal "locally finite memory" "strategies".
}\label{fig:product_energy}
\end{figure}

Note that, when applied to $\mu=\alephno$, since "well-founded" "orders" with bounded "antichains" correspond to "well-quasi-orders" ("wqo's"), Theorem~\ref{thm:implication_non_eps} states that the existence of "universal" "monotone" graphs which are "wqo's" for a given "objective" (or even, a "valuation") entails "locally finite memory".
Unfortunately this is not a characterisation: Proposition~\ref{prop:eps-memory-greater-than-memory} applied to $\mu=\alephno $ gives an "objective" with "$\eps$-free memory" $2$ but which does not admit such universal structures.

Still, by combining our knowledge so far with a few additional insights stated below, we may derive some strong closure properties pertaining to this class of "objectives".
\AP In the sequel, we will simply say ""monotone wqo"" for a "well-monotone" graph whose "antichains" are finite.

\AP Given two "partially ordered sets" $(U_1,\leq_1)$ and $(U_2,\leq_2)$, we define their ""(direct) product@@po"" to be the "partially ordered set" $(U_1 \times U_2, \leq)$, where
\[
    (u_1,u_2) \leq (u_1',u_2') \iff [u_1 \leq u_1'] \tand [u_2 \leq u_2'].
\]
Note that if $\leq_1$ and $\leq_2$ are "well-founded", then so is $\leq$.
However, there may be considerable blowup on the size of "antichains", for instance, $\omega \times \omega$ has arbitrarily large antichains whereas $\omega$ is a "total" "order".
However, it is a well-known fact from the theory of "wqo's" (see for instance~\cite{NotesWQO}) that, assuming "well-foundedness", one may not go from finite to infinite "antichains".

\begin{lem}[Folklore]\label{lem:dickson}
If $(U_1,\leq_1)$ and $(U_2,\leq_2)$ are "wqo's", then so is their "product@@po".
\end{lem}

\AP Given two "partially ordered" "$C$-graphs" $(G_1,\leq_1)$ and $(G_2,\leq_2)$, we define their ""(direct) product@@graphs"" to be the "partially ordered" "$C$-graph" $G$ defined over the "product@@po" of $(\Verts{G_1},\leq_1)$ and $(\Verts{G_2},\leq_2)$ by
\[
    E(G) = \{(v_1,v_2) \re c (v_1',v_2') \mid v_1 \re c v_1' \in E(G_1) \tand v_2 \re c v_2' \in E(G_2)\}.
\]
Note that if $(G_1,\leq_1)$ and $(G_2,\leq_2)$ are "monotone", then so is their "product@@graphs".
Therefore, if $(G_1,\leq_1)$ and $(G_2,\leq_2)$ are "monotone wqo's", then so is their product.
Our discussion hinges on the following simple result.

\begin{lem}\label{lem:universality_product}
Let $\kappa$ be a "cardinal", and $W_1,W_2 \subseteq C^\omega$ be two "objectives".
Let $(U_1,\leq_1)$ and $(U_2,\leq_2)$ be two $C$-graphs which are $(\kappa,W_1)$ and "$(\kappa,W_2)$-universal", respectively.
Then their "product@@graphs" $U$ is "$(\kappa, W_1 \cap W_2)$-universal".
\end{lem}

\begin{proof}
Let $T$ be a "tree" with cardinality $< \kappa$, by assumption there exist two "morphisms" $T \re {\phi_1} U_1$ and $T \re {\phi_2} U_2$ which "preserve the value" at the root $t_0$.
We prove that $\phi=(\phi_1,\phi_2) : t \mapsto (\phi_1(t),\phi_2(t)) \in V(U)$ defines a "morphism" from $T$ to $U$ which "preserves the value" at $t_0$.

Let $t \re c t' \in E(T)$, then for both $i \in \{1,2\}$ since $\phi_i$ is a "morphism" it holds that $\phi_i(t) \re c \phi_i(t') \in E(U_i)$ and therefore by definition of $U$, $\phi(t) \re c \phi(t') \in E(U)$ thus $\phi$ is a "morphism".
Now any path from $\phi(t_0)$ in $U$ projects to a path from $\phi_1(t_0)$ in $U_1$, and to a path from $\phi_2(t_0)$ in $U_2$.
Thus since $\phi_1$ and $\phi_2$ "preserve the value" at $t_0$ then so does $\phi$.
\end{proof}

Therefore, by combining Lemma~\ref{lem:dickson} with the above one, we obtain that if two "objectives" $W_1$ and $W_2$ have "monotone wqo's" as "universal" "graphs", then so does their intersection, hence from Theorem~\ref{thm:implication_non_eps}, $W_1 \cap W_2$ has "locally finite memory".
In particular, thanks to Theorem~\ref{thm:characterisation_eps}, we get the following weak closure property.

\begin{cor}\label{cor:intersection-finite-memory}
Let $W_1$ and $W_2$ be two "objectives" which have "monotone wqo's" as "universal" "graphs". Then so does $W_1 \cap W_2$. In particular the intersection of two "objectives" with finite "$\eps$-memory" has "locally finite memory".
\end{cor}

The upper bound stated in the corollary is met: bi-boundedness objectives (see Figure~\ref{fig:product_energy}) give an example where $W_1$ and $W_2$ are "positional" but $W_1 \cap W_2$ does not have finite memory (only "locally finite").
Note moreover that it is not true that the intersection of two "objectives" with finite "$\eps$-memory" has "$\eps$-memory" $<\alephno$ (bi-boundedness objectives are an example).
%Note moreover that the corollary fails (bi-boundedness objectives are an example) if ``"$\eps$-free memory"'' is replaced by ``"$\eps$-memory"'' in the conclusion.
Although our results fall short of implying such a strong closure property, we may still state the following conjecture.

\begin{conj}
"Objectives" with "locally finite memory" are closed under intersection.
\end{conj}

Finally, observe that if an "objective" has "locally finite memory", then it holds that for all finite "games" there is a "strategy" with finite (bounded) "memory@@strategy".
One may wonder if the converse statement is true; unfortunately this is not the case; a counterexample is given by the condition
\[
    W=\left\{w_0w_1 \dots \in \{-1,0,1\}^\omega \mid \exists k \in \N, \sum_{i=0}^{k-1} w_i \leq -1\right\}.
\]
Indeed, one can prove that this objective has finite memory over finite games, however, "Eve" requires (locally) infinite memory to win the game where Adam picks an arbitrary number $i \in \N$ (this is simulated by a chain of $1$-edges), and Eve replies with an arbitrary $j \in - \N$.

\subsection{Unions of prefix-independent \texorpdfstring{$\Sigma_2^0$}{Σ₂⁰} objectives}
\label{sec:unions_sigma2}

As already discussed in Section~\ref{sec:generalised_safety}, the Cantor topology on $C^\oo$ naturally provides a way to define general families of "objectives" that have been well-studied in the literature of formal languages (we refer to~\cite{PP04} for a general overview). In particular, some of these classes of objectives are given by the different levels of the Borel hierarchy; the lowest levels are $\SS_1^0$, consisting on the open subsets, and $\Pi_1^0$, consisting on the closed subsets. 
\AP The level  $\intro{\SS_{n+1}^0}$ (resp. $\Pi_{n+1}^0$) contains the countable unions (resp. countable intersections)  of subsets in $\Pi_{n}^0$ (resp. $\SS_{n}^0$).

In this final section, we prove that "prefix-independent" "objectives" in $\kl{\Sigma_2^0}$ with "$\eps$-memory" $\leq m \in \N$ are closed under countable unions.
This is closely related to Kopczy\'nski's conjecture, which stipulates that "prefix-independent" "positional objectives" are closed under unions; we refer to the conclusion for more discussion.
\AP We recall that $\kl{\Sigma_2^0}$ "objectives" are those of the form
\[
    W_\L = \{w \in C^\omega \mid w \text{ has finitely many prefixes in } \L \},
\]
where $\L \subseteq C^*$ is an arbitrary language of finite words~\cite{Skrzypczak13Topological}.

\begin{thm}
"Prefix-independent" $\kl{\Sigma_2^0}$ "objectives" with "$\eps$-memory" $\leq m \in \N$ are closed under countable unions. 
\end{thm}

This generalises\footnote{Formaly, Ohlmann proved the result for so called ``non-healing'' objectives, which are slightly more general than $\kl{\Sigma_2^0}$. Here we chose to prove it only for $\kl{\Sigma_2^0}$, but the proof is essentially the same, and can easily be adapted to non-healing objectives.} a result of~\cite{Ohlmann23UnivJournal} from "positionality" to finite "memory".

\AP Given a family $(G_\lambda)_{\lambda \in \alpha}$ of "$C$-graphs" indexed by "ordinals", we define their ""direct sum@@graphs"" $G$ to be the disjoint union of the $G_i$, with additionally all $C$-edges pointing from $G_\lambda$ to $G_{\lambda'}$ for $\lambda>\lambda'$; formally
\[
    E(G) = \bigcup_{\lambda \in \alpha} E(G_\lambda) \cup \{v \re c v' \mid v \in V(G_\lambda), v'\in V(G_\lambda') \tand \lambda > \lambda'\}.
\]
If the $G_i$'s are (partially) ordered graphs, then the order on their "sum@@graphs" is defined to be the concatenation of the orders on the $G_i$'s.
Note that if the $G_i$'s are "well-ordered" then so is their "sum@@graphs", and that if the "antichains" of the $G_i$'s are $< \mu$ then so are the "antichains" of their "sum@@graphs".
Recall that $G \ltimes \alpha$ denotes the "direct sum" of $\alpha$ copies of $G$ (which is also the "lexicographical product@@graphs" of $G$ and the "graph" consisting of $\alpha$ vertices and no edges). 

Our proof relies on the following lemma.

\begin{lem}\label{lem:combinations_sigma2}
Let $W_0,W_1,\dots \subseteq C^\omega$ be "prefix-independent" $\kl{\Sigma_2^0}$ objectives, $\kappa$ be a "cardinal", and $U_0,U_1,\dots$ be "$C$-graphs" such that for each $i$, $U_i$ is "$(W_i,\kappa)$-universal".
Let $W=\bigcup_i W_i$.
Then the graph $U \ltimes \kappa$, where $U$ is the "direct sum@@graphs" of the $U_i$'s, is "$(\kappa,W)$-universal".
\end{lem}

\begin{proof}
Thanks to Lemma~\ref{lem:rongeur_de_croute}, it suffices to prove that $U$ is "almost $(\kappa,W)$-universal".
Let $T$ be a "tree" of cardinality $< \kappa$ satisfying $W$.
We show that there exists $i \in \N$ and $t \in T$ such that $T[t]$ "satisfies@@graph" $W_i$; this implies the result since by "universality" of $U_i$ we then get $T[t] \to U_i \to U$.
Assume otherwise.
Take $e=e_0e_1\dots \in \N^\omega$ to be a word over the naturals with infinitely many occurrences of each natural, for instance $e=010120123\dots$.
For each $i \in \N$, let $\L_i \subseteq C^*$ be such that $W_i = \{w \in C^\omega \mid w \text{ has finitely many prefixes in } \L_i\}$.

We now construct an infinite path $\pi=\pi_0 \pi_1 \dots$ starting from the "root" $t_0$ in $T$ such that for each $i$, the coloration $w_0 \dots w_i$ of $\pi_0 \dots \pi_i$ belongs to $\L_{e_i}$.
This implies that the coloration $w$ of $\pi$ has infinitely many prefixes in each of the $\L_i$'s, therefore it does not belong to $W$, a contradiction.
Assume $\pi=\pi_0 \dots \pi_{i-1} : t_0 \rp{w_0 \dots w_{i-1}} t$ constructed up to $\pi_{i-1}$.
Since by assumption, $T[t]$ does not satisfy $W_{e_i}$, there is a path $\pi' : t \rp{w}$ such that $w \notin W_{e_i}$.
By prefix-independence of $W_{e_i}$, we get $w_0 \dots w_{i-1} w \notin W_{e_i}$, thus $w$ has a prefix $w_i$ such that $w_0 \dots w_{i-1} w_i \in \L_{e_i}$; this allows us to augment $\pi$ as required and conclude our proof.
\end{proof}

The theorem follows from combining Lemma~\ref{lem:combinations_sigma2}, Theorem~\ref{thm:characterisation_eps} and Proposition~\ref{prop:non_eps-implies-eps}, and the fact that "antichains" in the "well-founded" "graph" $U \times \kappa$ are no larger than those in $U$.

\section{Conclusion}
\label{sec:conclusion}

In this paper, we have extended Ohlmann's work~\cite{Ohlmann23UnivJournal} to the study of the "memory@@objective" of "objectives".
We have introduced different variants of "well-monotone" "universal" "graphs" adequate to the various models of "memory@@objective" appearing in the literature, and we have characterised the "memory@@objective" of "objectives" through the
existence of such "universal" "graphs" (Theorems~\ref{thm:characterisation_eps} and~\ref{thm:implication_non_eps}).

\paragraph{Possible applications.}
We expect these results to have two types of applications. The first one is helping to find tight bounds for the "memory@@objective" of different families of "objectives". We have illustrated this use of "universal" "graphs" by recovering known results about the "memory@@objective" of "topologically closed objectives"~\cite{CFH14} and "Muller objectives"~\cite{DJW97}, as well as providing non-trivial tight bounds on the "memory@@objective" of some new concrete examples. While finding "universal" "graphs" and proving their correctness might be difficult, we have provided tools to facilitate this task in the important case of "prefix-independent" "objectives" (Lemma~\ref{lem:rongeur_de_croute}). 

The second kind of application discussed in the paper is the study of the combinations of "objectives".
We have used our characterisations to bound the memory requirements of finite "lexicographical product@@objectives" of "objectives" (Section~\ref{sec:lexicographic_products}).
We have also established that intersections of "objectives" with finite "$\eps$-memory" always have "locally finite" "$\eps$-free" "memory@@epsFree".
Finally, we have proved that "prefix-independent" $\kl{\Sigma_2^0}$ objectives with finite "$\eps$-memory" are closed under countable unions.
We believe that the new angle offered by "universal" graphs will help to better understand general properties of "memory@@objective".

\paragraph{Open questions.}
Many questions remain open. First of all, as discussed in Section~\ref{sec:combinations}, we have proved that "objectives" admitting "universal" "monotone wqo's" are closed by intersection. However, we do not know  whether the larger class of "objectives" with unbounded finite "$\eps$-free" memory is closed under intersection. A related question is therefore understanding what are exactly the "objectives" admitting "universal" "monotone wqo's".

In the realm of "positional" "objectives", a long-lasting open question is  Kopczyński's conjecture~\cite{Kop08Thesis}: are unions of "prefix-independent" "positional" "objectives" "positional"? This conjecture has recently been disproved for finite game graphs by Kozachinskiy~\cite{Kozachinskiy22EnergyGroups}, but it remains open for arbitrary game graphs. We propose a generalisation of Kopczyński's conjecture in the case of "$\eps$-memory".

\begin{conj}\label{conj:strong_kop}
	Let $W_1 \subseteq C^\omega$ and $W_2 \subseteq C^\omega$ be two "prefix-independent" "objectives" with "$\eps$-memory" $\leq n_1, n_2$, respectively. Then $W_1\cup W_2$ has "$\eps$-memory" $\leq n_1n_2$.
\end{conj}

%One may verify that Conjecture~\ref{conj:strong_kop} holds for Muller objectives.

Objectives that are $\omega$-regular (those recognised by a deterministic parity automaton, or, equivalently, by a non-deterministic B\"uchi automaton) have received a great deal of attention over the years.
Casares and Ohlmann have recently characterised those $\omega$-regular objectives which are positional~\cite{CO24}, thereby establishing decidability in polynomial time and proving Kopczyński's conjecture for these objectives.
Their characterisation crucially relied on Ohlmann's characterisation of "positionality" via (totally ordered) "well-founded" "monotone" universal graphs.

However, very little is known about "memory" requirements of $\omega$-regular objectives, for instance, the precise "memory" requirement of a given $\omega$-regular objective is not known to be decidable.
We believe that our extension of Ohlmann's universal graphs to the setting of memory paves the way to answering the above question in the positive (possibly, even obtaining a polynomial-time decision procedure\footnote{There is no hope for a polynomial-time decision procedure in the case of chromatic memory, as the problem of deciding whether the chromatic memory of an objective is $\leq k$ is known to be $\mathtt{NP}$-hard already for simple subclasses of $\omega$-regular languages~\cite{Casares22,BFRV23Regular}.}).

Similarly, one may turn to (non-necessarily $\omega$-regular) objectives with topological properties, for instance, it is not known by now which topologically open objectives (or, recognised by infinite deterministic reachability automata) are positional, or finite memory.
We hope that the newly available tools presented in this paper will also help progress in this direction.

\section*{Acknowledgment}
  \noindent The authors wish to acknowledge fruitful discussions with
  Nathanaël Fijalkow, Rémi Morvan and Pierre Vandenhove.

  %% the following bibliography is gererated manually for the sake of brevity
  %% only; please use a separate .bib file in your submission

\bibliographystyle{alphaurl}
\bibliography{bib}
\appendix

\section{Some notes on set theory}
\label{sec:appendix_set_theory}
This appendix collects standard definitions and notations concerning basic set theory, as well as some results used throughout the paper.
In all the paper, the axiom of choice is accepted. 

A reference for all the results stated in this appendix is the book~\cite{Krivine1971SetTheory}.

\subsection{Orders and preorders}
 
\AP A binary relation $\leq$ over a set $A$ is a  ""preorder"" (resp. \emph{strict preorder}) if it is reflexive ($\forall x, x\leq x$) (resp. it is anti-reflexive, $\forall x, x\nleq x$) and transitive ($x\leq y$ and $y\leq z$ implies $x\leq z$). Given a "preorder" $\leq$, we note $<$ the "strict preorder" defined by: $x < y \text{ if } x\leq y \text{ and } x \neq y$.
\AP A "preorder" (resp. \emph{strict preorder}) is an ""order"" (resp. \emph{strict order}) if it is antisymmetric ($x\leq y$ and $y\leq x$ implies $x=y$).
A ""(pre)ordered set"" $(A,\leq)$ is a set together with a "(pre)order relation".

\AP We say that two elements $x, y$ of a preordered set are ""comparable"" if $x\leq y$ or $y\leq x$. A "(pre)order" over $A$ is a ""total (pre)order"" (also called a \emph{linear order}) if any two elements of $A$ are "comparable". If we want to emphasize that an "order relation" is not necessarily total, we may call it a \emph{partial order}.

\AP A ""chain"" of an "ordered set" $(A,\leq)$ is a subset $S\subseteq A$ whose elements are pairwise "comparable".
An ""antichain"" of an "ordered set" $(A,\leq)$ is a subset $S\subseteq A$ whose elements are pairwise incomparable ($\forall x,y \in S$, $x\nleq y$ and $y \nleq x$).

\AP Let $(A, \leq)$ be an "ordered set". A ""maximal (resp. minimal) element"" of $S$ is an element $m\in S$ such that $\forall x\in S$, $m \leq x$ (resp. $x\leq m$) implies $m\leq x$ (resp. $x\leq m$). An element $a\in A$ is a ""supremum (resp. infimum) of $S$"" if $\forall x\in S, x\leq a$ (resp. $\forall x\in S, a\leq x$ and for any other $b\in A$ with this property $a\leq b$ (resp. $b\leq a$). 
"Suprema" and "infima" of "ordered sets" are unique, but they do not necessarily exist.
\AP If the "supremum" (resp. infimum) of a set $S$ belongs to $S$, it is called a ""maximum"" (resp. \emph{minimum}).

\AP A ""lattice"" is an "ordered set" in which all nonempty finite subsets have both a "supremum" and an "infimum". A "complete lattice" is an "ordered set" in which all non-empty subsets have both a "supremum" and an "infimum". We add the adjective \emph{linear} if the "order" is "total".

\AP A "partially preordered set" $(A, \leq)$ is ""well-founded"" if any non-empty subset has a "minimal element"; or equivalently, if it has no infinite strictly decreasing sequence.
A "well-founded"  "(strict) total order" is called a  ""(strict) well-order"".
\AP A "preordered set" (resp. "ordered set") $(A, \leq)$ is a ""well-quasi order"" (wqo) if it is "well-founded" and has no infinite "antichains" (equivalently, if any infinite sequence of elements contains an increasing pair).

\AP Two ordered sets $(A, \leq_1)$, $(B, \leq_2)$ are ""order isomorphic"" if there exists an order preserving bijection between them, that is, a bijection $\phi \colon A \to B$ such that for all $x,y \in A$, $x\leq_1 y$ implies $\phi(x) \leq_2 \phi(y)$ and for all $x,y \in B$ $x\leq_2 y$ implies $\inv\phi(x) \leq_1 \inv\phi(y)$

\begin{prop}[Well-ordering principle]
	Any set admits a "well-ordering".
\end{prop}

\begin{prop}[Dilworth's Theorem~\cite{Dilworth50}]\label{prop:Dilworth-theorem}
	Let $(A,\leq)$ be a "partially ordered set". If the size of the "antichains" of $(A,\leq)$  is bounded by a finite number $k$, there are $k$ disjoint "chains" $S_1,\dots, S_k \subseteq A$, $S_i \cap S_j = \emptyset$ for $i \neq j$, such that $A = \bigcup_{i=1}^k S_i$.
\end{prop}

\subsection{Ordinals and cardinals}

Intuitively, the class of ordinals is defined so that it contains one ordinal for each possible "well-ordered" set, up to isomorphism.

\AP Formally, a set $\alpha$ is an ""ordinal"" if
\begin{enumerate}
	\item The membership relation $\in$ is a "strict well-order" over $\alpha$.
	\item If $x\in \alpha$, then $x\subsetneq \alpha$.
\end{enumerate}

For example, $\emptyset, \{\emptyset\}, \{\emptyset, \{\emptyset\}\}, \{\emptyset, \{\emptyset\}, \{\emptyset, \{\emptyset\}\}\}, \dots$ are ordinals, that we write $0, 1, 2,$ $3, \dots$.
The first infinite ordinal is represented by $\oo = \{0,1,2,3,\dots\}$.

Some important properties of "ordinals" are:
\begin{itemize}
	\item The collection of all "ordinals" is "well-ordered" by the relation of membership. This is the "order" that we will consider over this class.
	\item A "well-ordered set" is "order-isomorphic" to one and only one "ordinal".  
\end{itemize}

\begin{prop}[""Transfinite recursion""]
	Let $P(x)$ be a property about "ordinals". Property $P$ holds for every "ordinal" if and only if it is true that:
	\[ \text{For all ordinal } \aa, \text{ if } P(\bb) \text{ holds for every } \bb < \aa \text{ then } P(\aa) \text{ holds}.\]
\end{prop}

\AP Two sets are said to be ""equinumerous"" if there exits a bijection between them.
The relation of equinumerousity is an equivalence relation (reflexive, symmetric and transitive).
Just as the class of ordinals is defined to contain a representative for any "well-ordered set" up to isomorphism, the class of cardinals is defined to contain one representative for each equivalence class of the equinumerousity relation.

\AP Formally, a ""cardinal"" is defined to be an "ordinal" $\aa$ that is not "equinumerous" to any strictly smaller "ordinal" $\bb < \aa$.
\AP The ""cardinality"" of a set $A$ is the only "cardinal" "equinumerous" to $A$ (equivalently, the smallest "ordinal" "equinumerous" to $A$). 
\AP We denote it by $\intro*\card{A}$.

All finite ordinals are cardinals ($0,1,2,\dots$). The first infinite cardinal is $\oo$. However, when we use it in a context where we are interested in its properties as a "cardinal" and not in its order, we will denote it by $\intro*{\alephno}$.

We remark that "cardinals", as well as "ordinals", are sets. We will often use them to build graphs or other structures and use expressions as ``let $\kappa$ be a "cardinal" and let $x\in \kappa$''.

%Even if they designate the same object, we will use different notations to denote infinite "cardinals" when treating them as an "ordinal" (if we are interested in  

Some important facts about "cardinals" are:
\begin{itemize}
	\item The class of "cardinals" is "well-ordered" by membership. This is the order induced by the class of "ordinals"; in particular we can compare "ordinals" and "cardinals".
	\item  Let $\aa$ be a "cardinal". Its ""successor cardinal"" is the smallest "cardinal" that is strictly greater than $\aa$, it is denoted $\aa^+$.
	\item  The sum of "cardinals" coincides with that of natural numbers over finite cardinals. If $\aa$ and $\bb$ are cardinals and at least one of them is infinite, then $\aa+\bb = \max\{\aa, \bb\}$. In particular, if $\aa$ is infinite, $\aa+1 = \aa$.
	\item  The product of "cardinals" coincides with that of natural numbers over finite cardinals. If $\aa$ and $\bb$ are cardinals and at least one of them is infinite, then $\aa \times \bb = \max\{\aa, \bb\}$.
\end{itemize}

\section{Tight bounds for examples from Section~\ref{sec:examples}}
\label{sec:appendix_bounds_memory}

In this appendix we provide the proofs of the bounds appearing in Table~\ref{table:examples} that we have not included in the main document.

\paragraph*{Objective $W_2 = \{w_0w_1w_2\dots \in C^\oo \mid \forall i\, w_i \neq w_{i+1}\}$.}
\begin{prop}
	The "$\eps$-free" "chromatic memory" of $W_2$ is $\geq |C|$, and therefore, also its "$\eps$-chromatic memory".
\end{prop}
\begin{proof}
	(We suppose $|C|\geq 2$, since the result is trivial for $|C| = 1$.) 
	Let $L = \{u\in C^* \mid \forall i\, u_i \neq u_{i+1}\}\subseteq C^*$ (we remark that $W_2 = \Safe{L^c}$). We consider the "game" $\G=(G,\VE,v_0,W_2)$ given by:
	\begin{itemize}
		\item $\VE = \{v_{c,c'} \mid c,c'\in C, \; c\neq c'\}$,
		\item $\Verts{G} =  \{v_u \mid u\in L\} \sqcup \VE$,
		\item $v_0 = v_\eps$,
		\item $v_u \re{a} v_{ua}$ for all $u\in L$ and all $a\in C$ different from the last colour in $u$.
		\item $v_u \re{a} v_{c,c'}$ for all $u\in L$ and all $a,c,c'\in C$ such that $c\neq c'$,
		\item $v_{c,c'} \re{c} v_{c,c'}$ and $v_{c,c'} \re{c'} v_{c,c'}$ for all $c\neq c'$.
	\end{itemize}

That is, "Adam" starts by picking a finite word $u\in L$ that is safe for the "objective", and he chooses a subset of size $2$ of $C$. Then, "Eve" will have the opportunity to choose between these two colours. It is clear that "Eve" "wins" this "game": no matter Adam's choice, she will have an option to extend the chosen word with a colour different from the last colour of $u$, and then she just has to alternate between the two available colours for the rest of the play.

We now prove that she cannot win with a "chromatic strategy" with "memory@@strategy" $<|C|$.
	
Let $\S = (S,\pi_\S,s_0)$ be a "chromatic@@strategy" "product strategy" over a set $M$, that is, $\Verts{S} \subseteq \Verts{G} \times M$ and there is an "update function" $\delta\colon M\times C \to M$ giving the transitions in the second component of the strategy. Let $s_0 = (v_0, m_0)$.
Suppose that its "memory@@strategy" is $<\card{C}$, that is, for all $v\in \Verts{G}$, $|\inv{\pi_\S}(v)|<\card{C}$.

First, we claim that we can suppose $|M| = |\inv{\pi_\S}(v_{c,c'})|<|C|$. Indeed, without loss of generality we can restrict the "strategy" to the set of vertices $(v, m)$ that are accessible from $s_0 = (v_0,m_0)$ by reading words in $L$. For any $v_u\in \VA$, there is only one such vertex, and for any $v_{c,c'}$, the set of $m\in M$ such that $(v_{c,c'}, m)$ is accessible in that way is independent from the choice of $c$ and $c'$, so we can just suppose that $M$ is the set of such "memory states". %From now on, we will suppose that $M$ consists of such states. 

By the pigeonhole principle, there is some "memory state" $m\in M$ and two different colours $c_1, c_2\in C$ such that there are states $m_1, m_2$ and transitions $\delta(m_1,c_1) = m$ and $\delta(m_2,c_2) = m$. Therefore, in the "strategy" we can find the following two paths:
\[
\begin{array}{rcl}
	(v_0,m_0) \re{u_1} (v_{u_1},m_1) \re{c_1} (v_{c_1,c_2},m), \\
	(v_0,m_0) \re{u_2} (v_{u_2},m_2) \re{c_2} (v_{c_1,c_2},m).
\end{array}
\]
The "strategy" must contain either the edge $(v_{c_1,c_2},m)\re{c_1}$ or the edge $(v_{c_1,c_2},m)\re{c_2}$. In both cases we have found a path in $\S$ that does not "satisfy" the "objective" $W_2$.
\end{proof}

%\begin{prop}
%	The "$\eps$-chromatic memory" of $W_2$ is $\leq |C|$. Therefore, also its "$\eps$-memory" and "$\eps$-free chromatic" memory are $\leq |C|$.
%\end{prop}
%\ac{Maybe we can do this in Section~\ref{sec:examples}}
%\begin{proof}
%	We provide a simple "$\eps$-separated" "chromatic" "well-monotonic" "universal graph" $U$ of "breadth" $|C|$ to obtain this upper bound. This graph is just the straightforward generalisation of the graph from the right of Figure~\ref{fig:univGraph-muller-a-b}. We pick an arbitrary colour $c_1\in C$.
%	
%	\begin{itemize}
%		\item $\Verts{U} = \{v_0\} \sqcup \{v_c\mid c\in C\}$,
%		\item $v_{c_1}< v_0$; all other vertices are incomparable,
%		\item $v_0\re{c_1, \eps} v_{c_1}$, and for all $v\in \Verts{U}$ and all $c\in C$ $v\re{c} v_c$.
%	\end{itemize}
%
%It is evident that this is an "$\eps$-separated" "chromatic" "well-monotonic" graph of "breadth" $|C|$. We claim that it is moreover "$(W_2, \kappa)$-universal" for any "cardinal" $\kappa$. The proof follows the same steps as those in Section~\ref{sec:examples}... \textbf{TODO}
%\end{proof}

\paragraph*{Objective $W_{4} = \infOften(bb) \; \cup ( \neg \infOften(b) \cap \neg \infOften(aa))$ over $C=\{a,b,c\}$.}

\begin{prop}
	A minimal deterministic parity automaton recognising $W_4$ has $3$ states.
\end{prop}

\begin{proof}
	A deterministic parity automaton for $W_4$ with $3$ states was shown in Figure~\ref{fig:constr_w4}.
	
	We prove that a parity automaton with $2$ states cannot recognise $W_4$. Let $\A$ be a deterministic parity automaton with two states $\{q,p\}$.
	We remark that a parity automaton recognising $W_4$ must verify that, from any state $s$, if the run of two words $w, w'\in C^\oo$ from $s$ use the same set of transitions, then $w\in W_4 \iff w'\in W_4$.
	%by prefix-independence of the automaton
	
	We first claim that if $\A$ recognises $W_4$, then its restriction to transitions labelled by $a$ and $c$ must be strongly connected. Indeed, if this was not the case, there would be a state $s$ such that $s\re{a}s$ and $s\re{c}s$, and therefore $\A$ could not differentiate the words $(ac)^\oo\in W_4$ and $(aac)^\oo\notin W_4$ from $s$. Let then $x,y\in \{a,c\}$ be such that $\A$ contains transitions $q\re{x}p$ and $p\re{y}q$.
	
	Now, let us study the structure of the $b$-transitions. There are two possibilities:
	\begin{description}
		\item[(1) $s\re{b}s$ for some $s\in \{q,p\}$]  We suppose $s=q$ w.l.o.g. In this case, $\A$ does not differentiate between $(bxy)^\oo \notin W_4$ and $(bbxy)^\oo \in W_4$ from $q$.
		\item[(2) $q\re{b}p$ and $p\re{b} q$]  In this case, $\A$ does not differentiate between $(byxbxy)^\oo \notin W_4$ and $(bbxy)^\oo \in W_4$.
	\end{description}
	We conclude that $\A$ cannot recognise $W_4$.
\end{proof}

\end{document}

%% file: macros.tex
\newcommand\ac[1]{{\todo[inline,size=\scriptsize,backgroundcolor=YellowGreen]{#1 - \textbf{Antonio}}}}
\newcommand\po[1]{{\todo[inline,size=\scriptsize,backgroundcolor=Pink]{#1 - \textbf{Pierre}}}}

\newcommand\acchanged[1]{{\color{Green}{#1}}}
\newcommand\pochanged[1]{{\color{Pink}{#1}}}

\newenvironment{claimproof}[1]
{\begin{proof}[Proof of the claim]\renewcommand\qedsymbol{$\triangleleft$}\space#1}
{\end{proof}}

%COLOURS
\definecolor{Green2}{HTML}{3EA514}
\definecolor{Red2}{HTML}{FF0400}
\definecolor{Orange2}{HTML}{E6670A}
\definecolor{Violet2}{HTML}{CE1ff9}
\definecolor{Green3}{HTML}{45A229}
\definecolor{Navy}{HTML}{2943A2}

\let\ab\allowbreak
\mathchardef\hyphen=45 %Decimal
%\mathchardef\hyphen="2D %Hexadecimal

\newcommand{\N}{\mathbb N}
\newcommand{\R}{\mathbb R}
\newrobustcmd{\val}{\kl[\val]{\mathrm{val}}}
\knowledge{\val}{notion}
\newrobustcmd{\vale}{\kl[\vale]{\mathrm{val}^{\eps}}}
\knowledge{\vale}[\Weps|W_\mu^\eps]{notion}
\newrobustcmd{\Weps}{\kl[\Weps]{W^{\eps}}}
\newrobustcmd{\valGame}{\kl[\valGame]{\mathrm{val}(\G)}}
\knowledge{\valGame}{notion}
\newrobustcmd{\valStrat}{\kl[\valStrat]{\mathrm{val}(\S)}}
\knowledge{\valStrat}{notion}

\newcommand{\re}[1]{\xrightarrow{#1}}
\newcommand{\rp}[1]{\overset{#1}{\rightsquigarrow}}
\usetikzlibrary{calc,decorations.pathmorphing,shapes}

\newcounter{sarrow}
\newcommand\lrp[1]{%
	\stepcounter{sarrow}%
	\mathrel{\begin{tikzpicture}[baseline= {( $ (current bounding box.south) + (0,-0.5ex) $ )}]
			\node[inner sep=.5ex] (\thesarrow) {$\scriptstyle #1$};
			\path[draw,<-,decorate,
			decoration={zigzag,amplitude=0.7pt,segment length=1.2mm,pre=lineto,pre length=4pt}] 
			(\thesarrow.south east) -- (\thesarrow.south west);
	\end{tikzpicture}}%
}

\newcommand{\Eve}{\mathrm{Eve}}
\newcommand{\Adam}{\mathrm{Adam}}
\newrobustcmd{\VE}{\kl[\VE]{V_\Eve}}
\knowledge{\VE}{notion}
\newrobustcmd{\VA}{\kl[\VA]{V_\Adam}}
\knowledge{\VA}{notion}

\newrobustcmd{\Ceps}{\kl[\Ceps]{C^{\eps}}}
\knowledge{\Ceps}{notion}

\newrobustcmd{\Verts}[1]{\kl[\Verts]{V(#1)}}
\knowledge{\Verts}{notion}
\newrobustcmd{\Edges}[1]{\kl[\Edges]{E(#1)}}
\knowledge{\Edges}{notion}

\newrobustcmd{\treerooted}[2]{\kl[\treerooted]{#1[#2]}}
\knowledge{\treerooted}{notion}

\newcommand{\pow}{\mathcal P}
\newcommand{\powne}{\mathcal P^{\neq \emptyset}}
\newcommand{\powtwo}{\mathcal P^{= 2}}

\newrobustcmd{\card}[1]{\kl[\card]{|#1|}}
\knowledge{\card}{notion}

\newrobustcmd{\prodcard}[2]{\kl[\prodcard]{#1 \times #2}}
\knowledge{\prodcard}{notion}

\newcommand{\T}{\mathcal T}
\renewcommand{\S}{\mathcal S}
\newcommand{\U}{\mathcal U}
\newcommand{\G}{\mathcal G}
\newcommand{\eps}{\varepsilon}
\newcommand{\emptyword}{\epsilon}

\newcommand{\tand}{\text{ and }}
\newcommand{\tor}{\text{ or }}
\newcommand{\tin}{\text{ in }}
\newcommand{\tif}{\text{ if }}
\newcommand{\tow}{\text{ otherwise}}

\newcommand{\dleq}{\leqslant}
\newcommand{\dgeq}{\geqslant}

\newrobustcmd{\Ipath}[2]{\kl[\Ipath]{\Pi^\infty_{v_0}(G)}}
\knowledge{\Ipath}{notion}
\newrobustcmd{\Fpath}[2]{\kl[\Fpath]{\Pi^{\mathrm{fin}}_{v_0}(G)}}
\knowledge{\Fpath}{notion}

\newcommand{\inv}[1]{#1^{-1}}
\newrobustcmd{\lquot}[2]{\kl[\lquot]{#1^{-1}#2}}
\knowledge{\lquot}{notion}

\newrobustcmd{\Res}{\kl[\Res]{\mathrm{Res}}}
\knowledge{\Res}{notion}
\newcommand{\ResW}{\Res(W)}
\newcommand{\choice}{\mathrm{Choice}}
\newcommand{\init}{\mathrm{Init}}
\newrobustcmd{\minf}{\kl[\minf]{\mathrm{Inf}}}
\knowledge{\minf}{notion}

%CALIGRAPHIC
\newcommand{\NN}{\mathbb{N} }
\newcommand{\ZZ}{\mathbb{Z}}
\newcommand{\QQ}{\mathbb{Q}}
\newcommand{\RR}{\mathbb{R}}
\newcommand{\CC}{\mathbb{C}}
\newcommand{\WW}{\mathbb{W}}

\newcommand{\I}{\mathcal{I}}
\newcommand{\F}{\mathcal{F}}
\renewcommand{\H}{\mathcal{H}}
\let\oldL\L
\renewcommand{\L}{\mathcal{L}}
\newcommand{\M}{\mathcal{M}}
\newcommand{\Q}{\mathcal{Q}}
\newcommand{\C}{\mathcal{C}}
\newcommand{\A}{\mathcal{A}}
\newcommand{\B}{\mathcal{B}}
\newcommand{\Z}{\mathcal{Z}}
\newcommand{\W}{\mathcal{W}}
\newcommand{\V}{\mathcal{V}}
\let\oldP\P
\renewcommand{\P}{\mathcal{P}}
\let\oldO\O
\renewcommand{\O}{\mathcal{O}}
\let\oldS\S
\renewcommand{\S}{\mathcal{S}}

%GREEK LETTERS
\newcommand{\kk}{\kappa}
\newcommand{\uu}{\upsilon}
\newcommand{\dd}{\delta}
\let\oldss\ss
\renewcommand{\ss}{\sigma}
\let\oldll\ll
\renewcommand{\ll}{\lambda}
\newcommand{\rr}{\rho}
\let\oldaa\aa
\renewcommand{\aa}{\alpha}
\let\oldtt\tt
\renewcommand{\tt}{\tau}
\newcommand{\bb}{\beta}
\newcommand{\oo}{\omega}
\newcommand{\pp}{\varphi}
\let\oldgg\gg
\renewcommand{\gg}{\gamma}
\newcommand{\ee}{\varepsilon}

\let\oldSS\SS
\renewcommand{\SS}{\Sigma}
\newcommand{\GG}{\Gamma}
\newcommand{\DD}{\Delta}

\newrobustcmd{\alephno}{\kl[\alephno]{\aleph_0}}
\knowledge{\alephno}{notion}

%CONDITIONS
\newrobustcmd{\Safe}[1]{\kl[\Safe]{\mathrm{Safe}(#1)}}
\knowledge{\Safe}{notion}
\newrobustcmd{\Muller}[1]{\kl[\Muller]{\mathrm{Muller}(#1)}}
\knowledge{\Muller}{notion}

\newrobustcmd{\infOften}{\kl[\infOften]{\infty}}
\knowledge{\infOften}{notion}
\newrobustcmd{\finOften}{\kl[\finOften]{\mathrm{Fin}}}
\knowledge{\finOften}{notion}

\knowledge{\Sigma_2^0}[\Sigma_n^0|\Sigma_{n+1}^0|\SS_{n+1}^0]{notion}

%MULLER CONDITIONS
\newrobustcmd{\restr}[2]{\kl[\restr]{#1|_{#2}}}
\knowledge{\restr}{notion}
\newrobustcmd{\memory}{\kl[\memory]{\mathrm{mem}}}
\knowledge{\memory}{notion}

\newrobustcmd{\TW}{\kl[\TW]{\mathrm{TW}}}
\knowledge{\TW}{notion}
\newrobustcmd{\TL}{\kl[\TL]{\mathrm{TL}}}
\knowledge{\TL}{notion}
\newcommand{\anc}{\mathrm{anc}_2}

%% file: knowledge.tex
% !TEX root =  main.tex
\definecolor{Blue Sapphire}{HTML}{005f73} 
\definecolor{Gamboge}{HTML}{ee9b00}
\definecolor{Ruby Red}{HTML}{9b2226}
\definecolor{Blue Marine}{HTML}{022687}
\definecolor{Dark Ruby Red}{HTML}{580507}
\definecolor{Dark Blue Sapphire}{HTML}{053641}

\IfKnowledgePaperModeTF{
}{
	% If we are NOT in paper mode (i.e. in composition mode or electronic mode)
	\knowledgestyle{intro notion}{color={Dark Ruby Red}, emphasize}
	\knowledgestyle{notion}{color={Dark Blue Sapphire}}
	\hypersetup{
		colorlinks=true,
		breaklinks=true,
		linkcolor={}, % Links to sections, pages, etc.
		citecolor={}, % Links to bibliography
		filecolor={Blue Marine}, % Links to local file
		urlcolor={Blue Marine},
	}
}

%%GRAPHS, TREES, UNIVERSALITY
\knowledge{notion}
 | $C$-pregraph
 | pregraph

\knowledge{notion}
 | $C$-graph
 | graph
 | graphs
 | $\Ceps$-graph
 | $C_\mu $-graph
 | $C$-graphs

\knowledge{notion}
 | path
 | paths

\knowledge{notion}
 | $C$-pretree
 | pretree
 | $C$-pretrees
 | $C_2$-pretree
 | pretrees

\knowledge{notion}
 | $C$-tree
 | tree
 | trees
 | $C$-trees
 | $C^\eps $-tree
 | $C_\mu $-tree

\knowledge{notion}
 | subtree rooted at $t$

 \knowledge{notion}
 | root
 | roots

\knowledge{notion}
 | size@graph

\knowledge{notion}
 | restriction@graph

\knowledge{notion}
 | morphism
 | embeds
 | morphisms
 | embedding
 | embed

 \knowledge{notion}
 | isomorphism
 | isomorphisms
 | isomorphic

 \knowledge{notion}
 | unfolding

 \knowledge{notion}
 | monotone
 | monotonicity

 \knowledge{notion}
 | well-monotone
 | well-monotonicity
 | Well-monotonicity

 \knowledge{notion}
 | $(\kappa,\val)$-universal
 | $(|T|,\val)$-universal
 | $(|T|,\vale )$-universal
 | $(\kappa,W)$-universal 
 | $(\kappa, \Muller{\F})$-universal
 | $(\kappa,W_\mu)$-universal
 | universal
 | universality
 | $(\kappa ,W_3)$-universality
 | $(\kappa , W_3)$-universal
 | $(\kappa ,\vale )$-universal
 | ($\kappa ,\val $)-universal
 | $(\kappa ,\vale )$-universality
 | ($\kappa ,\val $)-universality
 | $(\kappa , W)$-universality
 | $(\kappa ,W_2)$-universal
 | $(\kappa , W_1 \cap W_2)$-universal
 | $(\kappa , W_2)$-universal
 | $(\kappa , W_2)$-universality
 | $(\kappa , W_2^\eps )$-universal
 | $(|T|^+,\val)$-universal
 | $(\kappa , \val )$-universal
 | $(\kappa ,W^\eps )$-universal
 | $(\kappa , W)$-universal
 | universal graph
 | $(W_i,\kappa )$-universal

\knowledge{notion}
 | $(\kappa ,W)$-universal for prefix-increasing objectives
 | $(\kappa ,W)$-universal@prefixIncreasing
 | $(\kappa , W_2)$-universal@prefixIncreasing
 | $(\kappa , W_2)$-universality@prefixIncreasing
 | $(\kappa , W_2^\eps )$-universal@prefixIncreasing
 | universality@prefixIncreasing

\knowledge{notion}
 | universality for prefix-independent objectives
 | $(\kappa ,W)$-universal for prefix-independent objectives
 | $(\kappa ,W)$-universal@prefixIndependent
 | $(\kappa ,W)$-universal@prefixIndependent
 | $(\kappa , \Muller {\F })$-universal@prefixIndependent
 | $(\kappa , \Muller {\F })$-universal@prefixIndependent
 | $(\kappa , \Muller {\F })$-universal@prefixIndependent
 | $(\kappa ,\Muller {\F })$-universal@prefixIndependent
 | $(\kappa , \Muller {\F })$-universal@prefixIndependent
 | $(\kappa ,W_2)$-universal@prefixIndependent
 | $(\kappa , W_1 \ltimes W_2)$-universal@prefixIndependent
 | $(\kappa , W_1^\eps )$-universal@prefixIndependent
 | $(\kappa , W_1)$-universal@prefixIndependent
 | universal@prefixIndependent
 | $(\kappa ,W)$-universal (for prefix-independent objectives)
 | universality@prefixIndependent

\knowledge{notion}
 | almost universality
 | almost universal
 | almost $(\kappa ,W)$-universal
 | almost $(\kappa ,W)$-universality

\knowledge{notion}
| width
| Width

 %%VALUATIONS, GAMES, STRATEGIES
\knowledge{notion}
 | $C$-valuation
 | $C$-valuations
 | valuation
 | valuations
 | value

 \knowledge{notion}
 | $C$-projection

\knowledge{notion}
 | value@graph

\knowledge{notion}
 | value@strategy

\knowledge{notion}
 | value@game

 \knowledge{notion}
 | preserves the value
 | preserves all values
 | preserving the value
 | value-preserving
 | preserve the value

\knowledge{notion}
 | objective
 | objectives
 | Objectives
 | Objective
 | \emph{objective}

\knowledge{notion}
 | satisfy
 | satisfies
 | satisfying

\knowledge{notion}
 | satisfies@pregraph
 | satisfy@pregraph
 | satisfy@pregraph
 | satisfying@pregraph
 | satisfying@graph
 | satisfies@graph

\knowledge{notion}
 | G|_X

\knowledge{notion}
 | $C$-game
 | $C$-games
 | game
 | games

\knowledge{notion}
 | Eve

\knowledge{notion}
 | Adam

\knowledge{notion}
 | strategy
 | strategies

 \knowledge{notion}
 | product strategy
 | product strategies
 | product strategies over
 | product strategy over

\knowledge{notion}
 | winning@strategy
 | winning strategy
 | losing@strategy

\knowledge{notion}
 | wins
 | win
 | won
 | victory

\knowledge{notion} 
 | $\S$-projection
 | projection@strategy

 %%MEMORY, CHROMATIC, EPSILONS
 \knowledge{notion}
 | memory@strategy
 | memories@objective

 \knowledge{notion}
 | memory@epsStrategy
 | memory of an $\eps $-strategy
 | $\eps $-memory@strategy

 \knowledge{notion}%Valuation/objective
 | memory $<$
 | $\eps$-free memory $<$
 | memory@strictlyLess
 | memory strictly less than
 | memory@less
 | memory strictly less than
 | memory@objective
 | memory@epsFree
 | memory

 \knowledge{notion}%Valuation/objective
 | memory at least
 | memory@atLeast
 | memory $>n$

 \knowledge{notion}%Valuation/objective
 | memory exactly
 | $\eps $-memory exactly

 \knowledge{notion}
 | product strategy
 | memory structure

 \knowledge{notion}
 | memory states
 | memory state
 | states of memory

\knowledge{notion}
 | chromatic@strategy
 | chromatic strategy
 | chromatic strategies
 | \emph{chromatic}@strategy
 | chromatic memories

 \knowledge{notion}
 | chromatic@memory
 | non-chromatic@memory
 | chromatic memory
 | chromatic memory $<$
 | ``chromaticity''
 | non-chromatic memory

 \knowledge{notion}
 | update function
 | updated
 
 \knowledge{notion}
 | update function@graph

 \knowledge{notion}
 | $\eps$-extension

 \knowledge{notion}
 | $\eps$-game
 | $\eps$-games

 \knowledge{notion}
 | $\eps$-strategy
 | $\eps$-strategies
 | strategy@eps

\knowledge{notion}
 | $\eps$-memory $<$
 | memory@eps
 | $\eps $-memory <
 | $\eps$-memory

\knowledge{notion}
 | exact memory
 | exact $\eps $-memory

 \knowledge{notion}
 | $\eps $-memory $\geq $

\knowledge{notion}
 | $\eps$-chromatic strategy
 | $\eps$-chromatic strategies
 | chromatic@epsStrategy

 \knowledge{notion}
 | $\eps$-chromatic memory $<$
 | $\eps $-chromatic
 | ($\eps $-)chromatic memory
 | $\eps $-chromatic memory
 | chromatic@epsMemory
 | $\eps $ chromatic memory
 
 \knowledge{notion}
 | exact $\eps $-chromatic memory

\knowledge{notion}
 | $\eps$-separated monotone graph
 | $\eps $-separated monotone graph over a set $M$
 | $\eps$-separated
 | $\eps $-separated monotone graph over

\knowledge{notion}
 | $\eps$-free
 | $\ee$-free
 | $\eps $-free chromatic
 | $\eps $-free memory
 | $\eps $-free chromatic memory

\knowledge{notion}
 | chromatic@epsGraph
 | chromatic@graph

\knowledge{notion}
 | breadth

\knowledge{notion}
 | arena-independent memory
 | arena-independent memories

 %%SET THEORY
  
 \knowledge{notion}
 | cardinal
 | cardinals

 \knowledge{notion}
 | successor cardinal

 \knowledge{notion}
 | ordinal
 | ordinals

 \knowledge{notion}
 | well-founded
 | well-foundedness
 
 \knowledge{notion}
 | antichain
 | antichains

\knowledge{notion}
 | preorder
 | (pre)order relation
 | (pre)order
 
\knowledge{notion}
 | orders
 | order
 | partially ordered 
 | order relation
 | partial order
 
\knowledge{notion}
 | ordered set
 | partially ordered set
 | (pre)ordered set
 | ordered sets
 | preordered set
 | partially ordered sets
 | partially preordered set

\knowledge{notion}
 | chains
 | chain

\knowledge{notion}
 | total order
 | total
 | total (pre)order
 | (strict) total order
 | strict preorder
 | totally ordered

\knowledge{notion}
 | comparable
 | incomparable

\knowledge{notion}
| transfinite recursion
| induction@transfinite
| Transfinite recursion

\knowledge{notion}
| well-quasi-orders
| wqo's
| well-quasi order
| wqo
| well-quasi orders

\knowledge{notion}
| maximal (resp. minimal) element
| minimal element

\knowledge{notion}
| supremum (resp. infimum) of $S$
| Suprema
| infima
| supremum
| infimum

\knowledge{notion}
| maximum

\knowledge{notion}
 | lattice
 | complete lattice
 
\knowledge{notion}
| (strict) well-order
| well-orders
| well-ordering
| well-ordered
| strict well-order
| well-ordered set

\knowledge{notion}
 | order isomorphic
 | order-isomorphic
 
\knowledge{notion}
 | equinumerous

\knowledge{notion}
 | cardinality

%%MULLER

\knowledge{notion}
 | Muller objective
 | Muller objectives
 | \Muller(\F)
 
\knowledge{notion}
 | positive
 
\knowledge{notion}
 | negative

\knowledge{notion}
 | child
 | children

\knowledge{notion}
 | basic
 | non-basic
 
\knowledge{notion}
 | restriction@Muller
 | restrictions@Muller

\knowledge{notion}
 | Zielonka tree

%OBJECTIVES
\knowledge{notion}
| prefix-increasing

\knowledge{notion}
| prefix-independent
| prefix-independence
| prefix independence
| Prefix-independent

\knowledge{notion}
| positionality
| positional
| Positional
| Positional objectives
| positional strategies
| positional objectives

\knowledge{notion}
| topologically closed
| topologically closed objectives

\knowledge{notion}
| left quotient
| left quotients
| Left Quotient
| Left Quotients
| Left quotients

\knowledge{notion}
 | safety objective associated to $L$

\knowledge{notion}
 | parity condition
\knowledge{notion}
| parity automaton

\knowledge{notion}
| priorities
| priority

%LEXICOGRAPHICAL PRODUCTS
\knowledge{notion}
| lexicographical products@objectives
| lexicographic products@objectives
| lexicographical product@objectives

\knowledge{notion}
| lexicographical product@po

\knowledge{notion}
| lexicographical product@graphs
| lexicographic product @graphs
| lexicographical products

\knowledge{notion}
| monotone wqo
| monotone wqo's

\knowledge{notion}
| (direct) product@graphs
| product@graphs

\knowledge{notion}
| product@po
| (direct) product@po

\knowledge{notion}
| trivially winning
| trivial conditions

\knowledge{notion}
| trivially losing

\knowledge{notion}
| locally finite memory
| locally finite

\knowledge{notion}
| direct sum@graphs
| sum@graphs
| direct sum

%% file: main.bbl
\newcommand{\etalchar}[1]{$^{#1}$}
\begin{thebibliography}{DFGL{\etalchar{+}}17}

\bibitem[BCRV24]{BCRV24HalfJournal}
Patricia Bouyer, Antonio Casares, Mickael Randour, and Pierre Vandenhove.
\newblock Half-positional objectives recognized by deterministic {B}{\"{u}}chi
  automata.
\newblock {\em Log. Methods Comput. Sci.}, 20(3), 2024.
\newblock \href {https://doi.org/10.46298/LMCS-20(3:19)2024}
  {\path{doi:10.46298/LMCS-20(3:19)2024}}.

\bibitem[BFRV23]{BFRV23Regular}
Patricia Bouyer, Nathana{\"{e}}l Fijalkow, Mickael Randour, and Pierre
  Vandenhove.
\newblock How to play optimally for regular objectives?
\newblock In {\em {ICALP}}, volume 261, pages 118:1--118:18, 2023.
\newblock \href {https://doi.org/10.4230/LIPICS.ICALP.2023.118}
  {\path{doi:10.4230/LIPICS.ICALP.2023.118}}.

\bibitem[BORV23]{BORV23Journal}
Patricia Bouyer, Youssouf Oualhadj, Mickael Randour, and Pierre Vandenhove.
\newblock Arena-independent finite-memory determinacy in stochastic games.
\newblock {\em Log. Methods Comput. Sci.}, 19(4), 2023.
\newblock \href {https://doi.org/10.46298/LMCS-19(4:18)2023}
  {\path{doi:10.46298/LMCS-19(4:18)2023}}.

\bibitem[BRO{\etalchar{+}}22]{BLORV22Journal}
Patricia Bouyer, St{\'{e}}phane~Le Roux, Youssouf Oualhadj, Mickael Randour,
  and Pierre Vandenhove.
\newblock Games where you can play optimally with arena-independent finite
  memory.
\newblock {\em Log. Methods Comput. Sci.}, 18(1), 2022.
\newblock \href {https://doi.org/10.46298/lmcs-18(1:11)2022}
  {\path{doi:10.46298/lmcs-18(1:11)2022}}.

\bibitem[BRV23]{BRV23TheoretiCS}
Patricia Bouyer, Mickael Randour, and Pierre Vandenhove.
\newblock Characterizing omega-regularity through finite-memory determinacy of
  games on infinite graphs.
\newblock {\em TheoretiCS}, 2, 2023.
\newblock \href {https://doi.org/10.46298/THEORETICS.23.1}
  {\path{doi:10.46298/THEORETICS.23.1}}.

\bibitem[Cas22]{Casares22}
Antonio Casares.
\newblock On the minimisation of transition-based {R}abin automata and the
  chromatic memory requirements of {M}uller conditions.
\newblock In {\em CSL}, volume 216, pages 12:1--12:17, 2022.
\newblock \href {https://doi.org/10.4230/LIPIcs.CSL.2022.12}
  {\path{doi:10.4230/LIPIcs.CSL.2022.12}}.

\bibitem[CCFL24]{CCFL24FromMtoP}
Antonio Casares, Thomas Colcombet, Nathana{\"{e}}l Fijalkow, and Karoliina
  Lehtinen.
\newblock From {M}uller to parity and {R}abin automata: {O}ptimal
  transformations preserving (history) determinism.
\newblock {\em TheoretiCS}, 3, 2024.
\newblock \href {https://doi.org/10.46298/THEORETICS.24.12}
  {\path{doi:10.46298/THEORETICS.24.12}}.

\bibitem[CCL22]{CCL22SizeGFG}
Antonio Casares, Thomas Colcombet, and Karoliina Lehtinen.
\newblock On the size of good-for-games {R}abin automata and its link with the
  memory in {M}uller games.
\newblock In {\em ICALP}, volume 229, pages 117:1--117:20, 2022.
\newblock \href {https://doi.org/10.4230/LIPIcs.ICALP.2022.117}
  {\path{doi:10.4230/LIPIcs.ICALP.2022.117}}.

\bibitem[CF18]{CF18}
Thomas Colcombet and Nathana{\"{e}}l Fijalkow.
\newblock Parity games and universal graphs.
\newblock {\em CoRR}, abs/1810.05106, 2018.
\newblock \href {http://arxiv.org/abs/1810.05106} {\path{arXiv:1810.05106}}.

\bibitem[CF19]{CF19}
Thomas Colcombet and Nathana{\"{e}}l Fijalkow.
\newblock Universal graphs and good for games automata: New tools for infinite
  duration games.
\newblock In {\em FoSSaCS}, pages 1--26, 2019.
\newblock \href {https://doi.org/10.1007/978-3-030-17127-8\_1}
  {\path{doi:10.1007/978-3-030-17127-8\_1}}.

\bibitem[CFH14]{CFH14}
Thomas Colcombet, Nathana{\"{e}}l Fijalkow, and Florian Horn.
\newblock Playing safe.
\newblock In {\em FSTTCS}, volume~29, pages 379--390, 2014.
\newblock \href {https://doi.org/10.4230/LIPIcs.FSTTCS.2014.379}
  {\path{doi:10.4230/LIPIcs.FSTTCS.2014.379}}.

\bibitem[CN06]{CN06}
Thomas Colcombet and Damian Niwi{\'{n}}ski.
\newblock On the positional determinacy of edge-labeled games.
\newblock {\em Theor. Comput. Sci.}, 352(1-3):190--196, 2006.
\newblock \href {https://doi.org/10.1016/j.tcs.2005.10.046}
  {\path{doi:10.1016/j.tcs.2005.10.046}}.

\bibitem[CO23]{CO23Memory}
Antonio Casares and Pierre Ohlmann.
\newblock Characterising memory in infinite games.
\newblock In {\em {ICALP}}, volume 261 of {\em LIPIcs}, pages 122:1--122:18,
  2023.
\newblock \href {https://doi.org/10.4230/LIPICS.ICALP.2023.122}
  {\path{doi:10.4230/LIPICS.ICALP.2023.122}}.

\bibitem[CO24]{CO24}
Antonio Casares and Pierre Ohlmann.
\newblock Positional {\(\omega\)}-regular languages.
\newblock In {\em {LICS}}, pages 21:1--21:14. {ACM}, 2024.
\newblock \href {https://doi.org/10.1145/3661814.3662087}
  {\path{doi:10.1145/3661814.3662087}}.

\bibitem[DFGL{\etalchar{+}}17]{NotesWQO}
St{\'e}phane Demri, Alain Finkel, Jean Goubault-Larrecq, Sylvain Schmitz, and
  Philippe Schnoebelen.
\newblock Well-quasi-orders for algorithms. {L}ecture notes, {M}aster {MPRI},
  2017.
\newblock URL:
  \url{https://wikimpri.dptinfo.ens-cachan.fr/lib/exe/fetch.php?media=cours:upload:poly-2-9-1v02oct2017.pdf}.

\bibitem[Dil50]{Dilworth50}
Robert~P. Dilworth.
\newblock A decomposition theorem for partially ordered sets.
\newblock {\em Annals of Mathematics}, 51(1):161--166, 1950.
\newblock \href {https://doi.org/10.2307/1969503} {\path{doi:10.2307/1969503}}.

\bibitem[DJW97]{DJW97}
Stefan Dziembowski, Marcin Jurdzinski, and Igor Walukiewicz.
\newblock How much memory is needed to win infinite games?
\newblock In {\em LICS}, pages 99--110. {IEEE} Computer Society, 1997.
\newblock \href {https://doi.org/10.1109/LICS.1997.614939}
  {\path{doi:10.1109/LICS.1997.614939}}.

\bibitem[GH82]{GH82}
Yuri Gurevich and Leo Harrington.
\newblock Trees, automata, and games.
\newblock In {\em STOC}, page 60–65, 1982.
\newblock \href {https://doi.org/10.1145/800070.802177}
  {\path{doi:10.1145/800070.802177}}.

\bibitem[GZ05]{GZ05}
Hugo Gimbert and Wieslaw Zielonka.
\newblock Games where you can play optimally without any memory.
\newblock In {\em CONCUR}, volume 3653 of {\em Lecture Notes in Computer
  Science}, pages 428--442. Springer, 2005.
\newblock \href {https://doi.org/10.1007/11539452\_33}
  {\path{doi:10.1007/11539452\_33}}.

\bibitem[Hor08]{Horn07PhDThesis}
Florian Horn.
\newblock {\em {Random Games}}.
\newblock PhD thesis, Universit{\'e} {D}enis {D}iderot - {P}aris 7 {\&}
  {R}heinisch-{W}estf{\"a}lische Technische Hochschule {A}achen, 2008.

\bibitem[Kop08]{Kop08Thesis}
Eryk Kopczy{\'n}ski.
\newblock {\em Half-positional Determinacy of Infinite Games}.
\newblock PhD thesis, Warsaw University, 2008.

\bibitem[Koz22a]{Kozachinskiy22EnergyGroups}
Alexander Kozachinskiy.
\newblock Energy games over totally ordered groups.
\newblock {\em CoRR}, abs/2205.04508, 2022.
\newblock \href {https://doi.org/10.48550/arXiv.2205.04508}
  {\path{doi:10.48550/arXiv.2205.04508}}.

\bibitem[Koz22b]{Kozachinskiy22InfSeparation}
Alexander Kozachinskiy.
\newblock Infinite separation between general and chromatic memory.
\newblock {\em CoRR}, abs/2208.02691, 2022.
\newblock \href {https://doi.org/10.48550/arXiv.2208.02691}
  {\path{doi:10.48550/arXiv.2208.02691}}.

\bibitem[Koz22c]{Kozachinskiy22ChromaticMem}
Alexander Kozachinskiy.
\newblock State complexity of chromatic memory in infinite-duration games.
\newblock {\em CoRR}, abs/2201.09297, 2022.
\newblock \href {http://arxiv.org/abs/2201.09297} {\path{arXiv:2201.09297}}.

\bibitem[Kri71]{Krivine1971SetTheory}
Jean{-}Louis Krivine.
\newblock {\em Introduction to Axiomatic Set Theory}.
\newblock Dordrecht, Netherland: Springer, 1971.

\bibitem[Ohl21]{Ohlmann21Thesis}
Pierre Ohlmann.
\newblock {\em Monotonic graphs for parity and mean-payoff games}.
\newblock PhD thesis, Universit{\'e} de Paris, 2021.

\bibitem[Ohl23]{Ohlmann23UnivJournal}
Pierre Ohlmann.
\newblock Characterizing positionality in games of infinite duration over
  infinite graphs.
\newblock {\em TheoretiCS}, 2, 2023.
\newblock \href {https://doi.org/10.46298/THEORETICS.23.3}
  {\path{doi:10.46298/THEORETICS.23.3}}.

\bibitem[Per63]{Perles63}
Micha~A. Perles.
\newblock On {D}ilworth’s theorem in the infinite case.
\newblock {\em Israel Journal of Mathematics}, 1(1):108--109, 1963.
\newblock \href {https://doi.org/10.1007/BF02759806}
  {\path{doi:10.1007/BF02759806}}.

\bibitem[PP04]{PP04}
Dominique Perrin and Jean{-}{\'{E}}ric Pin.
\newblock {\em Infinite words - automata, semigroups, logic and games}, volume
  141 of {\em Pure and applied mathematics series}.
\newblock Elsevier Morgan Kaufmann, 2004.

\bibitem[Skr13]{Skrzypczak13Topological}
Michał Skrzypczak.
\newblock Topological extension of parity automata.
\newblock {\em Information and Computation}, 228-229:16--27, 2013.
\newblock \href {https://doi.org/10.1016/j.ic.2013.06.004}
  {\path{doi:10.1016/j.ic.2013.06.004}}.

\bibitem[Zie98]{Zielonka98}
Wies{\l}aw Zielonka.
\newblock Infinite games on finitely coloured graphs with applications to
  automata on infinite trees.
\newblock {\em Theoretical Computer Science}, 200(1-2):135--183, 1998.

\end{thebibliography}
